\documentclass[%
reprint,
amsmath,amssymb,
aps,
pra,
floatfix,
twocolumn,
superscriptaddress,
]{revtex4-2}
\usepackage[utf8]{inputenc}
\usepackage[T1]{fontenc}
\usepackage[english]{babel}
\usepackage{enumerate}
\usepackage{float}
\usepackage{amsfonts}
\usepackage{booktabs}
\usepackage{stackengine}
\usepackage[normal]{subfigure}
\usepackage{floatflt}
\usepackage{listings}
\usepackage{graphicx}
\usepackage{amsmath}
\usepackage{csquotes}
\usepackage{hyperref}\hypersetup{colorlinks=true, allcolors=blue}
\usepackage{enumitem}
\usepackage{multirow}
\usepackage{soul}
\usepackage{mathtools}
\usepackage{orcidlink}
\usepackage{amsthm}
\usepackage{amssymb}
\usepackage{physics}
\usepackage{wrapfig}
\usepackage{dsfont}
\usepackage{comment}
\usepackage[title]{appendix}
\usepackage{xcolor}
\usepackage{pict2e}
\usepackage{relsize}
\usepackage{tikz}
\usetikzlibrary{positioning}

\usepackage{tgtermes}
\usepackage{newtxmath}

\DeclareMathAlphabet{\mathbb}{U}{msb}{m}{n}
\DeclareMathAlphabet{\mathfrak}{U}{euf}{m}{n}

\providecommand*{\unit}[1]{\,\ifmmode
\mathrm{\,#1}\else\textup{#1}\fi}
\newcommand{\Md}[1]{M_{#1}(\mathbb{C})}
      
\newcommand{\normt}[1]{\norm{#1}_1}
\newcommand{\bs}[1]{\boldsymbol{#1}}
\newcommand{\mc}[1]{\mathcal{#1}}
\newcommand{\wt}[1]{\widetilde{#1}}

\newcommand{\ave}[3]{\left\langle{#1}\middle|{#2}\middle|{#3}\right\rangle}

\newcommand{\psiplus}[1]{\psi_{+}^{\mathsmaller{(#1)}}}
\newcommand{\pplus}[1]{P_+^{\mathsmaller{(#1)}}}

\newcommand{\orcidauthorFB}{0000-0002-0712-2057} 
\newcommand{\orcidauthorGN}{0009-0006-3232-1222} 

\makeatletter
\newcommand{\bigcomp}{%
	\DOTSB
	\mathop{\vphantom{\sum}\mathpalette\bigcomp@\relax}%
	\slimits@
}
\newcommand{\bigcomp@}[2]{%
	\begingroup\m@th
	\sbox\z@{$#1\sum$}%
	\setlength{\unitlength}{0.9\dimexpr\ht\z@+\dp\z@}%
	\vcenter{\hbox{%
			\begin{picture}(1,1)
				\bigcomp@linethickness{#1}
				\put(0.5,0.5){\circle{1}}
			\end{picture}%
	}}%
	\endgroup
}
\newcommand{\bigcomp@linethickness}[1]{%
	\linethickness{%
		\ifx#1\displaystyle 2\fontdimen8\textfont\else
		\ifx#1\textstyle 1.65\fontdimen8\textfont\else
		\ifx#1\scriptstyle 1.65\fontdimen8\scriptfont\else
		1.65\fontdimen8\scriptscriptfont\fi\fi\fi 3
	}%
}

\newtheorem{definition}{Definition}
\newtheorem{proposition}{Proposition}

\newtheorem{remark}{Remark}

\theoremstyle{plain}
\newtheorem{example}{Example}

\theoremstyle{plain}
\newtheorem{corollary}{Corollary}

\begin{document}
	\title{\textbf{Quantum Dynamical Entropy and non-Markovianity: a collisional model perspective}}
	\author{Giovanni Nichele\,\orcidlink{\orcidauthorGN}}
	\email{gnichele@umk.pl}
	\affiliation{Institute of Physics, Faculty of Physics, Astronomy and Informatics, Nicolaus Copernicus University,  Grudziądzka 5/7, 87-100 Toruń, Poland}
	\affiliation{Dipartimento di Fisica, Università degli Studi di Trieste, I-34151 Trieste, Italy}
	\affiliation{Istituto Nazionale di Fisica Nucleare, Sezione di Trieste, I-34151 Trieste, Italy}
	\author{Fabio Benatti\,\orcidlink{\orcidauthorFB}}
	\affiliation{Dipartimento di Fisica, Università degli Studi di Trieste, I-34151 Trieste, Italy}
	\affiliation{Istituto Nazionale di Fisica Nucleare, Sezione di Trieste, I-34151 Trieste, Italy}
	
	\date{\today} 

\begin{abstract}
Accessing the physical mechanisms behind non-Markovian phenomena in open quantum dynamics  requires the study of the full multi-time statistical properties of the joint system-environment dynamics. This is impossible  
at the level of the reduced dynamics of the open system alone as the latter is obtained by suitably eliminating the environment.
The task is instead made possible by considering multi-time correlation functions involving observables of the open system, only:
the open system-environment interactions turn them into  global ones thus building up correlations between the two systems.
Multi-time correlations  form the basis of both the theory of quantum stochastic processes and of the Alicki-Lindblad-Fannes dynamical entropy (ALF entropy for short). This latter quantity provides for quantum systems a measure of the dynamical entropy production as the Kolmogorov-Sinai entropy does for 
classical systems.  Indeed, in the following we propose the open-system ALF entropy as a measure of non-Markovianity  signalling deviations from
the so-called Quantum Regression (QR) regime. In turn, the operational interpretation of the open-system ALF entropy makes it a suitable witness of back-flow of information from the environment to the open system.
In the case of a collisional model whereby the dissipative dynamics of a finite-level system is obtained by its coupling to an infinite classical spin chain, 
the ALF entropy can be explicitly computed. It turns out to depend on the parameters characterizing the statistical properties of the environment and can be  related to the activation and super-activation of memory effects in the open quantum system.
\end{abstract}

\maketitle
\section{Introduction}
Due to their importance in quantum information technologies, the non-Markovian features of open quantum dynamics have recently become the focus of much research and have been addressed by several, non-equivalent, techniques resulting in an intricate hierarchy~\cite{Lili}. 
As an instance, typically the time-evolution of an open quantum system is provided by a family $\{\Lambda^\ddag_t\}_{t\geq 0}$ of completely positive and trace preserving (CPTP) linear maps $\rho\longmapsto\Lambda^\ddag_t[\rho]$ on the space of states $\rho$ that result from tracing over the environment degrees of freedom. Then, in the absence of further knowledge of the system–environment state, one may identify Markovianity with the completely positive divisibility (CP-divisibility) of the evolution.~\cite{ChrusReview22,RHPmeasure,RivasHuelga}. Namely, with the complete positivity of intermediate-time propagators $\Lambda^\ddag_{t,s}$, $t\ge s\ge0$, composing as 
$\Lambda^\ddag_t=\Lambda^\ddag_{t,s}\circ\Lambda^\ddag_s$. In such a scenario, the lack of positivity-preservation and thus of contractivity of the propagators (lack of P-divisibility) gives rise to the increase, instead of the decrease, of several information distances, a phenomenon generally interpreted as a physical back-flow of information  from the environment into the open system~\cite{BLP}. 
Considering instead a classical discrete time scenario as a benchmark for possible quantum extensions of the notion of Markovianity, there one deals with multi-time probabilities $\pi_{[0,n-1]}$ of a stochastic variable $X$ having, say, discrete values $x_j$ at subsequent times:
\begin{gather*}
	 p_{\bs{x}_{[0,n-1]}}^{[0,n-1]}:=P(X_0=x_0, \cdots, X_{n-1}=x_{n-1})\,,\\
	\pi_{[0,n-1]}=\left\{p_{\bs{x}_{[0,n-1]}}^{[0,n-1]}\right\}_{\bs{x}_{[0,n-1]}}\,.
\end{gather*}
Markovianity is then identified with lack of memory and mathematically characterized by  multi-time conditional probabilities depending only on the immediately preceding outcome:
\begin{equation}
	p_{\bs{x}_{[0,n-1]}}^{(n-1\vert n-2, \cdots ,0)}:=\frac{p_{\bs{x}_{[0,n-1]}}^{[0,n-1]}}{p_{\bs{x}_{[0,n-2]}}^{[0,n-2]}}=p_{x_{n-1}  \, x_{n-2}}^{(n-1\vert n-2)}\,,\label{MCSP0}
\end{equation}
where $\bs{x}_{[0,n-1]}:= x_0 \cdots x_{n-1} $.
Evidently, for classical Markovian processes, multi-time probabilities factorize into products of conditional probabilities,
\begin{equation}
	\label{MCSP0b}
	p_{\bs{x}_{[0,n-1]}}^{[0,n-1]}=\prod_{j=1}^{n-1} p_{x_j x_{j-1}}^{(j\vert j-1)}\ p_{x_0}^{(0)}\,,
\end{equation}
and satisfy the Chapman-Kolmogorov equations; namely, for all $j\geq 	\ell \geq k$,
\begin{equation}
		\label{MCSP01}
		p_{x_j x_k}^{(j\vert k)}=\sum_{x_\ell} p_{x_j x_\ell}^{(j\vert\ell)}\ p_{x_\ell x_k}^{(\ell\vert k)}\ ,
\end{equation}
where
$\displaystyle p_{x_j x_k}^{(j\vert k)}=\frac{\sum_{x_i\neq x_j, x_k}p_{\bs{x}_{[0,j]}}^{[0,j]}}{\sum_{x_i\neq x_k}{p_{\bs{x}_{[0,k]}}^{[0,k]}}}\, .$
Therefore, the transition matrices $T(j,k)=\left[p_{x_j x_k}^{(j\vert k)}\right] 
$,  are stochastic matrices composing as $T(j,k)=T(j,\ell)T(\ell, k)$ and play 
the same role as the discrete-time  quantum propagators $\Lambda^\ddag_{j,k}, j\ge k$.
However, this composition law does not imply classical Markovianity as embodied by~\eqref{MCSP0}, which is stronger. Instead, it  is equivalently identified by the so-called \textit{Classical Regression} (CR). This latter property holds if multi-time correlation functions $\left\langle F^{(0)} F^{(1)}\cdots F^{(n-1)}\right\rangle_\pi
$ are expressible in terms of transition probabilities, as follows:
\begin{align}
	\nonumber
	&\left\langle F^{(0)} F^{(1)}\cdots F^{(n-1)}\right\rangle_\pi=\!\sum_{\bs{a}_{[0,n-1]}}  p_{\bs{x}_{[0,n-1]}}^{[0,n-1]}\  F_{x_0}^{(0)} F_{x_1}^{(1)}\ldots F_{x_n}^{(n-1)}
	\\&=
	\sum_{\bs{x}_{[0,n-1]}}\prod_{j=2}^n\ T_{x_j x_{j-1}}(j,j-1)\ F_{x_0}^{(0)} F_{x_1}^{(1)}\ldots F_{x_n}^{(n-1)}\,,\label{CRF}
\end{align}
for all functions $F^{(j)}$ of the random variable $X$.
In the quantum realm, in order to properly discuss non-Markovian issues, two are the main reasons for choosing the non-commutative analog of~\eqref{CRF}, namely \textit{Quantum Regression} (QR), as a starting point, instead of the reduced dynamics~\cite{Lax,ChrusReview22,LonigroChrusc1,LonigroChrusc2}:
the first one is that the latter in practice involves $2$-time correlation functions losing precious information coming from
the multi-time statistics. 
The second reason is that multi-time correlations provide a powerful tool to access the dynamics of information exactly when a given 
open quantum dynamics at discrete time $n$, $\Lambda^\ddag_n$, is derived from the global reversible
dynamics $\Theta^\ddag_n$ of system-environment factorized states $\rho_S\otimes\rho_E$ after tracing away the environment, 
$\Lambda^\ddag_n[\rho_S]={\rm Tr}_E\Big(\Theta^\ddag_n[\rho_S\otimes \rho_E]\Big)$.
Indeed, QR regards multi-time correlations 
${\rm Tr}_{S+E}\Big(\rho_S\otimes\rho_E\Big(\widetilde{B}^{(n)\dagger}\widetilde{A}^{(n)}\Big)\Big)$
of observables $A_j, B_j$ of the open system only, evolving according to the Heisenberg joint reversible dynamics, $\Theta_n$, 
associated to $\Theta_n^\ddag$, where
$$\widetilde{A}^{(n)}:=\Theta_{n-1}[A_{n-1}\otimes\mathds{1}_E] \ldots\Theta_1[A_{1}\otimes\mathds{1}_E]A_0\otimes\mathds{1}_E\,,
$$
and amounts to the possibility of expressing them 
in terms of CPU (completely positive and unity-preserving) 
propagators $\Lambda_{n,n-1}$~\cite{ChrusReview22,Lili} acting on $S$:
	\begin{multline*}	
		\Tr_{S+E}\Big(\rho_S\otimes\rho_E\Big(\widetilde{B}^{(n)\dagger}\widetilde{A}^{(n)}\Big)\Big)
		=
		 \Tr_S\!\bigg(\rho_S\,B_{0}^\dagger\,\Lambda_{1}\Big[B_{1}^\dagger\,\Lambda_{2,1}\big[\dots \\ 	\dots \Lambda_{{n-1},{n-2}}\big[B_{{n-1}}^{\dagger} A_{n-1}^{\phantom{\dagger}}\big]\dots\big]A_{1}\Big]\,A_{0}\bigg )\, ,
	\end{multline*}  
for all possible choices of observables of $S$. Another approach to non-Markovian quantum issues is thus to check whether QR is violated.
\par
The proper quantum extension of non-Markovian stochastic processes, including the multi-time statistics, is rooted in a long history dating back to pioneering works~\cite{AccardiFrigerioLewis,Accarditopics,Lindblad1979,Lewis1981,Frigerio1981} and still much debated in recent years. In particular, including the multi-time statistics in the description involves repeated quantum measurements that interfere with the dynamics and generally change the quantum state (see~\cite{AlickiFannesBook, ColloquiumBreuer,milz2020non} and the recent review \cite{Shrikant_Review2023}). Among the most recent proposals concerning such operational approach, particularly prominent ones include the so-called process tensor formalism~\cite{MilzModi_tutorial,ProcessTensorPRL,PollockModiPater2018},  quantum combs (see~\cite{Chiribella08,milz2020non} and references therein), conditional past-future independence~\cite{Budini2018quantum,BudiniOPvdNONOP} and temporal entanglement (see the recent work~\cite{vilkoviskiy2025temporal}). A shared idea among most of these approaches is to map multi-time correlations into spatial ones encoded in a quantum many-body state. 
Furthermore, in such a non-commutative scenario, G.~Lindblad~\cite{Lindblad1979, Lindblad88} proposed to look at Markovianity using 
a quantum version of the Kolmogorov-Sinai (KS) entropy, this latter concept measuring the amount of entropy, that is of ignorance before measuring or, equivalently, of information gathered afterwards, per unit time-step provided by the dynamics.
Such a notion was later fully developed by Alicki and Fannes in~\cite{AF1994defining}, with applications  mainly to the reversible dynamics of  quantum many-body systems~\cite{AlickiFannesBook}. Differently from other proposals of quantum dynamical entropy,
the Alicki-Lindblad-Fannes (ALF) entropy takes into consideration the fact that information about quantum systems is obtained by measuring schemes that, unlike classical phase-space localization with respect to a suitable coarse graining, interfere with the given dynamics. Notice, that such an approach contrasts with the standard concept of information back-flow~\cite{BLP} which instead focuses on the dynamics of the system $S$  only, 
without the need of any measurement and so, effectively, on marginal probabilities.\hfill\break
\indent 
In the following, we adapt the ALF entropy to the open quantum system scenario, and propose it as a measure of information flowing back to the open system, capturing  the full multi-time statistics.
The proposal is presented first in a rather general setting and then applied to the concrete case of so-called collisional models~\cite{CMreview,CampbellCollisionModelsOpen2021a}. Namely, of finite-level systems coupled to an infinite spin chain serving as their environment, whereby the open quantum system interacts locally with a single site algebra at each tick of time. In particular, explicit results will be obtained with the environment being a correlated classical spin chain. Nevertheless, despite the memory-storing mechanisms of our model being entirely classical, the memory effects arising in the open quantum evolution might have no classical counterpart~\cite{FBGN_PhysicaScr}.
Moreover, the multi-time structure  upon which, taking inspiration from the ALF-entropy, our approach is based, has not yet been fully explored even with classical environments and certainly provides a necessary first step towards genuinely quantum environments.
In such a context, we are able to explicitly compute the open-system ALF entropy of the irreversible qubit dynamics and relate its behaviour  to the activation and super-activation of memory effects, the latter being a pure quantum phenomenon whereby back-flow of information shows up putting together a system with no back-flow of information with a copy of itself.
\hfill\break \indent The work is structured as follows. In Section~\ref{sec:preliminaries}, we review the main algebraic and conceptual tools necessary in the remaining sections, such as the algebraic approach to collisional models, the QR formula and the ALF entropy for reversible systems. In Section~\ref{sec:ALFDISS}, we introduce the concept of open-system ALF entropy for generic open quantum systems and compare and contrast its typical behavior for closed finite systems with Markovian open ones. Moreover, we discuss its general physical interpretation in Section~\ref{sec:dynGNS_model} in the context of the so-called GNS construction.
Section~\ref{sec:collisionalmodels} is devoted to the main concrete results of the work that are based a collisional model derivation of the dynamics.

\subsection*{Results}
The main results of the paper are as follows.
First, in Section~\ref{sec:ALFDISS}, we introduce the general framework for an open-system generalization of the ALF entropy. 
Specifically, we give an independent new proof of the results of~\cite{Lindblad1979}; namely, we show that, by looking at specific POVM-dependent entropy rates, one can ascertain the property of QR-Markovianity. 
 
Then, we propose the decrease of such rates with respect to the ones holding in the  QR-Markovian setting as a measure of how much information has flown back from environment to system. 

In Section~\ref{sec:dynGNS_model}, the operational interpretation of the previous proposal is discussed relatively to the state purification scheme known as GNS construction. 
\par
 Section~\ref{sec:collisionalmodels} is dedicated to  concrete applications by framing collisional models within the algebraic approach. First, we provide a general upper bound to the ALF entropy of a finite level system coupled to a quantum spin chain.
Then, we focus on the concrete model of a finite-level system  collisionally interacting with a classical spin chain, for which we explicitly compute the open-system ALF entropy. In particular, we show that a dissipative, highly non-Markovian dynamics can achieve vanishing entropy production, 
as a reversible dynamics would. 
Finally, through the GNS construction, we introduce dissipative dilations of the reduced dynamics, that allow us to assess memory effects that are not detectable by the reduced dynamics, only; in particular, the so-called super-activation of memory effects that may characterize .

\section{Preliminaries: Quantum Regression and ALF entropy}
\label{sec:preliminaries}

In the following we shall identify finite or infinite discrete-time quantum dynamical systems by means of triples $(\mathcal{A},\omega,\Theta)$~\cite{Bratteli} where $\mathcal{A}$ is a suitable algebra of observables, $\omega$ a state, that is a positive, normalized linear functional on $\mathcal{A}$, and $\Theta$ denotes the
dynamics which generically consists of a countable family $\{\Theta_n\}_n$ of completely positive and unital (CPU) maps, not necessarily satisfying the (semi)-group property $\Theta_n=\Theta^n=\Theta\circ\Theta\cdots\circ\Theta$, $n$ times.
As a simple instance, consider a finite $d$ level system; then, $\mathcal{A}=M_d(\mathbb{C})$ is the algebra of $d\times d$ matrices, $\omega(A)=\Tr(\rho\,A)$ with $\rho$ a density matrix, while $\Theta_n$ may be the unitary evolution of a closed, but not necessarily autonomous qubit, $\Theta_n[A]=U_n^\dag\,A\,U_n$, with $U_n$ unitary, or a dissipative channel $\Theta_n[A]=\sum_kL^\dag_k(n)\,A\,L_k(n)$, with $\sum_kL_k^\dag(n)L_k(n)=\mathds{1}$, describing an open $d$-level system in interaction with its environment. In the first case, one talks of an automorphic dynamics as the adjoint and the algebraic relations are preserved; namely, 
$(\Theta_n[X])^\dag=\Theta_n[X^\dag]$ and $\Theta_n[XY]=\Theta_n[X]\Theta_n[Y]$ for all $X,Y\in\mathcal{A}$, $n\in\mathbb{N}$. Moreover, the algebraic inverse $\Theta_n^{-1}$ is an automorphism itself. In the second case only the first relation holds.\par
As an infinite dimensional instance, consider a quantum spin chain whose algebra $\mathcal{A}$ is generated by all  tensor products $X^{[a,b]}_{\bs{i}_{[a,b]}}=\bigotimes_{k=a}^{b}X^{(k)}_{i_k}$, where the upper index refers to the site at which the operator $X_{i_k}$ is located. These tensor products form strictly local algebras $\mathcal{A}^{[a,b]}$ supported by the site interval $[a,b]$.
States $\omega$ on $\mathcal{A}$ are specified by a set of density matrices $\rho^{[a,b]}\in\mathcal{A}^{[a,b]}$ such that:
\begin{align}
	\label{localstates}
	\omega\left(\bigotimes_{k=a}^b X^{(k)}_{i_k}\right)&=\Tr\left(\rho^{[a,b]}\bigotimes_{k=a}^b X^{(k)}_{i_k}\right)\ , \\ 
	\Tr_{b}\left(\rho^{[a,b]}\right)&=\rho^{[a,b-1]}\ ,
\end{align}
where  $\Tr_k$ defines the partial trace over the $k$-th site.  
Furthermore, if ${\rm Tr}_{a}\rho^{[a,b]}=\rho^{[a+1,b]}=\rho^{[a,b-1]}$, then the state $\omega$ is invariant under the (right) shift,  
$\omega\circ\sigma=\omega$, where $\sigma:\mc{A}\to\mc{A}$ moves all single site observables one step to the right:
\begin{equation}
	\label{shift}
	\sigma\Big[X^{(a)}_{i_{a}}\otimes\cdots\otimes X^{(b)}_{i_b}\Big]=X^{(a+1)}_{i_{a}}\otimes\cdots\otimes X^{(b+1)}_{i_b}\ .
\end{equation}
In such a case, the local states $\rho^{[1,n]}$ have von Neumann entropy 
$S(\rho^{[1,n]})=-\Tr\Big(\rho^{[1,n]}\log\rho^{[1,n]}\Big)$ and the chain is characterized by a mean entropy
\begin{equation}
	\label{meanentropy}
	\mathfrak{S}_{\omega}:=\lim_n\frac{1}{n}S(\rho^{[1,n]})\ .
\end{equation}
The limit exists because of the strong sub-additivity of the von Neumann entropy and of the assumed shift-invariance of $\omega$~\cite{AlickiFannesBook}; actually, for the same reasons the following limit also holds:
\begin{equation}\label{meanentropy2}
	\mathfrak{S}_\omega=\lim_n\left[ S(\rho^{[1,n]})-S(\rho^{[1,n-1]})\right]\,.
\end{equation}

\subsection{Collisional models}\label{sec:collisionalmodels_intro}
The reduced dynamics of open quantum systems emerges from suitably handling their non negligible interactions with the environment: collisional models prove a most convenient way of doing that, in a way that allows several types of reduced dynamics to emerge. They consist of a suitable quantum spin chain environment $E$ with algebra $\mathcal{A}_E$ generated by strictly local matrix algebras 
$M^{[-a,b]}_D(\mathbb{C})=\bigotimes_{k=-a}^b M^{(k)}_D(\mathbb{C})$.
The chain is coupled to  a $d$-level system $S$, $\mc{A}_S=\Md{d}$ and the tensor algebra $\mathcal{A}_S\otimes\mathcal{A}_E$ describing the compound system $S+E$ is endowed with a factorized state $\omega_{SE}=\omega_S\otimes\omega_E$, represented on each strictly local
algebra
$\mathcal{A}_S\otimes\mathcal{A}_E^{[-a,b]}$ by a factorized density matrix 
$\Omega_{S[-a,b]}=\rho_S\otimes\rho_E^{[-a,b]}$. 

The $SE$ collisional coupling is constructed by means of an automorphism $\Phi$ on $\mathcal{A_S}\otimes \mathcal{A}_E^{(0)}$ that describes the local interaction of $S$ with the chain $0$-th site and 
extends to  $\mathcal{A}_S\otimes \mathcal{A}_E$ as
\begin{multline}
	\Phi\Big[X_S\otimes X^{[-a,-1]}_{\bs{i}_{[-a,-1]}}\otimes X_{i_0}^{(0)}\otimes X_{\bs{i}_{[1,b]}}^{[1,b]}\Big] \\ =X_{\bs{i}_{[-a,-1]}}^{[-a,-1]}\otimes
	\Phi[X_S\otimes X_{i_0}^{(0)}]\otimes X_{\bs{i}_{[1,b]}}^{[1,b]}\ .
	\label{coupling}
\end{multline}
The compound automorphic dynamics $\Theta$ of system $S$ and chain $E$ is finally given by composing $\Phi$ with the shift dynamics $\sigma_E$ on $\mathcal{A}_E$, yielding a discrete group of automorphisms $\{\Theta_n={\Theta^n_{\rm coll}}\}_{n\in\mathbb{N}}$, where 
\begin{equation}	
	\label{reduceddynamicsn}
	{\Theta_{\rm coll}}:= ({\rm id}_S\otimes\sigma_E)\circ\Phi \ .
\end{equation}
The algebraic implementation of~\eqref{reduceddynamicsn} is illustrated in Figure~\ref{fig:collisionalfigure}. 
Collisional models will be exploited in Section~\ref{sec:collisionalmodels}, where the open-system ALF entropy will be worked out explicitly in a concrete example.

\hskip-0.1cm
\begin{figure}[t!]
	\begin{tikzpicture}[scale=0.89, transform shape, baseline=(current bounding box.center),
	every node/.style={font=\normalsize},
	tensor/.style={inner sep=1.1pt},
	graybox/.style={fill=gray!25, rounded corners=3pt, inner sep=2.6pt},
	purplebox/.style={fill=purple!10, inner xsep=10pt,   
		inner ysep=20pt, rounded corners=2.2pt,minimum height=2.5em},
	]

	\node[label] (int) at (-12.2,1.2) {\small{(a)}} ;
	\node[label] (phi) at (-11.5,1.2) {$\boxed{\Phi}$};

	\node[tensor] (left1) at (-10.2,-0.3){$\cdots 
			\otimes
		 \mathcal{A}_E^{(-a)} \otimes \mathcal{A}_E^{(-a+1)}  \otimes\cdots \otimes$};
	
	
	\node[graybox] (ASmid) at (-7.45,0.2){
		$
		\begin{aligned}
			&\mathcal{A}_S\\
			&\;\, \otimes \\
			&\mathcal{A}_E^{(0)}
		\end{aligned}
		$
	};
	

	\node[tensor, right=.9cm of left1] (right1) {$\otimes\, 
		\cdots \otimes \mathcal{A}_E^{(b)} \!\otimes \mathcal{A}_E^{(b+1)} \otimes
		\cdots
		$
		};

	\node[label] (int) at (-12.2,-1.5) {\small{(b)}} ;

	\node[purplebox, tensor] (block1) at (-8.3,-3.1) {$ \mathcal{A}_E^{(-a)} \otimes \mathcal{A}_E^{(-a+1)} \otimes \cdots \otimes \mathcal{A}_E^{(0)} \otimes
		 \cdots
		 \otimes\mathcal{A}_E^{(b)} $};

	\node (ASmid2) at (-7.6,-2) {
		$
		\begin{aligned}
			&\mathcal{A}_S\\
			&\; \otimes \\
		\end{aligned}
		$
	};
	
	\node[tensor,left=0.0001cm of block1] (left2){$\cdots
		\otimes
		$};
	
	\node[tensor, right=0.1cm of block1] (right2) {$ \!\otimes \mathcal{A}_E^{(b+1)} 
		\otimes
		\cdots
		$
	};

	\node[label] (int) at (-12.2,-5) {\small{(c)}} ;
	\node[label] (ids) at (-11.,-5) {$\boxed{\mathrm{id}_S \otimes \sigma_E}$};

	\node[purplebox, tensor] (block2) at (-7.15,-6.5) {$
		\mathcal{A}_E^{(-a+1)}  \otimes 
		\cdots 
		\otimes \ \mathcal{A}_E^{(0)} \otimes 
		 \cdots\otimes\mathcal{A}_E^{(b)}
		\otimes\mathcal{A}_E^{(b+1)}
		$};

	\node (ASmid2) at (-7.5,-5.4) {
		$
		\begin{aligned}
			&\mathcal{A}_S\\
			&\; \otimes \\
		\end{aligned}
		$
	};
	
	\node[tensor,left=0.001cm of block2] (left2){$
		\cdots 
		\otimes
		 \mathcal{A}_E^{(-a)} 
		 \otimes $};
	
	\node[tensor, right=0.1cm of block2] (right2) {$ 
		\otimes 
		\cdots
		$
	};
\end{tikzpicture}
\caption{One step $\Theta_{\rm coll}$ of the algebraic collisional dynamics: (a) $\Phi$ acts on the algebra of the system $\mc{A}_S$ and on the $0$-th site of the chain. (b) An operator $A^{[-a,b]}$, localized in $\mc{A}_E^{[-a,b]}$, is then translated to the right in (c) by the shift automorphism $\sigma_E$.}
\label{fig:collisionalfigure}
\end{figure}

\subsection{Quantum Regression}\label{sec:QR}

Non-Markovianity is intuitively associated with the presence of memory effects; as already anticipated in the Introduction, given the reduced dynamics $\{\Lambda^\ddag_n\}_{n\geq 0}$ of an open quantum system, a standard approach is via the notion of \emph{divisibility}. Namely, one looks at propagators     
\begin{equation}
    \label{schrdivisibility}\Lambda^\ddagger_n=\Lambda^\ddagger_{n,m}\circ\Lambda^\ddagger_m\,, \quad n\ge m
    \end{equation}
and calls the reduced evolution (C)P-divisible if $\Lambda_{n,m}^\ddagger$ are (completely) positive and trace preserving for all $n\ge m$. In such a setting, where only the reduced dynamics of the system is available, 
CP-divisibility represents a standard notion of Markovianity~\cite{Lili, ChrusReview22}. An associated remarkable feature of non-Markovianity is the so-called back-flow of information~\cite{BLP,GTDWissmannBreuerAmato}, that manifests itself when the evolution is not even P-divisible. 
For algebraically invertible dynamics, in fact, lack of P-divisibility equals lack of trace-norm contractivity with the consequence that
the distinguishability of two density matrices $\rho, \sigma$ might increase in time. Such  time-revivals of trace norms, $\normt{\rho-\sigma}$, as well as of other distances or entropic divergences,  
is accordingly interpreted as  back-flow of information from the environment into the system, a clear signature of memory effects.
We further stress that, although both rely only on the reduced dynamics, CP-divisibility and the lack of  back-flow of information are inequivalent notions of Markovianity. Indeed, CP-divisibility implies the absence of  back-flow of information, whereas the converse does not hold~\cite{Lili}.
Back-flow of information can also occur in a commutative setting when the transition matrices from time-step $m$ to time-step $n$ are not positive~\cite{BCN, Amato2025bridging}; instead, a purely quantum phenomenon, impossible in a commutative scenario, is the super-activation of back-flow of information~\cite{BenattiChrusFil,FBGN2023}.
Indeed, going from a single-system reduced dynamics  $\Lambda^\ddag_n$ to a bipartite one of the form $\Lambda^\ddag_n\otimes \Lambda^\ddag_n$, lack of CP-divisibility  of the former might generally cause super-activation of back-flow of information in the latter,  while the subsystems do not display back-flow in terms of distinguishability revivals.
\par
Evidently, the physical mechanisms behind a possible back-flow of information cannot be appreciated sticking to the open quantum system only, as it asks to
take into account  the full joint system-environment dynamics $\Theta$ which originates that particular dissipative open quantum dynamics $\{\Lambda_n\}_{n\geq 0}$. 
	For the moment, we shall not assume that the family $\{\Theta_n\}_n$ forms a group, namely the automorphic Heisenberg propagators
	\begin{equation}\label{theta_propagators}
		\Theta_{n,m}=\Theta_m^{-1}\circ\Theta_n^{\phantom{-1}}
	\end{equation}
	between subsequent times $n \ge m$ might be genuinely time dependent.
As emphasized in the Introduction, by focussing upon multi-time correlation functions, one can 
substantially refine the definition of  Markovian open quantum behavior  by identifying it with QR~\cite{GardinerZoller, BreuerPetruccione}, that we now reformulate. 
QR requires measuring observables of the open system only as the environment is typically neither accessible nor under control.
Using the algebraic notation just introduced,  let then $\{A_k\}_{k=0}^{n-1} \in\mathcal{A}_S$, $n\ge1$, define
\begin{equation}
    \widetilde{A}^{(n)}:=\Theta_{n-1}[A_{n-1}\otimes\mathds{1}_E] \ldots\Theta_1[A_{1}\otimes\mathds{1}_E]A_0\otimes\mathds{1}_E 
    \,,
\end{equation}
belonging to   $\mathcal{A}_S\otimes \mathcal{A}_E$ and evaluate the multi-time correlation functions $\omega_S\otimes\omega_E\left(\widetilde{B}^{(n)\dagger}\widetilde{A}^{(n)}\right)$ given by
\begin{align}
	\label{pre_QR}	
	&\omega_{S}\otimes\omega_E\bigg(B_0^\dagger\otimes\mathds{1}_E\,\Theta_1\left[B_1^\dagger\otimes\mathds{1}_E\right]\ldots\\ \nonumber
	&\ldots\Theta_{n-1}\left[B_{n-1}^{\dagger} A_{n-1}^{\phantom{\dagger}}\otimes \mathds{1}_E\right]\ldots\Theta_1\left[A_{1}\otimes\mathds{1}_E\right]A_{0}\otimes\mathds{1}_E\bigg)\, .
\end{align}

\begin{definition}
	Let the triple $(\mc{A}_S\otimes\mc{A}_E,\widetilde{\omega}_S\otimes \omega_E, \Theta)$ define an open dynamical system, with $\widetilde{\omega}_S$ a generic state of the system. The open quantum system $S$ will be called 
	\emph{QR-Markovian}
	if it satisfies QR, namely if,
	for arbitrary $\{A_k\}_{k=0}^{n-1}$, $\{B_k\}_{k=0}^{n-1}$  in $\mc{A}_S$, $n\ge 1$,
	\begin{align}\nonumber	&\widetilde{\omega}_S\otimes\omega_E\Big(\widetilde{B}^{(n)\dagger}\widetilde{A}^{(n)}\Big)\\ 
		&\!\!= \widetilde{\omega}_S\bigg(B_{0}^\dagger\Lambda_{1}\bigg[B_{1}^\dagger\Lambda_{2,1}\Big[\dots\Lambda_{{n-1},{n-2}}\left[B_{{n-1}}^{\dagger} A_{n-1}^{\phantom{\dagger}}\right]\dots
		\Big]A_{1}\bigg]A_{0}\bigg),\label{qrformula}
	\end{align} 
	with $\Lambda_{n,n-1},\, n\in\mathbb{N}$ completely positive, unital maps. 
\end{definition}
For $n=j+1$ and $A_i=B_i=\mathds{1}_S$, for $i=0,1,\ldots, j-1$ QR yields
\begin{align}
\nonumber
&\omega_S\otimes\omega_E\Big(\Theta_j\Big[B^\dag_{j}\,A_{j}\otimes\mathds{1}_E\Big]\Big)=\\&=\omega_S\Big(\Lambda_1\circ\Lambda_{2,1}\circ\cdots\circ\Lambda_{j,j-1}\Big[B_{j}^\dag\,A_{j}\Big]\Big)	\nonumber\\
\label{complaw}
&=\Tr\Big(\Lambda_{j,j-1}^\ddag\circ\cdots\circ\Lambda^\ddag_{2,1}\circ\Lambda_1^\ddag[\rho_S]\,B_{j}^\dag\,A_{j}\Big)\ .
\end{align} 
The left hand side above defines the reduced dynamics of the open system $S$ at time-step $j$ so that, in the Schrödinger picture,
one retrieves the composition law 
\begin{equation}
\label{complaw1}
\Lambda^\ddag_j=\Lambda_{j,j-1}^\ddag\circ\cdots\circ\Lambda^\ddag_{2,1}\circ\Lambda_1^\ddag\ .
\end{equation}
Furthermore, since in QR scenario the one-step propagators
$\Lambda^\ddag_{j,j-1}$ must be CPTP, the reduced dynamics is  Markovian in the divisibility sense of~\cite{RivasHuelgaPlenio} (see Remark~\ref{rmk:QRimpliesCPdiv}). This is only a necessary condition for QR-Markovianity; indeed, as we shall show later in a concrete 
example, CPTP propagators need not imply QR.

\begin{remark}\label{QRCR}
	Suppose that the system algebra is maximally commutative, namely, spanned by one-dimensional orthogonal projectors $P_i, P_i P_j = \delta_{ij}P_j$.
	Then,  QR condition reduces to CR. Indeed, let $C_k=B_k^\dag A_k=\sum_j c_j^{(k)} P_j$;
	due to commutativity, using~\eqref{theta_propagators}, we can rewrite~\eqref{pre_QR} as
	\begin{align*}
		&\omega_S\otimes \omega_E\left(\Theta_1 \bigg[\Theta_{2,1}\Big[\dots
		\Theta_{n-1,n-2}\big(C_{n-1}\otimes\mathds{1}_E\big)\dots\right.\\
		&\hskip+4cm\left.\dots
		\Big]C_1\otimes\mathds{1}_E\bigg]C_0\otimes\mathds{1}_E\right)
		\\
		&=\sum_{i_0 \dots i_{n-1}} c_{i_0}^{(0)}c_{i_1}^{(1)} \dots c_{i_{n-1}}^{(n-1)}\ p^{[0,n-1]}_{i_0 i_1 \dots i_{n-1}}\ ,
	\end{align*}
	where we let
	\begin{align}\label{lhs_cr}
	&p^{[0,n-1]}_{i_0 i_1 \dots i_{n-1}}:=\omega_S\otimes \omega_E\left(\Theta_1\bigg[\Theta_{2,1}\Big[\dots
	\right. \\
	&\hskip+1cm \left.
	\Theta_{n-1,n-2}\big[P_{i_{n-1}}\otimes\mathds{1}_E\big]\dots
	\Big]P_{i_1}\otimes\mathds{1}_E\bigg]P_{i_0}\otimes\mathds{1}_E\right).\nonumber
	\end{align}
	On the other hand,~\eqref{qrformula} can be rewritten as
	\begin{align}\nonumber &\omega_S\bigg(\Lambda_{1}\Big[\Lambda_{2,1}\Big[\dots\Lambda_{{n-1},{n-2}}\left[C_{{n-1}} \right]\dots
		\dots
		\Big]C_{1}\Big]C_{0}\bigg)\nonumber\\
		&=\sum_{i_0 \dots i_{n-1}} c_{i_0}^{(0)}c_{i_1}^{(1)} \dots c_{i_{n-1}}^{(n-1)} \ \omega_S\Big(P_{i_0}\Lambda_{1}\Big[P_{i_1}\Lambda_{2,1}\Big[\dots \nonumber\\
		& \hskip+1cm \,P_{i_{n-2}}\Lambda_{{n-1},{n-2}}\left[P_{i_{n-1}} \right]P_{i_{n-2}}\dots 
		\Big]P_{i_1}\Big]P_{i_0}\Big)\nonumber\\
		&=\sum_{i_0 \dots i_{n-1}} c_{i_0}^{(0)}c_{i_1}^{(1)} \dots c_{i_{n-1}}^{(n-1)} \prod_{k=0}^{n-1} T_{i_{k} i_{k-1}}(k,k-1)
		\,p_{i_0}\,,\label{rhs_cr}
	\end{align} 
	where $p_j=\omega_S(P_j)$, we defined $T_{ij}(k,k-1):=\Tr(P_j\Lambda[P_i])$ and used that $P_{i_{k-1}} \Lambda_{{k},k-1} [P_{i_{k}}]P_{i_{k-1}}=T_{i_k i_{k-1}}(k,k-1) P_{i_{k-1}}$. 
	By equating~\eqref{lhs_cr} and ~\eqref{rhs_cr}, we recover~\eqref{CRF}.
	
\end{remark}

\subsection{ALF Entropy}\label{sec:ALF}

The degree of instability of a classical dynamics, for instance the existence of a positive Lyapounov exponent, can be associated to dynamical entropy production
by means of the Kolmogorov-Sinai entropy~\cite{Walters1982}. The basic idea is to reduce the dynamics to the shift along sequences of symbols labeling 
the cells of a finite partition of the phase-space that are visited by the time-evolving system at each time-step. Any time-invariant classical state
assigns probabilities to the sequences with Shannon entropies tending  to a well-defined asymptotic rate.
The Kolmogorov-Sinai entropy is the largest among such rates.
The  ALF-entropy provides a quantum counterpart to it in that
the coarse-graining of the phase-space is replaced by a POVM-measurements at successive time-steps of the state dynamics.
Let a quantum dynamical system $(\mathcal{A}, \omega,\Theta)$ be endowed with a reversible automorphic dynamics $\Theta$  and with  a $\Theta$-invariant state $\omega$ on $\mathcal{A}$, $\omega\circ\Theta_n=\omega$.
 Then, any finite POVM $\mathcal{X}$ consisting of $|\mathcal{X}|$ operators $\mathcal{X}:=\{X_a\}_{a=1}^{|\mathcal{X}|} \,\subseteq\,\mc{A}$ 
such that $\sum_{a=1}^{\abs{\mathcal{X}}} X_a^\dagger X_a^{\phantom{*}} =\mathds{1}_{\mathcal{A}}$,
gives rise to a family of time-dependent POVM $\mathcal{X}^{(n)}=\left\{X_{\bs{a}}^{(n)}\right\}_{\bs{a}}$, where $\bs{a}=a_0a_1\cdots a_{n-1}$ and 
\begin{equation}\label{time_partitionelement}
	X_{\bs{a}}^{(n)}:= \Theta_{n-1}[X_{a_{n-1}}]\,\dots\,\Theta_1[X_{a_1}]\,X_{a_0}\,.
\end{equation}
Notice that this is not true if $\Theta_n$ is dissipative (see Remark~\ref{rmk:dissipativepovm}).
For such a reversible dynamical system, it follows that
\begin{equation}\label{coarse_grained}
\rho\big[\mathcal{X}^{(n)}\big]\ = \ \sum_{\bs{a},\bs{b}} \omega\left(X_{\bs{b}}^{(n)\dagger} \, X_{\bs{a}}^{(n)}\right) \ketbra{\bs{a}_{[0,n-1]}}{\bs{b}_{[0,n-1]}} 
\end{equation}
is a density matrix, that is self-adjoint, positive and normalized with  multi-time correlation functions as matrix elements:
\begin{align}\nonumber
	\left(\rho\big[\mathcal{X}^{(n)}\big]\right)_{\bs{a},\bs{b}}&=\omega\left(X_{\bs{b}}^{(n)\dagger} X_{\bs{a}}^{(n)}\right)\\&\hskip-1.5cm=\omega\bigg(X_{b_0}^\dagger\Theta_1\big[X_{b_1}^\dagger\big]\Theta_2\big[X_{b_2}^\dagger\big]\cdots\Theta_{n-1}\big[X_{b_{n-1}}^{\dagger} X_{a_{n-1}}^{\phantom\dagger}\big] \cdots\nonumber \\
	&\hskip+1.5cm \dots\Theta_2\big[X_{a_2}\big]\Theta_1\big[X_{a_1}\big]X_{a_0}\bigg)\,.\label{cmatrix_elements}
\end{align}
Using~\eqref{theta_propagators} and the fact  that $\Theta_{n,m}[XY]=\Theta_{n,m}[X]\,\Theta_{n,m}[Y]$, along with~\eqref{theta_propagators}, one can recast~\eqref{cmatrix_elements} in a ``nested'' fashion as in Remark~\ref{QRCR}:
\begin{multline}	\label{nestedcmatrix_elements}
	\left(\rho\big[\mathcal{X}^{(n)}\big]\right)_{\bs{a},\bs{b}}=\omega\bigg(X_{b_0}^\dagger\Theta_1\bigg[X_{b_1}^\dagger\Theta_{2,1}\Big[X_{b_2}^\dagger\cdots
	\\
	 \dots\Theta_{n-1,n-2}\big[X_{b_{n-1}}^{\dagger} X_{a_{n-1}}^{\phantom\dagger}\big] \cdots X_{a_2}\big]X_{a_1}\bigg]X_{a_0}\bigg).
\end{multline}
The state $\rho\big[\mathcal{X}^{(n)}]$ provides a description of the statistics of the evolving state as reconstructed by measuring a fixed POVM.

\begin{remark}	\label{rem:HeisSchr}
	By passing  from the Heisenberg to the Schr\"odinger picture via the following dual maps on the space of linear functionals $\omega$ over the algebra $\mathcal{A}$ of the system:
	$\omega\mapsto \mathbb{X}^\ddag_{ab}[\omega]$ and $\omega\mapsto \Theta_{k,k-1}^\ddag[\omega]$ such that
	\begin{align*}
		\mathbb{X}^\ddag_{ab}[\omega](X)&=\omega(X_b^\dag\,X\,X_a)\,,\\  \Theta_{k,k-1}^\ddag[\omega](X)&=\omega(\Theta_{k,k-1}[X])\, ,
	\end{align*}
	one rewrites
	\begin{multline}
		\left(\rho\big[\mathcal{X}^{(n)}\big]\right)_{\bs{a},\bs{b}}=\mathbb{X}^\ddag_{a_{n-1}b_{n-1}}\circ\Theta_{n-1,n-2}^\ddag\circ\cdots\\
		\cdots\circ\mathbb{X}^\ddag_{a_1b_1}\circ\Theta_1^\ddag\circ\mathbb{X}^\ddag_{a_0b_0}[\omega](\mathds{1})\, .
		\label{dualpicture}
	\end{multline}
	The diagonal entries of $\rho\big[\mathcal{X}^{(n)}\big]$ are then the probabilities of sequences of selective measurements of the various POVM terms followed by the one-step reversible dynamics. 
\end{remark}
The information gathered in the long run can then be characterized by the so-called ALF entropy of the given POVM~\cite{AF1994defining}, by means of the following entropy rate
	\begin{align} \label{ALF_rate}
		\mathfrak{h}_\omega(\Theta, \mathcal{X}):=\limsup_{n} \frac{1}{n} S\left(\rho\big[\mc{X}^{(n)}\big]\right)\, ,	
		\end{align}
where $S(\rho)$ is the von Neumann entropy. One has to do with the $\limsup$ since the limit of the ratio could not exist. Indeed, even if the family of automorphisms $\{\Theta_n\}_n$ form a group, the family of density matrices $\rho\big[\mc{X}^{(n)}\big]$ is not generally stationary since measuring the POVM changes in general the state on which the POVM measurement is performed.
The entropy rate is finally made independent from the chosen POVM by maximizing over all of them chosen from a suitable subalgebra of physically admissible partitions $\mathcal{B}$, typically assumed to be globally $\Theta$-invariant and dense in $\mc{A}$: 
\begin{equation}\label{def:dynamicalentropy}
		\mathfrak{h}_\omega^{\mathcal{B}}(\Theta)\ :=\ \sup_{\mc{X}\subseteq\mathcal{B}} \ \mathfrak{h}_\omega(\Theta, \mathcal{X})\; .
	\end{equation}
A useful control of the ALF entropy is in some cases provided by the following bound whose proof can be found in~\cite{AlickiFannesBook}.
\begin{proposition}\label{prop:finitelevel}
	Let  $A \in \mathcal{A}$ and $\omega(A)=\Tr(\rho A)$ for some density matrix $\rho$ 
	and let $\mc{X}=\{{X}_a\}_{a=1}^{|\mc{X}|}$ be a POVM.
	Then,
	\begin{equation}
		\label{finite_h}S\left(\rho\left[\mathcal{X}\right]\right)\ \le\ S(\rho)\ + \ S\left(\sum_{a=1}^{\abs{\mc{X}}} \,X_a\, \rho X_a^\dagger\right)\, .
	\end{equation}
\end{proposition}

   \begin{remark}\label{rmk:dissipativepovm}
   	 If the quantum dynamical system $(\mathcal{A},\omega,\Theta)$ is dissipative, the dynamics $\Theta$ stands for a family of  
   	 CPTP channels $\Lambda_n$. Then, given a POVM 
   	 $\mathcal{X}=\{X_i\}_{i=1}^{\abs{\mathcal{X}}}$, it is no longer true that $\mathcal{X}^{(n)}$ as in~\eqref{time_partitionelement} is a POVM as required by the
   	 ALF construction. Indeed, for general CPU maps the Schwarz inequality yields
   	 \begin{multline*}
   	 	\sum_{a_{n-1}}\Lambda_{n-1}[X^\dag_{a_{n-1}}]\Lambda_{n-1}[X_{a_{n-1}}]
   	 	\\
   	 	\le \sum_{a_{n-1}}\Lambda_{n-1}[ X^\dag_{a_{n-1}}X_{a_{n-1}}]=\mathds{1}\, ,
   	 \end{multline*}
   	 so that the quantity
   	 \begin{multline*}
   	 	\sum_{\bs{a}^{(n)}}X^\dag_{\bs{a}^{(n)}}X_{\bs{a}^{(n)}}= \sum_{\bs{a}^{(n)}}X_{a_0}\Lambda_1(X^\dag_{a_1})\cdots\\ \cdots\Lambda_{n-1}[X^\dag_{a_{n-1}}]\Lambda_{n-1}[X_{a_{n-1}}]
   	 	\cdots\Lambda_1[X_{a_1}]\,X_{a_0}\ 
   	 \end{multline*}
   	 does not generally sum  to the identity.  
   	 Nevertheless, consider the nested form~\eqref{nestedcmatrix_elements} of the coarse grained density matrix; it then appears natural to replace the automorphism $\Theta$ with CPU intertwiners $\Lambda_{n,m}$:
   	 \begin{multline}\label{Lambdanested}
   	 	\left(\rho\big[\mathcal{X}^{(n)}\big]\right)_{\bs{a},\bs{b}}=\omega\Big(X_{b_0}^\dagger\Lambda_1\big[X_{b_1}^\dagger\Lambda_{2,1}\big[X_{b_2}^\dagger\cdots \\
   	 	\cdots \Lambda_{n-1,n-2}\big[X_{b_{n-1}}^{\dagger} X_{a_{n-1}}^{\phantom\dagger}\big]X_{a_2}\big]X_{a_1}\big]X_{a_0}\Big)\ .
   	 \end{multline}
   	 One thus obtains a compatible family of legitimate density matrices, $\Tr_{n-1}\left(\rho[\mathcal{X}^{(n)}]\right)=\rho[\mathcal{X}^{(n-1)}]$.
   	 As seen in the Quantum Regression formula~\eqref{qrformula}, \eqref{Lambdanested} has the precise physical meaning of a coarse-grained density matrix of an open system undergoing a QR-Markovian dynamics.
   \end{remark}

\section{Open-system ALF entropy}

\label{sec:ALFDISS}

In this Section, we introduce a generalization of the ALF framework to study quantum dissipative evolutions. We will consider concrete applications to collisional models later, in Section~\ref{sec:collisionalmodels}.

As seen in Section~\ref{sec:QR}, rather than through the reduced dynamics, the statistical properties of the dissipative behaviour are
better accessed by using 
multi-time correlation functions of the form $\omega_S\otimes\omega_E\left(\widetilde{B}^{(n)\dagger}\widetilde{A}^{(n)}\right)$ in~\eqref{pre_QR}.
These involve the compound automorphic dynamics $\Theta$ of system $S$ and environment $E$ together and operators $X_S\otimes\mathds{1}_E$  of the system $S$, only. 

Within the ALF construction applied to the compound system $S+E$ restricting to such particular operators is justified by the  experimental inaccessibility of the environment: we shall thus focus on POVMs $\mc{X}$ belonging to the subalgebra $\subseteq \Md{d}\otimes\mathds{1}_E$. 
\par
In practice, we shall then focus upon the dynamical triple $(\mathcal{A}_S\otimes\mathcal{A}_E,\omega_S\otimes\omega_E,\Theta)$ where
$\mathcal{A}_S$ is a finite matrix algebra of $S$ operators, $\mathcal{A}_E$ a quantum spin chain algebra, evolving under a family $\Theta$ of 
automorphisms, while $\omega_S\otimes\omega_E$ is a $\Theta$-invariant state over $\mathcal{A}_S\otimes\mathcal{A}_E$.

Following the construction described in Section~\ref{sec:ALF}, the elements of the coarse-grained density matrix are thus of the form~\eqref{pre_QR}: 
\begin{multline}\label{system_cmatrix}
\left(\rho_S\Big[\mc{X}^{(n)}\Big]\right)_{\bs{a},\bs{b}}
=\omega_{S}\otimes\omega_E\Big(X^{\dagger}_{b_0}\otimes\mathds{1}_E\, \Theta_1[X^{\dagger}_{b_1}\otimes\mathds{1}_E]\ldots\\
\dots\Theta_{n-1}[X_{b_{n-1}}^{\dagger} X_{a_{n-1}}\otimes \mathds{1}_E]\dots\Theta_1[X_{a_1}\otimes\mathds{1}_E]X_{a_0}\otimes\mathds{1}_E\Big).
\end{multline}
 Here, the subscript $S$ stressed that we are restricting all possible POVMs to the subalgebra $\mc{A}_S\otimes \mathds{1}_E$, only.
One can thus compute the entropy rate
\begin{equation}
\label{ALFrate}
	\mathfrak{h}_S(\Theta,\mc{X}):=
	\limsup_{n} \frac{1}{n} S\left(\rho_S\left[\mc{X}^{(n)}\right]\right), \quad \mc{X} \subseteq \mc{A}_S\otimes\mathds{1}_E\,,
\end{equation}
and maximize it over the POVMs of the open system to obtain the asymptotic entropy rate
\begin{equation}
\mathfrak{h}_S(\Theta)\quad :=\
\sup_{\mc{X}\,\subseteq\,\mathcal{A}_S \otimes \mathds{1}_E}\mathfrak{h}_{S}(\Theta, \mathcal{X})\, ,
\label{ALFopen}
\end{equation}
Note that, differently from what is usually done in most applications of the ALF entropy, the optimization is now over the subalgebra $\mathcal{A}_S\otimes \mathds{1}_E$, which is neither time-invariant nor, clearly, dense in the full $S+E$ algebra. Rather, \eqref{ALFopen} represents the information rate due to repeated measurements on $S$ only, intertwining the $S+E$ unitary evolution (see Section~\ref{sec:dynGNS_model}). Accordingly, we shall call it the \emph{open-system ALF entropy}.

\begin{remark}\label{rmk:measurementinterfere}
	Here, following the Alicki–Fannes approach~\cite{AF1994defining,AlickiFannesBook} and unlike previous operational approaches to quantum dynamical entropies~\cite{OhyaPetz,Lindblad88}, we do not impose any restrictions on the measurements performed on the system, which may therefore generally interfere with both the dynamics and the state.
\end{remark}

 \begin{figure}[t]
	\begin{tikzpicture}[scale=0.89, transform shape, baseline=(current bounding box.center),
		every node/.style={font=\normalsize},
		tensor/.style={inner sep=1.1pt},
		graybox/.style={fill=gray!25, rounded corners=3pt, inner sep=2.6pt},
		purplebox/.style={fill=blue!10, inner xsep=10pt,   
			inner ysep=20pt, rounded corners=2.2pt,minimum height=2.5em},
		]

		\node[label] (int) at (-12.2,1.2) {\small{(a)}} ;
		\node[label] (phi) at (-11.5,1.2) {$\boxed{\Theta}$};

		\node[tensor] (left1) at (-10.2,0.2){$\cdots 
			\otimes
			\mathcal{A}_S^{(-a)} \otimes \mathcal{A}_S^{(-a+1)}  \otimes\cdots \otimes$};

		\node[graybox] (ASmid) at (-7.45,-0.3){
			$
			\begin{aligned}
				&\mathcal{A}_S^{(0)}\\
				&\;\, \otimes \\
				&\mathcal{A}_E
			\end{aligned}
			$
		};

		\node[tensor, right=.9cm of left1] (right1) {$\otimes\, 
			\cdots \otimes \mathcal{A}_S^{(b)} \!\otimes \mathcal{A}_S^{(b+1)} \otimes
			\cdots
			$
		};

		\node[label] (int) at (-12.2,-1.5) {\small{(b)}} ;

		\node[purplebox, tensor] (block1) at (-8.3,-2.3) {$ \mathcal{A}_S^{(-a)} \otimes \mathcal{A}_S^{(-a+1)} \otimes \cdots \otimes \mathcal{A}_S^{(0)} \otimes
			\cdots
			\otimes\mathcal{A}_S^{(b)} $};

		\node (ASmid2) at (-7.5,-3.3) {
			$
			\begin{aligned}
				&\; \otimes \\
				&\mathcal{A}_E
			\end{aligned}
			$
		};
		
		\node[tensor,left=0.0001cm of block1] (left2){$\cdots
			\otimes
			$};
		
		\node[tensor, right=0.1cm of block1] (right2) {$ \!\otimes \mathcal{A}_S^{(b+1)} 
			\otimes
			\cdots
			$
		};

		\node[label] (int) at (-12.2,-4.3) {\small{(c)}} ;
		\node[label] (ids) at (-11.,-4.3) {$\boxed{ \sigma_S\otimes\mathrm{id}_E }$};

		\node[purplebox, tensor] (block2) at (-7.15,-5.5) {$
			\mathcal{A}_S^{(-a+1)}  \otimes 
			\cdots 
			\otimes \ \mathcal{A}_S^{(0)} \otimes
			\cdots\otimes\mathcal{A}_S^{(b)}
			\otimes\mathcal{A}_S^{(b+1)}
			$};

		\node (ASmid2) at (-7.5,-6.5) {
			$
			\begin{aligned}
				&\; \otimes \\
				&\mathcal{A}_E
			\end{aligned}
			$
		};
		
		\node[tensor,left=0.001cm of block2] (left2){$
			\cdots 
			\otimes
			\mathcal{A}_S^{(-a)} 
			\otimes $};
		
		\node[tensor, right=0.1cm of block2] (right2) {$ 
			\otimes 
			\cdots
			$
		};
\end{tikzpicture}
\caption{Dynamical system underlying the construction of the map~\eqref{mapdefinition_enunciato}.  (a) $\Theta$ acts non trivially on the $0$-th copy of the system $\mc{A}_S^{(0)}=M_d^{(0)}(\mathbb{C})$ and on the algebra of the environment. (b) An operator, initially localized in $\mc{A}_S^{[-a,b]}$, is then translated to the right in (c) by the shift automorphism $\sigma_S$.
	Note the similarity with the algebraic collisional model of Figure~\ref{fig:collisionalfigure}.
}
\label{fig:dualfigure}
\end{figure}

\subsection{Quantum correlation maps}
The coarse-grained density matrices \eqref{coarse_grained}  define a many-body quantum state on a quantum spin chain that encodes the multi-time correlation functions of the original open quantum system $S$. In principle, the structure of such state is highly non-trivial and depends on the choice of the POVM.
Suppose that the $d$-level open system $S$ is endowed with a faithful state $\omega_S$, namely, represented by a density matrix $\rho_S$ without zero eigenvalues.
Select then the following  special POVM  from $\mc{A}_S$,
\begin{equation}\label{specialOPU}
	\mc{F}:=\{F_{a,a'}	\otimes \mathds{1}_E\}_{a,a'=1}^d\,, \quad F_{a,a'}=\sqrt{r_a} \ketbra{r_a}{r_{a'}},
\end{equation}
where $r_a>0, \ket{r_a}$ are the eigenvalues, respectively, eigenvectors of $\rho_S$.  
It preserves $\rho_S$: $\sum_{a,a'=1}^d F_{a,a'}\,\rho_S\, F^\dag_{a,a'}=\rho_S$.
The speciality of such POVM is twofold: i)
the structure of the coarse-grained matrix~$\rho_S\big[\mathcal{F}^{(n)}\big]$ can be explicitly worked out  and ii) it allows to clearly identify the QR-Markovian regime in terms of entropy rates, as we shall see in the following. Note however that this POVM might well not be the one achieving the supremum in~\eqref{ALFopen}.

\par

First, we show that the coarse-grained density matrix corresponding to the POVM~\eqref{specialOPU} has, remarkably, a closed form.
To this end, it proves useful to consider, together with  our original $d$-level open system $S$, the auxiliary multi-partite system $M_d^{\mathbb{Z}}(\mathbb{C})=\bigotimes_{k=-\infty}^\infty\mc{A}_S$; namely, 
 a quantum spin chain algebra, each site of which  being a matrix algebra identical to the system algebra $\mc{A}_S$. Denote by $\sigma_S$ the right shift on such infinite algebra. 
Then, as done in Section~\ref{sec:collisionalmodels_intro}, an operator $\bigotimes_{k=0}^{n-1}A_k\in M_d^{\otimes n}(\mathbb{C})$ can be embedded into the chain algebra $M_d^{\mathbb{Z}}(\mathbb{C})$ by localizing it within the interval $[0,n-1]$.

Assume also that the Heisenberg propagators $\Theta_{j,j-1}$, defined in~\eqref{theta_propagators}, act non-trivially only on $ M_d^{(0)}(\mathbb{C})\otimes \mathcal{A}_E$. With respect to a time-invariant state $\omega_S\otimes \omega_E$, 
consider now 
the CPU map $\mathbb{T}_n : M_{d}^{\otimes n}(\mathbb{C})\to M_{d}^{\otimes n}(\mathbb{C})$,  defined by:
\begin{align}\label{mapdefinition_enunciato}
	&\mathbb{T}_n\left[\bigotimes_{k=0}^{n-1}A_k 
	\right]\\
	&\hskip-1mm:=\,
	\!\omega_E\Bigg(
	\displaystyle\bigcomp_{j=1}^{n}\left( \Theta_{j,j-1} \circ
	(\sigma_S\otimes\mathrm{id}_E)\right)
	\left[\bigotimes_{k=0}^{n-1}A_k^{(-n+k)}\otimes \mathds{1}_E
	\right]
	\!\Bigg) ,\nonumber
\end{align}
where, to compactly denote the ordered composition of maps $\Phi_{j}, j=1,\dots,n$, we introduced the  symbol
$$\bigcomp_{j=1}^n\, \Phi_j:= \Phi_1 \circ \Phi_2 \circ \cdots\circ \Phi_n\,.$$
We shall refer to~\eqref{mapdefinition_enunciato} as \emph{quantum correlation maps}.
We are now ready to state the following

\begin{proposition}\label{thm:partitionF_general}
	The system coarse-grained density matrix with respect to the POVM $\mc{F}$ is 
	given, up to a rearrangement of tensor factors, by
	\begin{equation}\label{thm1_1}
		\rho_S\left[\mc{F}^{(n+1)}\right] =\rho_S\otimes \rho_S \otimes  \left(\mathbb{T}_n^\ddag \otimes \mathrm{id} \left[\ketbra{\sqrt{\rho_S^{\otimes n}}}\right]\right),
	\end{equation}
	where $\mathbb{T}_n^\ddag$ is the CPTP correlation map dual of $\mathbb{T}_n$
	and $\ket{\sqrt{\rho}}$ denotes the standard purification of a faithful density matrix $\rho=\sum_a r_a \ketbra{r_a}$,
	\begin{equation*}
		\ket{\sqrt{\rho}}
		=\sum_a \sqrt{r_a} \ket{r_a\otimes r_a}\,, 	\qquad  0< r_a < 1.
	\end{equation*}
\end{proposition}

\begin{remark}\label{rmk:collisional modelcomparison}
	The interaction-shift automorphism appearing in~\eqref{mapdefinition_enunciato} is characteristic of the collisional approach of Section~\ref{sec:collisionalmodels_intro}, wherein the open system $S$ iteratively interacts with a new site of a spin-chain environment. In~\eqref{mapdefinition_enunciato} the emerging picture is, in a sense, dual, in that it is the same environment~$E$ that interacts with a fresh copy of the system at each time step. Such dynamics is depicted in Figure~\ref{fig:dualfigure}.
\end{remark}

An independent and self-contained proof of the above Proposition~\ref{thm:partitionF_general} based on the formalism of quantum spin chains, is reported in Appendix~\ref{app:proof_partitionF_general}.

\subsection{Open-system ALF entropy and QR Markovianity}\label{sec:open-system-alf-entropy-and-qr-markovianity}
Let us  consider again the open system symbolic dynamics 
	in the Quantum Regression regime. This means that the entries of the coarse-grained density matrix
	$$\left(\rho_S\Big[\mc{X}^{(n)}\Big]\right)_{\bs{a},\bs{b}}=\omega_S\otimes\omega_E\left(X_{\bs{b}}^{(n)\dagger} X_{\bs{a}}^{(n)}\right),$$ 
	become, for all $n\ge 1$,
	\begin{multline}
		\left(\rho_S\Big[\mc{X}^{(n)}\Big]\right)_{\bs{a},\bs{b}}=\omega_S\Bigg(X_{b_0}^\dagger\,\Lambda_{1}\bigg[X_{b_1}^\dagger\,\Lambda_{2,1}\Big[\dots
		\\
	\dots\,\Lambda_{{n-1},{n-2}}\left[X_{b_{n-1}}^{\dagger} X_{a_{n-1}}^{\phantom{\dagger}}\right]\dots\Big]X_{a_1}\bigg]\,X_{a_0}\Bigg).
	\end{multline}

		By choosing the POVM as in~\eqref{specialOPU}, one is lead to the following alternative characterization of QR-Markovianity, thereby recovering the results of Lindblad~\cite{Lindblad1979}.

	\begin{proposition}	\label{prop:QRprop_NEW}
		The  dynamical system  $(\mc{A}_S\otimes\mc{A}_E,\omega_S\otimes \omega_E, \Theta)$
		satisfies QR-Markovianity  if and only if
	\begin{equation}\label{statement:factorizedTn}
		\mathbb{T}_n=\bigotimes_{j=1}^n \Lambda_{j,j-1}\,,
	\end{equation}
	with $\Lambda_{j,j-1}$ CPU maps on $\mathcal{A}_S=\Md{d}$.
	\end{proposition}
	The proof is reported in Appendix~\ref{app:propQRnew}.

	\begin{remark}\label{rmk:QRimpliesCPdiv}
		From Definition~\ref{mapdefinition_enunciato}, map $\mathbb{T}_n$ is a CPU map on $M_d^{\otimes n}(\mathbb{C})$.
		Complete positivity of $\Lambda_{j,j-1}$ in \eqref{statement:factorizedTn} then follows from that of $\mathbb{T}_n$. Moreover, this explicitly shows that 
		QR-Markovianity implies CP-divisibility of the reduced dynamics $\Lambda_n$, thus its P-divisibility from which follows that there is no back-flow of information as distinguishability revivals. In Section~\ref{sec:collisionalmodels}, we will present explicit examples showing that the converse implication does not hold: a (C)P-divisible dynamics need not be QR-Markovian.
	\end{remark}

 Let us now focus upon the case in which the system-environment dynamics is given by a one-parameter group of automorphisms, $\Theta_n=\Theta^n$. Under such fairly general conditions, the QR condition  actually implies a semigroup dynamics. 

\begin{proposition}\label{prop:QRsg}
	Suppose that the QR condition~\eqref{qrformula} holds for a $\Theta$-invariant state $\omega_S\otimes\omega_E$, with $\omega_S$ represented by a faithful density matrix. Then, the dissipative Heisenberg propagators are all the same:
	\begin{equation}\label{allsemigroup}
		\Lambda_{n+1,n}=\Lambda_1\equiv\Lambda\,,\qquad \forall\; n\ge0\,.
	\end{equation}
\end{proposition}	

	The proof is reported in Appendix~\ref{app:QRsg}.
	In this scenario, it is instructive to compare the behavior of the ALF entropy under reversible and non-reversible dynamics.
	
	\paragraph*{Finite, closed quantum system.}
	For a closed system $S$ not interacting with its environment, the dynamics is such that
	$\Theta=\Theta_S\otimes\Theta_E$, 
	with $\Theta_S$ and $\Theta_E$ automorphisms of $\mc{A}_S$ and $\mc{A}_E$.
	In this case QR holds trivially with $\Lambda_{j,j-1}=\Theta_S$ and~\eqref{system_cmatrix} becomes
	\begin{align}\label{noninteracting}
		&\left(\rho_S\Big[\mc{X}^{(n)}\Big]\right)_{\bs{a},\bs{b}}\\&\quad=\Tr(\rho_S X^{\dagger}_{b_0} \Theta_S\big[X^{\dagger}_{b_1}\dots\Theta_S\big[X_{b_n}^\dagger X_{a_n}^{\phantom{\dagger}}\big] \dots X_{a_1}\big] X_{a_0})\,.\nonumber
	\end{align}
	Then, the dynamics maps the system matrix algebra into itself and $\mathfrak{h}_S(\Theta)=0$ as for all closed finite-level systems~\cite{AlickiFannesBook}.
	
\paragraph*{Finite, open quantum system in the QR regime.} 
	From the above Proposition~\ref{thm:partitionF_general}, the following result allows to ascertain QR-Markovianity by looking at POVM-specific entropy rates.
\begin{corollary}
	\label{cor:partition}
	The ALF entropy associated with the POVM $\mc{F}$ in~\eqref{specialOPU} is bounded from above by
	\begin{equation}\label{Lindblad_criterion}
		\mathfrak{h}_S(\Theta,\mc{F})\ \le \ S\left(\Lambda^\ddag\otimes \mathrm{id}_d\left[\ketbra{\sqrt{\rho_S}}\right]\right) .
	\end{equation}
	Furthermore, the bound saturates if and only if $\mathbb{T}_n=\bigotimes_{k=1}^n \Lambda$ or, equivalently, if and only if QR condition holds.
\end{corollary}
See  Appendix~\ref{proof:corollary_part} for the proof.
From its definition in~\eqref{ALFopen}, the POVM-dependent ALF entropy~\eqref{Lindblad_criterion} provides a lower bound to the open-system ALF entropy. 
\begin{corollary}\label{lowerboundALF}
	In the QR regime one has
	\begin{equation}
		\mathfrak{h}_S(\Theta) \ge \mathfrak{h}_S(\Theta,\mathcal{F}) = S\left(\Lambda^\ddag \otimes \mathrm{id}_d\left[\ketbra{\sqrt{\rho_S}}\right]\right)\,,
	\end{equation}
which are strictly positive quantities unless $\Lambda$ is an isometry, $\Lambda[X]=V^\dagger X V$, $V^\dagger V=\mathds{1}_d$. Thus,
for finite quantum systems, the dynamical entropy vanishes in all reversible dynamical settings,
while it is strictly positive for an irreversible QR-Markovian one.
\end{corollary}
\begin{proof}
	The inequality is evident by definition of $\mathfrak{h}_S(\Theta)$. 
	On the other hand, the lower bound vanishes if and only if $\xi_\Lambda:=\Lambda^\ddag \otimes \mathrm{id}_d\left[\ketbra{\sqrt{\rho_S}}\right]$ is a one-dimensional projector,
	 which is obvious if $\Lambda^\ddag$  is an isometry.
	To prove the converse, let $\rho_S=\sum_{a} r_a \ketbra{r_a}, 0<r_a<1$. Notice that the action of the dual map $\Lambda$ (Heisenberg picture) can be written as
	\begin{equation}
		\Lambda[\ketbra{\phi}{\psi}]=\sum_{ab} \frac{1}{\sqrt{r_a r_b}}  \ave{\psi\otimes r_b}{\xi_\Lambda}{\phi\otimes r_a}\ketbra{r_a}{ r_{b}}\,.
	\end{equation}
	If $\xi_\Lambda=\ketbra{\psi_\Lambda}$, we have
	$\Lambda[\ketbra{\psi}{\phi}]= V^\dag\ketbra{\psi}{\phi} V$ with
	$V^\dag:\ket{\phi}\mapsto \braket{\psi_\Lambda}{\phi\otimes a}\ket{a}$. Finally, unity preservation yields $V^\dag V=\mathds{1}_d$.
\end{proof}

In light of the previous results, we propose the following:
\begin{definition}[Entropic back-flow of information measure]
\label{def:entr-mark} 
We identify the difference 
\begin{equation}\label{ALFmeasure}
\Delta \mathfrak{h}_S(\Theta,\mathcal{F}): = S\left(\Lambda^\ddag \otimes \mathrm{id}_d\left[\ketbra{\sqrt{\rho_S}}\right]\right)-\mathfrak{h}_S(\Theta,\mathcal{F})\,,
\end{equation}
between the entropy rate corresponding to QR-Markovianity and the POVM-dependent entropy rate $\mathfrak{h}_S(\Theta,\mathcal{F})$, with $\mathcal{F}$ being the POVM in~\eqref{specialOPU}, as a measure of the back-flow of information from environment to the open system that accounts for the multi-time statistics.
\end{definition}
  The previous definition is motivated by the fact that, as will become clear in the next Section,
  the decrease of the residual ignorance about the system dynamics signals information coming back from the environment.
  
  Crucially, such depletion of gathered information per unit time-step is a stronger witness of back-flow of information than lack of contractivity; indeed, if the distinguishability  of evolving states increase, than there cannot be P-divisibility (contractivity) of propagators, therefore no QR is possible  which, in turn, implies a dynamical entropy rate strictly smaller than the QR-rate. 
	On the other hand, as will be explicitly shown in Section~\ref{sec:collisionalmodels}, $\Delta \mathfrak{h}_S(\Theta,\mathcal{F})$ can be positive even when other standard measures of non-Markovianity vanish, in particular the BLP measure~\cite{BLP}, which identifies back-flow of information with revivals of state distinguishability.

\section{ALF entropy: operational interpretation}
\label{sec:dynGNS_model}

We now discuss a natural physical interpretation of the ALF entropy provided by the so-called GNS representation. 

Given an operator algebra $\mathcal{A}$ equipped with a state $\omega$, the familiar Hilbert space setting of quantum mechanics
emerges through the  approach developed by Gelfand, Neimark and Segal (GNS construction, see e.g.~\cite{MurphyCstar}). It consists in associating to the pair 
$(\mc{A},\omega)$ a triple $(\mathcal{H}_\omega,\pi_\omega,\ket{\Omega_\omega})$ consisting of a Hilbert space $\mathcal{H}_\omega$, a representation $\pi_\omega:	\mc{A}\to\mathcal{B}(\mathcal{H}_\omega)$ of observables as bounded operators on $\mathcal{H}_\omega$ and a state vector $\ket{\Omega_\omega}$ enjoying the cyclic property, i.e. such that the subspace
$$
\pi_\omega(\mc{A})\ket{\Omega_\omega}:=\{\pi_\omega(X)\ket{\Omega_\omega},\ X\in\mathcal{A} \} \,
$$
is dense in $\mathcal{H}_\omega$. The expectation value of any observable $X\in\mc{A}$ can be then expressed as 
\begin{equation}
	\omega(X)=\braket{\Omega_\omega}{\pi_\omega(X)|\Omega_\omega},\quad X\in \mathcal{A}\,.
\end{equation}
The GNS triple $(\mathcal{H}_\omega,\pi_\omega,\ket{\Omega_\omega})$ is unique up to a unitary equivalence. From this fact, one deduces that 
every automorphism $\Theta:\mc{A}\to\mc{A}$ that leaves $\omega$ invariant, namely such that $\omega\circ \Theta=\omega$, can be implemented by a unitary operator
\begin{equation}
	\mathbb{U}_\omega[\pi_\omega(X)]:=U^\dagger_\omega\,\pi_\omega(X) \,U_\omega=\pi_\omega(\Theta(X))\,, \quad X\in \mathcal{A}\,,
\end{equation}
where $U_\omega^\dagger$ is defined on the dense subspace $\pi_\omega(\mc{A})\ket{\Omega_\omega}$ as
\begin{equation}\label{unitaryimplementation_d}
	U_\omega^\dagger \,\pi_\omega(X)\ket{\Omega_\omega}:=\pi_\omega\circ\Theta(X)\ket{\Omega_\omega}\,,\quad  
	U_\omega^\dagger \ket{\Omega_\omega}=\ket{\Omega_\omega}\,.
\end{equation}
\begin{example}\label{ex:GNS_finite}
	Finite $d$-level systems are described by the $C^*$-algebra of $d\times d$ matrices $X\in M_d(\mathbb{C})$ (endowed with the matrix norm) and by states given by density matrices
	\begin{equation}
		\omega(X)=\Tr(\rho\,X)\,, \qquad \Md{d}\,\ni\,\rho\ge0\,,\ \Tr(\rho)=1\,.
	\end{equation}
	For faithful density matrices $\rho=\sum_\alpha r_\alpha \ketbra{\alpha}$, $0<r_\alpha<1$, the GNS construction is achieved by state purification. Indeed, through the Hilbert-Schmidt scalar product $\Tr(X^\dagger Y)=\braket{X}{Y}$ operators $X\in\Md{d}$ are readily identified as vectors in $\mathbb{C}^d\otimes\mathbb{C}^d$. One then introduces the following representations of $M_d(\mathbb{C})$:
	\begin{equation}
		\pi(A)\ket{B}=\ket{AB}\,, \qquad \pi'(A)\ket{B}=\ket{BA}\ .
	\end{equation} 
	By vectorizing matrices $A=\sum_{\alpha\beta} A_{\alpha\beta} \ketbra{\alpha}{\beta}\in M_d(\mathbb{C})$ as $\ket{A}=\sum_{\alpha\beta} A_{\alpha \beta}\ket{\alpha\otimes \beta}$ the above ones become factor representations:
	\begin{equation}
		\pi(A)=A\otimes\mathds{1}_d\,,\qquad\pi'(A)=\mathds{1}_d\otimes\overline{A}\,,\qquad A\in\Md{d}\,,
	\end{equation}
	where $\braket*{\alpha}{\overline{A}|\beta}\equiv\overline{\braket{\alpha}{A|\beta}}$.
	For faithful density matrices $\rho>0$, we have
	\begin{equation}
		\Tr(\rho \, A)=\Tr(\sqrt{\rho}\,A\,\sqrt{\rho})=\bra{\sqrt{\rho}}{\pi(A)}\ket{\sqrt{\rho}}\ ,
	\end{equation}
	where the purification of $\rho$,
	\begin{equation}
		\ket{\sqrt{\rho}}=\sum_{\alpha} \sqrt{r_\alpha}\ket{\alpha\otimes \alpha} \in \mathbb{C}^d\otimes \mathbb{C}^d\,,
	\end{equation} 
	is a cyclic vector for $\pi(\mathcal{A})$. The latter is also separating, namely, $\Tr(\rho\,X^\dag X)=0$ if and only if $X=0$, thus, equivalently, cyclic for the commutant $\pi'(\mathcal{A})$.
	Let the automorphism of $\Md{d}$ that leaves $\rho$ invariant  be $\Theta(A)=V^\dagger A V$ with dual map $\Theta^\ddag$ such that 
	$\Theta^\ddag[\rho]=V \rho\, V^{\dagger}=\rho$. 
	Then its GNS implementation reads 
	$ U_\rho=V\otimes \overline{V}$. 
	Indeed, as $V$ commutes with $\rho$,
	\begin{align*}\nonumber
		&U_\rho^\dagger	\pi_\rho(A)\ket{\sqrt{\rho}}=\sum_\alpha\sqrt{r_\alpha}\,V^\dag\,A\,\ket{\alpha}\otimes\,\overline{V}^\dag\ket{\alpha}
		\\&=\sum_{\alpha,\beta}\sqrt{r_\alpha}\, V^\dagger A \ket{\alpha}\bra{\alpha}V\ket{\beta}\otimes\ket{\beta}\\
		&=\sum_{\beta}V^\dagger A \sqrt{\rho}\,V\ket{\beta}\otimes\ket{\beta}=V^\dag\,A\,V\otimes\mathds{1}_d\ket{\sqrt{\rho}}\\
		&=\pi_\rho(\Theta(A))\ket{\sqrt{\rho}}\,.
	\end{align*}
	Setting $A=\mathds{1}_d$ yields $U_\rho^\dagger \ket{\sqrt{\rho}}=U_\rho \ket{\sqrt{\rho}}=\ket{\sqrt{\rho}}$. As a concrete instance,  let $V=e^{i H}$ with
	$H=\sum h_\alpha \ket{\alpha}\bra{\alpha}$, that is $H$ and $V$ are diagonal in the eigenbasis $\{\ket{\alpha}\}_\alpha$ of any fixed faithful $\rho$. 
	Then, $\overline{V}=V^\dagger$.
\end{example}
If a time-invariant state $\omega\circ\Theta=\omega$ is considered, the physical interpretation of the ALF entropy 
\eqref{def:dynamicalentropy}
is neatly obtained in the GNS representation as an alternating sequence of measurement processes and reversible one-step dynamics. This possibility already emerged in Remark~\ref{rem:HeisSchr}; however, there the interpretation is only valid for the diagonal entries of $\rho\big[\mc{X}^{(n)}\big]$.
One indeed proves that the von Neumann entropy of the coarse-grained density matrix can  equivalently be computed, in the GNS representation, as~\cite{Alicki2002,AlickiFannesBook},
\begin{equation}S\left(\rho\big[\mc{X}^{(n)}\big]\right)=S\left(\left(\mathbb{U}_\omega^\ddagger\circ\wt{\mathbb{X}}^\ddagger\right)^{n}[\ketbra{\Omega_\omega}]\right),
	\label{equivalenceGNS}
\end{equation}
where  $\mathbb{U}_\omega[\pi_\omega(X)]=U_\omega^\dagger\pi_\omega(X)U_\omega=\pi_\omega(\Theta[X])$ is the unitary implementation of the dynamics in the GNS representation and,  for any given POVM $\mathcal{X}=\{X_a\}_{a=1}^{\abs{\mc{X}}}$,
\begin{equation}
	\label{OPU-map}
	\widetilde{\mathbb{X}}[A]:=\sum_a \pi_\omega(X_a)^\dagger \, A \, \pi_\omega(X_a)\,,\qquad A \in\mathcal{B}(\mathcal{H}_\omega)\ ,
\end{equation} 
while $\mathbb{U}_\omega^\ddagger$ and $\wt{\mathbb{X}}^\ddagger$ are the dual maps of $\mathbb{U}_\omega$ and $\wt{\mathbb{X}}$ in the Schr\"odinger picture.
The derivation of~\eqref{equivalenceGNS} is reported in Appendix~\ref{app:GNSbig} and its structure supports an interpretation of the ALF entropy similar to that of the classical KS entropy that is as the maximal information per time-step extracted from repeated measurements on the system intertwined with the time evolution. Only, measurements and dynamics do not commute in a quantum context.

\begin{remark}\label{rmk:measurement_apparatus}
	The completely positive map associated to an POVM $\mc{X}$ as in~\eqref{OPU-map} can be dilated as to comprehend the measurement apparatus~\cite{MikeandIke}. In such case, the von Neumann entropy of $\rho[\mc{X}]=\Big[\omega(X_b^\dagger X_a)\Big]$ can be also seen as the entropy exchanged with  the apparatus; this puts a bound upon the possible reduction of the ignorance about the initial state $\omega$ due to the measurement~\cite{Lindblad1973,Schumacher96}. Similarly, through the GNS construction,~\eqref{equivalenceGNS} can  also be seen as the entropy  exchanged with the measurement apparatus due to repeated measurements. More details can be found in Appendix~\ref{app:GNSbig}.
\end{remark}

An analogous operational interpretation can be developed for the open-system ALF entropy~\eqref{ALFopen}. 
When the state $\omega_S$ corresponds to a  faithful density matrix $\rho_S>0$, the GNS representation associated with the state
$\omega_{SE}=\omega_{S}\otimes\omega_E$ is the tensor product
$\left(\mathcal{H}_S\otimes\mathcal{H}_E,\pi_S\otimes\pi_E,\left|\sqrt{\rho_S}\otimes\Omega_E\right\rangle\right)$ of the GNS representations of $S$ and $E$, respectively~\cite[Theorem~6.4.7]{MurphyCstar}. Here, $\ket{\sqrt{\rho_S}}$ is the canonical purification of $\rho_S$ as in Example~\ref{ex:GNS_finite}. Assuming a semigroup global evolution $\Theta_n=\Theta^n$ as in the previous Section \ref{sec:open-system-alf-entropy-and-qr-markovianity}, with $\Theta$ leaving the state $\omega_S\otimes\omega_E$ invariant, we can implement the joint $S+E$ evolution by a unitary
\begin{gather}\nonumber
	\hskip-3cm	 U_\Theta^{\dagger}\pi_S\otimes\pi_E(X_S\otimes X_E)\ket{\sqrt{\rho_S}\otimes\Omega_E}\\ :=\pi_S\otimes\pi_E(\Theta[X_S\otimes X_E])\ket{\sqrt{\rho_S}\otimes\Omega_E},\nonumber\\
	\mathbb{U}_{\Theta}[\pi_S\otimes\pi_E(X_{S}\otimes X_E)]
	:=U_{\Theta}^\dagger\pi_S\otimes\pi_E(X_{S}\otimes X_E) U_{\Theta}\nonumber\\=\pi_S\otimes\pi_E(\Theta[X_{S}\otimes X_E])\,,\nonumber\\ U_{\Theta}\ket{\sqrt{\rho_S}\otimes\Omega_E}=\ket{\sqrt{\rho_S}\otimes\Omega_E}\ .
	\label{SE_unitary}
\end{gather}
Then, by~\eqref{equivalenceGNS}, the entropy of $\rho_S\big[\mc{X}^{(n)}\big]$ turns out to be the same as that of 
\begin{equation}\label{GNS_SE}
	\left(\mathbb{U}_{\Theta}^\ddagger\circ \left(\wt{\mathbb{X}}^\ddagger\otimes\mathrm{id}_{E}\right)\right)^{n}\left[\ketbra{\sqrt{\rho_S}\otimes\Omega_E}\right]\,.
\end{equation}
where   
\begin{equation*}\label{concrete_X}
	\wt{\mathbb{X}}=\mathbb{X}\otimes\mathrm{id}_d	\quad	\hbox{and}\quad
	{\mathbb{X}}[ Y ]=\sum_{a=1}^{\abs{\mc{X}}}  X_a^\dagger\, Y\, X_a\,,\ Y \in\Md{d}\,.
\end{equation*} 
$\mathfrak{h}_S(\Theta)$ is then the  entropy production rate due to the discrete semi-group evolution \eqref{GNS_SE}, that describes repeated measurements on $S$ intertwining the $S+E$ unitary evolution: it quantifies the maximal rate of information about the dynamics that is extractable by measurements on the open system $S$ only. \hfill\break\indent
Moreover, the above discussion justifies interpreting~\eqref{ALFmeasure} as a measure of  back-flow of information that relies on the full multi-time statistics of the open dynamics. Indeed, asymptotic entropy rates quantify the information gained from repeated observations of the open dynamics. If the POVM-dependent rate $\mathfrak{h}(\Theta,\mc{F})$ decreases relative to its value in the QR-Markovian regime, the dynamics becomes less unpredictable, indicating that information has flowed back into the open system.

\section{Dynamical entropy and collisional models}\label{sec:collisionalmodels}

In the following, we exploit the collisional approach to open systems to get more insight regarding the behaviour of $\mathfrak{h}_S(\Theta)$ beyond the Markovian regime, i.e. beyond Quantum Regression. 
We now begin with a general upper bound to the ALF entropy $\mathfrak{h}_S(\Theta)$ of the collisional model discussed in Section~\ref{sec:preliminaries}.

\begin{proposition}\label{prop:upperboundchain}
	Let $\omega_{S}\otimes\omega_E$ be a $\Theta$-invariant state. Then,
	\begin{equation}\label{upperboundchain}
		\mathfrak{h}_S(\Theta)\ \le 
		\ \mathfrak{S}_{\omega_E}\,+\,\log(D)\,,
	\end{equation}
	where $\mathfrak{S}_{\omega_E}$
	is the mean von Neumann entropy of the chain~\eqref{meanentropy}.
\end{proposition}
\begin{proof}
	Consider a POVM involving operators of the system $S$ only, as in~\eqref{system_cmatrix}:
	\begin{equation*}
		\mathcal{X}=\left\{{X}_a\otimes \mathds{1}_E\right\}_a, \qquad 
		X_a \in \mathcal{A}_S,\qquad \sum_{a=1}^d X_a^\dagger X_a=\mathds{1}_d\,.
	\end{equation*}
	After $n$ steps of the dynamics the algebra of the system is mapped into the local algebras $\mathcal{A}_S\otimes\mathcal{A}_E^{[1,n]}$, due to the action of the shift (see~\eqref{reduceddynamicsn}): $\Theta^n(\mathcal{A}_S\otimes\mathds{1}_E)\ \subseteq\ \mathcal{A}_S\otimes\mathcal{A}_E^{[1,n]}$.
	With respect to a $\Theta$-invariant state, the correlation matrix reads:
	\begin{align*}
		\left(\rho\left[\mc{X}^{(n)}\right]\right)_{\bs{a},\bs{b}}&=\omega_{SE}\left({X}_{\bs{b}}^{(n)\dagger}\,{X}_{\bs{a}}^{(n)}\right)\\&=\Tr(\Omega_{S[1,n]}\,
		{X}_{\bs{b}}^{(n)\dagger}\,{X}_{\bs{a}}^{(n)})\,,
	\end{align*}
	with $\Omega_{S[1,n]}=\rho_S\otimes\rho_E^{[1,n]}\in M_d(\mathbb{C})\otimes M_D^{[1,n]}(\mathbb{C})$. Hence, 
	from Proposition~\ref{prop:finitelevel}, together with the fact that $S(\Omega_{S[1,n]})=S(\rho_S)+S(\rho_E^{[1,n]})$ and
	$$
	S\left(\sum_{\bs{a}}{X}_{\bs{a}}^{(n)}\,\Omega_{S[1,n]}\,{X}_{\bs{a}}^{(n)\dagger}\right) \leq
	\log(dD^n)\ ,
	$$
	the following upper bound ensues:
	\begin{align*}
		S(\rho[\mc{X}^{(n)}]) \le S(\rho_S)+S\big(\rho_E^{[1,n]}\big)+\log(d)+n\log(D)\ .
	\end{align*}
	Dividing both sides by $n$ and taking the $\limsup$, one gets inequality~\eqref{upperboundchain}.
\end{proof}

\begin{remark}\label{rmk:mutual_corr}
	The so-called quantum mutual information, 
	\begin{align*}
		I_{n}(\omega_E):=S(\rho_E^{[1,n]})+S(\rho_E^{\{1\}})-S(\rho_E^{[1,n+1]})
	\end{align*}
	quantifies the correlations between the first $n$ sites of the spin chain and the $n+1$-th. From~\eqref{meanentropy}, the large $n$ limit yields
	\begin{equation}\label{asympt_corr}
		\mathfrak{I}_{\omega_E}:=\lim_n I_{n}(\omega_E)=S(\rho_E^{\{1\}})-\mathfrak{S}_{\omega_E} \ge 0 \,,
	\end{equation}
	the inequality saturating for a Bernoulli source. Thus, the stronger the correlations in the environment,  the lower results its mean entropy  and thus the  tighter becomes the bound in~\eqref{upperboundchain} on the maximal entropy production of the open system $S$. Notice that the r.h.s. of bound~\eqref{upperboundchain} amounts to the ALF entropy of the shift on a quantum spin chain as proved in~\cite{AF1994defining}. However that quantum bound is obtained using all possible POVMs of the chain, whereas here we restricted to those of the open system appended to the chain. It is an open question whether for special system POVMs bound~\eqref{upperboundchain} could be saturated.
\end{remark}

\subsection{Classical spin chain environment}
In this section, we compute the open-system ALF entropy $\mathfrak{h}_S(\Theta)$ for the collisional model studied in~\cite{FBGN_PhysicaScr} which in its simplicity exhibits  different types of memory effects directly related to the “microscopic” parameters of the environment $E$. The latter is taken to be a classical stationary source; namely, each site of the chain carries  a same commutative, that is diagonal, matrix algebra $\mathcal{A}=D_D(\mathbb{C})$, spanned by $1$-dimensional orthogonal projections $\{\Pi_k\}_{k=0}^{D-1}$, $\sum_{k=0}^{D-1}\Pi_k=\mathds{1}_D$.
Furthermore, the chain is endowed  with a shift-invariant state $\omega_E$ whose restrictions to any local algebra $\mathcal{A}_E^{[1,2]}$ are diagonal density matrices: 
\begin{equation}
	\label{ergodic_chainstate}
	\rho_E^{[a,b]}=\sum_{\bs{i}_{[a,b]}}p_{\bs{i}_{[a,b]}}\,\Pi_{\bs{i}_{[a,b]}}^{[a,b]}\, ,\quad \Pi^{[a,b]}_{\bs{i}_{[a,b]}}=\bigotimes_{k=a}^{b} \Pi_{i_k}^{(k)}\in\mathcal{A}_E^{[a,b]},
\end{equation} 
with the family of  probabilities $\pi_{[1,n]}=\{p_{\bs{i}_{[1,n]}}\}_{\bs{i}_{[1,n]}}$ satisfying the consistency and stationarity conditions 
\begin{equation*}
	\sum_{i_b} p_{\bs{i}_{[a,b]}}=p_{\bs{i}_{[a,b-1]}}\quad\hbox{and}\quad \sum_{i_a} p_{\bs{i}_{[a,b]}}=p_{\bs{i}_{[a-1,b]}}\, .
\end{equation*}
A finite $d$-level open quantum system is coupled to such a classical environment by means of an interaction with the chain  $0$-th site via a 
controlled-unitary type,
\begin{align}
	\label{collision}
	\Phi[X_S\otimes A_{i_0}^{(0)}]&=\sum_{k=0}^{D-1}\phi_{k}[X_S]\otimes \Pi_k^{(0)} A_{i_0}^{(0)} \Pi_k^{(0)}\\
	&\hskip-1.2cm=\sum_{k=0}^{D-1}\phi_{k}[X_S]\otimes A_{i_0}^{(0)} \Pi_k^{(0)}, \quad \phi_k[X_S]:=U_k^\dagger\, X_S\, U_k	\nonumber
\end{align}
where  $U_k U_k^\dagger=U_k^\dagger U_k=\mathds{1}_d$ so that the dual maps on the states of $S$, $\phi^{\ddagger}_{k}[\rho_S]=U_k\, \rho_S \,U^\dagger_k=\phi^{-1}_{k}[\rho_S]$, coincide with the inverse  maps  as well as the the dual map of $\Phi^{\ddag}$ extended to $\mathcal{A}_S\otimes\mathcal{A}_E$:
\begin{align*}
	\Phi^{\ddag}[X_S\otimes A_{i_0}^{(0)}]&=\sum_{k=0}^{D-1}\phi^{\ddagger}_{k}[X_S] \otimes \Pi_k^{(0)} A_{i_0}^{(0)} \Pi_k^{(0)}\\
	&= \Phi^{-1}[X_S\otimes A_{i_0}^{(0)}]\,.
\end{align*}
Finally, the compound $SE$-automorphic dynamics is given by powers of $\Theta={\rm \Theta_{coll}}=({\rm id}_S\otimes\sigma_E)\circ\Phi$ obtained by acting with the chain right-shift after the interaction of the system $S$ with the $0$-th site of the classical chain (see~\eqref{reduceddynamicsn}).

We then take as state of the system $S$ the maximally mixed state $\omega_S(X_S)={\rm Tr}(\rho_S X_S)$, where $\rho_S = \mathds{1}_d/d$, thus enforcing  the time-invariance of $\omega_S\otimes\omega_E$. Indeed, leveraging the shift-invariance of $\omega_E$ and the commutativity of the algebra $\mathcal{A}_E$, one gets, for all $X_S\in\mathcal{A}_S$ and $X_E\in\mathcal{A}_E$,
\begin{align*}
	&\omega_S\otimes\omega_E(\Theta(X_S\otimes X_E))\\
	&=\frac{1}{d}\sum_{k=0}^{D-1} \Tr(\phi_k[X_S]) \,\omega_E\left(\sigma_E\left(\Pi_k^{(0)}\,X_E\, \Pi_k^{(0)}\right)\right)\\
	&=\frac{\Tr(X_S) }{d} \omega_E\!\left(\sum_k\Pi_k^{(0)}X_E \Pi_k^{(0)}\right)=
	\omega_S\otimes\omega_E(X_S\otimes X_E).
\end{align*}
\begin{remark}\label{rmk:SBFI}
The classically controlled coupling~\eqref{collision} is still actively debated in the literature~\cite{BAV_2018,buscemi2025causal,backer_localdisclosure}. For sure, despite the memory-storing mechanism of this model being entirely classical, it can give rise to memory effects in the open quantum dynamics that have no classical counterpart. This is the case of the so-called super-activation of back-flow of information~\cite{FBGN_PhysicaScr}, occurring in the tensor product dynamics of a bipartite system, each copy undergoing the same dissipative dynamics. Such phenomenon  occurs when the reduced dynamics is positive but not completely positive-divisible. 
\end{remark}
\subsection{Structure of the coarse grained density matrix}
The $n+1$-th step coarse-grained density matrix based upon a system $S$ POVM  ${\mc{X}}=\left\{X_a\otimes\mathds{1}_E\right\}_{a=1}^{\abs{\mathcal{X}}}$, $X_a\in\mathcal{A}_S=M_d(\mathbb{C})$, and the $\Theta$-invariant state $\omega_{SE}=\omega_S\otimes\omega_E$ takes the form 
\begin{align}
	\rho_S\left[{\mc{X}}^{(n+1)}\right]\ =\ \sum_{\bs{i}_{[1,n]}} \, p_{\bs{i}_{[1,n]}}\ \rho_S\left[\mc{X}_{\bs{i}_{[1,n]}}\right]\,,
	\label{cmatrix_mixture}
\end{align}
where $\bs{i}_{[1,n]}=i_1\ldots i_{n}$; namely, of a convex combination of sub-coarse-grained matrices $\rho_S\left[\mc{X}_{\bs{i}_{[1,n]}}\right]$ with entries 
\begin{align}	
	&\left(\rho_S\left[\mc{X}_{\bs{i}_{[1,n]}}\right]\right)_{\bs{a},\bs{b}} \nonumber\\
	&=\omega_S\bigg(X_{b_0}^\dagger \phi_{\bs{i}_{\{1\}}}[X_{b_1}^\dagger]\phi_{\bs{i}_{[1,2]}}[X_{b_2}^\dagger]\ldots	\nonumber\\
	&\hskip+1cm\ldots\phi_{\bs{i}_{[1,n]}}[X_{b_{n}}^\dagger X_{a_{n}}^{\phantom{\dagger}}]\ldots \phi_{\bs{i}_{[1,2]}}[X_{a_1}]\phi_{\bs{i}_{\{1\}}} [X_{a_1}] X_{a_0}\bigg)
	\nonumber	
	\\
	&= \omega_S\Bigg(X_{b_0}^\dagger \phi_{i_1}\bigg[X_{b_1}^\dagger \phi_{i_2}\Big[X_{b_2^\dagger}\ldots \nonumber\\
	&\hskip+2cm \ldots\phi_{i_n}\left[X_{b_{n}}^\dagger X_{a_{n}}^{\phantom{\dagger}}\right]\ldots X_{a_2}\Big] X_{a_1}\bigg] X_{a_0}\Bigg)\label{cmatrix_mixtureelements2}
\end{align}
specified by the automorphisms $\phi_{\bs{i}_{[1,k]}}=\phi_{i_1}\circ\,\cdots\,\circ\,\phi_{i_k}$ and by the $n$-step POVM $\mc{X}_{\bs{i}_{[1,n]}}$ 
from $\mathcal{A}_S$ with components
$X_{\bs{i}_{[1,n]}}^{\bs{a}}:=\phi_{\bs{i}_{[1,n]}}[X_{a_{n}}]\cdots\phi_{\bs{i}_1} [X_{a_1}] X_{a_0}\in\mathcal{A}_S$.
\begin{proof}
	Consider $X_{a_1}\in\mathcal{A}_S$; then,
	\begin{eqnarray*}
		\Theta[X_{a_1}\otimes\mathds{1}_E]&=&\sum_{k_1}\phi_{k_1} [X_{a_1}]\otimes \Pi_{k_1}^{(1)}\ ,\\
		\Theta^2[X_{a_2}\otimes\mathds{1}_E]&=&\sum_{k_1,k_2}\phi_{k_1}\circ\phi_{k_2} [X_{a_2}]\otimes \Pi_{k_1}^{(1)}\otimes\Pi_{k_2}^{(2)}\,,
	\end{eqnarray*}
	so that the $2$-step POVM  consists of operators
	\begin{multline*}
		\Theta^2[X_{a_2}\otimes\mathds{1}_E]\Theta[X_{a_1}\otimes\mathds{1}_E]X_{a_0}\otimes\mathds{1}_E\\=\sum_{k_1,k_2}\phi_{k_1}\circ\phi_{k_2} [X_{a_2}]\phi_{k_1}[X_{a_1}]X_{a_0}\otimes \Pi_{k_1}^{(1)}\otimes\Pi_{k_2}^{(2)}\,.
	\end{multline*}
	Thus, one has
	\begin{align}\nonumber
	{X}_{\bs{a}}^{(n+1)}&=\Theta^{n}[X_{a_{n}}\otimes\mathds{1}_E]\ldots\Theta[X_{a_1}\otimes\mathds{1}_E]X_{a_0}\otimes\mathds{1}_E\\
		&\hskip-1cm=\sum_{\bs{k}_{[1,n]}}\phi_{\bs{k}_{[1,n]}}[X_{a_n}]\phi_{\bs{k}_{[1,n-1]}} [X_{a_{n-1}}]\ldots\phi_{\bs{k}_{\{1\}}}[X_{a_1}]X_{a_0}\otimes \nonumber\\
		&\hskip+5cm \otimes\bigotimes_{j=1}^n \Pi_{k_j}^{(j)}\,.\nonumber
	\end{align}
	The elements of the coarse-grained density matrix then become
	\begin{align*}
		\omega_{SE}\left({X}_{\bs{b}}^{(n
			+1)\dagger} {X}_{\bs{a}}^{(n+1)}\right)&=\sum_{\bs{k}_{[1,n]}} p_{\bs{k}_{[1,n]}}\, \omega_S\bigg(X_{b_0}^\dagger \phi_{k_1}[X_{b_1}^\dagger]\ldots\\ 
			&\hskip-1cm\ldots\phi_{\bs{k}_{[1,n]}}[X_{b_{n}}^\dagger X_{a_{n}}^{\phantom{\dagger}}]\ldots\phi_{{k}_{1}} [X_{a_1}] X_{a_0}\bigg)\,.
	\end{align*}
	Given a fixed multi-index $\bs{k}_{[1,n]}$, the corresponding operators
	${X_{\bs{k}_{[1,n]}}^{\bs{a}}:=\phi_{\bs{k}_{[1,n]}}[X_{a_{n}}]\cdots\phi_{\bs{k}_{[1,2]}}[X_{a_2}]\phi_{{k}_{1}} [X_{a_1}] X_{a_0}\in\mathcal{A}_S}$ form a POVM. Indeed, the maps $\phi_{\bs{k}_{[1,n-1]}}$ are automorphisms; namely, they are invertible and  such that $\phi_{\bs{k}_{[1,n]}}[XY]=\phi_{\bs{k}_{[1,n]}}[X]\phi_{\bs{k}_{[1,n]}}[Y]$ and $\phi_{\bs{k}_{[1,n]}}[X^\dag]=
	(\phi_{\bs{k}_{[1,n]}}[X])^\dag$. In the case $k=2$, for example,
	\begin{align*}
		&\sum_{\bs{a}} \left(X_{\bs{k}_{[1,2]}}^{\bs{a}}\right)^\dag  X_{\bs{k}_{[1,2]}}^{\bs{a}}\\&=\sum_{a_0,a_1}  X_{a_0}^\dagger\phi_{k_1} [X_{a_1}^\dagger] \phi_{\bs{k}_{[1,2]}} \left[\sum_{a_2} X_{a_2}^\dagger X_{a_2}\right] \phi_{{k}_{1}}[X_{a_1}]X_{a_0}\\
		&=\sum_{a_0} X_{a_0}^\dagger  \phi_{{k}_{1}}\left[\sum_{a_1}X_{a_1}^\dagger X_{a_1}\right] X_{a_0}=\sum_{a_0} X_{a_0}^\dagger X_{a_0}=\mathds{1}_S\,.\qedhere
	\end{align*}
\end{proof}

\begin{remark}
	\label{rmk:concrete_QR}
	Similarly, the reduced Heisenberg dynamics 
	$\Lambda_n[X_S]=\mathrm{id}_S\otimes\omega_E(\Theta^n[X_S\otimes\mathrm{1}_E])$ is implemented by the CPU map
	\begin{equation}\label{CPUdynamical_map}
		\Md{d}\,\ni X \ \longmapsto\ \Lambda_n[X]\,=\,\sum_{\bs{k}_{[1,n]}}\,p_{\bs{k}_{[1,n]}}\ \phi_{\bs{k}_{[1,n]}}[X]\,.
	\end{equation}
	If the classical chain represents a Bernoulli process, then
	$p_{\bs{k}_{[1,n]}}=\prod_{j=1}^n\, p_{k_j}$.
	The collisional environment is fully uncorrelated and the reduced dynamics is a semigroup, $\Lambda_n=\Lambda_1^n$, that is the reduced dynamics at time-step $n$ is the $n$-th power of the single time-step reduced dynamics $\Lambda_1=\sum_k p_k\, \phi_k$. Further, QR  holds. For example, after two iterations of the dynamics,
	\begin{align*}
		&\left( \rho_S\big[\mc{X}^{(2)}\big]\right)_{\bs{a},\bs{b}}\\&=\sum_{k_2,k_1} p_{k_1}p_{k_2} \ \omega_S\left(X_{b_0}^\dagger \phi_{k_1}\left[X_{b_1}^\dagger \phi_{k_2}[X_{b_2}^\dagger X_{a_2}] X_{a_1}\right]X_{a_0}\right)\\
		&=\omega_S\left(X_{b_0}^\dagger \sum_{k_1}p_{k_1}\phi_{k_1}\Bigg[X_{b_1}^\dagger \sum_{k_2}p_{k_2}\phi_{k_2} \left[X_{b_2}^\dagger X_{a_2}\right]X_{a_1}\Bigg]\right)\\&=\omega_S\left(X_{b_0}^\dagger \Lambda_1\left[ X_{b_1}^\dagger \Lambda_1\left[X_{b_2}^\dagger X_{a_2}\right]X_{a_1}\right]X_{a_0}\right)\,,
	\end{align*}
	so that the knowledge of $\Lambda_1$ suffices to reconstruct the full coarse-grained density matrix. 
\end{remark}

\subsection{Upper and lower bounds to the ALF entropy} 
Given the convex combination~\eqref{cmatrix_mixture}, its von Neumann entropy can then be bounded from above as:
\begin{align}\nonumber
	S\Big(\rho_{S}\left[{\mc{X}}^{(n+1)}\right]\Big)&\le \ H\left(\pi_{[1,n]}\right) +\sum_{\bs{i}_{[1,n]}} p_{\bs{i}_{[1,n]}} S\left(\rho_S\left[\mc{X}_{\bs{i}_{[1,n]}}\right]\right)\nonumber\\
	& \le\ H\left(\pi_{[1,n]}\right)\, +2\,\log(d)\,,
	\label{bound_mixture}
\end{align}
where we used the well-known inequality for the entropy of a convex mixture~\cite[Theorem 11.10]{MikeandIke}, along with Proposition~\ref{prop:finitelevel} and  the fact that, for all states of $d$-level systems, $S(\rho)\leq\log d$.
Dividing both sides of~\eqref{bound_mixture} by $n$ and taking the $\limsup$, one gets that the dynamical entropy of the open system is upper-bounded by the 
entropy rate of the chain,
\begin{equation}\label{deupperbound}
	\mathfrak{h}_S(\Theta)\
	\le\ \lim_n \,\frac{1}{n}H\left(\pi_{[1,n]}\right) =\ \mathfrak{S}_{\omega_E}\,,
\end{equation}
We now show that inequality~\eqref{deupperbound} can be saturated. 

\begin{proposition}\label{prop:model_Tn}
	Consider a $d$-level system collisionally interacting with a classical stationary spin chain through~\eqref{collision}.  With respect to the invariant state $\mathds{1}_d / d \otimes \omega_E$ and the special POVM $\mathcal{F}$ in~\eqref{specialOPU}, the coarse-grained density matrix takes the form
	\begin{align}\label{specialpartition_model_prop}
		\rho_S\left[\mathcal{F}^{(n+1)}\right]\\& \hskip-1.4cm=\frac{\mathds{1}_d}{d}\otimes\frac{\mathds{1}_d}{d}\otimes\left( \mathbb{T}_{n}^\ddagger\otimes\mathrm{id}_{d}^{\otimes n}\left[\ketbra{\sqrt{\frac{\mathds{1}_d}{d}^{\otimes n}}}\right]\right)\,, \nonumber
		\\
		&\hskip-2.1cm\hbox{ with correlation maps}\nonumber\\
		 & \hskip-1.34cm \mathbb{T}_n\ =
		\ \sum_{\bs{i}_{[1,n]}} \,p_{\bs{i}_{[1,n]}} \,\phi_{\bs{i}_{[1,n]}}^{\otimes [1,n]} \, , \quad \phi_{\bs{i}_{[1,n]}}^{\otimes [1,n]}:=\bigotimes_{k=1}^n \phi_{i_k}\, ,\label{map_randomunitary_prop}
	\end{align}
	Moreover, if the unitaries in~\eqref{collision}
	are such that $$\Tr({U_{j}^\dagger}{U_{k}})=d\,  \delta_{jk}\,,$$ one has 
	\begin{equation}\label{statement_whenequalityholds}
		\mathfrak{h}_{S}(\Theta)=\mathfrak{h}_S(\Theta,\mc{F})=\lim_{n}\frac{1}{n}H\big(\pi_{[1,n]}\big)=\mathfrak{S}_{\omega_E}\,,
	\end{equation}
	that is, the ALF entropy of the system is equal to the chain entropy rate.
\end{proposition}
The proof is reported in Appendix~\ref{app:model_Tn}. There, it is shown that the saturation of the bound~\eqref{deupperbound} relies on
 the condition $\Tr({U_{j}^\dagger}{U_{k}})=d\,\delta_{jk}$. If the latter is not satisfied, the POVM $\mathcal{F}$ is not guaranteed to be optimal. 	A natural example is obtained for $D=d^2$ and by choosing  $U_{i_k}$ as the discrete Weyl operators,
	\begin{equation}
		U_{a,b}=\sum_{k=0}^{d-1}\omega^{k b} \ketbra{a+k}{k}\,, \qquad \omega=e^{2\pi i/d}\,,
	\end{equation}
	which are unitary and satisfy  $\Tr(U_{a,b}^\dagger U_{c,d}^{\phantom{\ddagger}})=d\, \delta_{a c} \delta_{bd}$.  
	\color{black}  

	\begin{remark}\label{rmk:clarify_saturation}
	Suppose that $\Tr({U_{j}^\dagger}{U_{k}})=d\, \delta_{jk}$, $D=d^2$. As shown in Remark~\ref{rmk:mutual_corr},
	$$\mathfrak{h}_S(\Theta)=\mathfrak{S}_{\omega_E}=S(\rho_E^{(1)})-\mathfrak{I}_{\omega_E}\,.$$
	If the environment is a Bernoulli source, one  has simply
	\begin{equation*}
		\mathfrak{h}_S(\Theta)=S(\rho_E^{(1)})=H(\pi_{\{1\}})\,.
	\end{equation*}
	Note also that
	$\ket{\sqrt{\rho_S}}=\ket{\sqrt{\mathds{1}_d/d}}=d^{-1/2}\sum_i\ket{i\otimes i}=:\ket{\psiplus{d}}$, namely the fully symmetric maximally entangled state. Then, 
	\begin{equation*}
		\Lambda_1^\ddag\otimes\mathrm{id}_d\left[\ketbra{\psiplus{d}}\right]=\sum_i p_i \, U_i\otimes \mathds{1}_d\ketbra{\psiplus{d}} U_i^\dagger\otimes \mathds{1}_d.
	\end{equation*}
	Since $U_i\otimes \mathds{1}_d\ket{\psiplus{d}}$, $i=1, \dots, d^2$ form a basis of $\mathbb{C}^d\otimes \mathbb{C}^d$, the  von Neumann entropy of the latter state amounts to $H(\pi_{\{1\}})=-\sum_i p_i \log p_i$  and thus, consistently, equals to  $\mathfrak{h}_S(\Theta)$ since QR holds (compare Corollary~\ref{cor:partition}).
\end{remark}

\subsection{Markov chain environment}\label{sec:Markov}
We now let the classical environment be a stationary Markov chain; namely, we choose the probabilities $\pi_{[1,n]}=\big\{p_{\bs{i}_{[1,n]}}\big\}_{\bs{i}_{[1,n]}}$ to form a stationary Markov process,
\begin{equation}
	\label{Markovprocess}
	p_{\bs{i}_{[1,n]}} = \prod_{k=2}^{n} \,T_{i_{k}\,i_{k-1}}\ p_{i_1}\,,
\end{equation}
defined by the stochastic matrix
$T=[T_{ij}]$, $T_{ij}\ge0$, with
$ \sum_i T_{ij}=1$, $\sum_j T_{ij} \,p_j=p_i$. The mean entropy rate of the Markov source is equal to the two-site conditional entropy of the 
chain~\cite{Walters1982},
\begin{equation}
	\label{MarkovEnt}
	\mathfrak{S}_{\omega_E}=-\sum_{ij}p_j \, T_{ij} \log(T_{ij})=H(\pi_1)-I(\pi_1; \pi_2)\ ,
\end{equation}
where
\begin{equation}
	\label{condent}
	I(\pi_1;\pi_2)=H({\pi_1})+H(\pi_2)\,-\,H(\pi_{[1,2]})\,,
\end{equation}
denotes the mutual information between the first two subsequent sites of the classical chain. Because of stationarity, the  latter quantity does not depend on where the chosen pair is located along the chain and measures the correlations between any two subsequent sites. Then, one sees that the entropy rate decreases with increasing correlations between subsequent sites. 
In the following, we construct a process with $U_k$ as in~Remark~\ref{rmk:clarify_saturation} and for which $T_{ij}\in\{0,1\}$, so that 
$\mathfrak{h}_S(\Theta)=\mathfrak{S}_{\omega_E}=0$ as for reversible time-evolutions. 

\subsubsection{Pauli dynamics: divisibility.}
\label{ex:Markov1}
We now let the open quantum system $S$ to be a qubit ($d=2$), we couple it to a Markov chain with $4\times 4$ diagonal matrices at each site and choose the
local probability distributions to be generated by the following $4\times 4$ stochastic matrix
\begin{gather}\label{tmatrix}
	T=\begin{pmatrix}
		p_0 & p_0 & p_0 & p_0 \\
		p   & p+\Delta& p-\Delta & p\\
		p   & p-\Delta & p+\Delta & p\\ 
		r        & r& r     & r \\ 
	\end{pmatrix}, \\  0\le\Delta\le p\leq \frac{1}{2}\,,\quad  p_0+2p+r=1\,,\nonumber
\end{gather}
with invariant probability vector $\bs{p}=(p_0,p,p,r)$,
guaranteeing stationarity $T\bs{p}=\bs{p}$. Let the coupling be as in~\eqref{collision} with
\begin{equation}
	\phi_k[X]=\sigma_k\, X\, \sigma_k\,, \qquad k=0,\dots,3\,,
\end{equation}
 $\sigma_0=\mathds{1}_2$ and $\sigma_k, k=1,2,3$ being the Pauli matrices so that the reduced dynamics corresponds to a Pauli map
\begin{equation}\label{1_qubitPaulikraus}
	\Lambda_n^\ddagger=\sum_{\bs{i}_{[1,n]}} p_{\bs{i}_{[1,n]}} \, \phi_{\bs{i}_{[1,n]}}=\Lambda_n \,, \quad \phi_{\bs{k}_{[1,n]}}:= \phi_{k_1}\circ\cdots\circ\phi_{k_n}\,.
\end{equation}
The latter can be explicitly computed for the Markov chain generated by~\eqref{tmatrix}. Indeed, for all $n\in\mathbb{N}$,~\eqref{1_qubitPaulikraus} is diagonal in the Pauli basis,  
\begin{equation}\label{1qubit_eigen}
	\Lambda_n^\ddagger[\sigma_\alpha]=\lambda_n^{(\alpha)}\sigma_\alpha\,, \qquad \alpha=0,\dots,3\,,
\end{equation}
where $\lambda_n^{(0)} = 1$ follows from the preservation of unity and the other eigenvalues evolving according to recursive equations,
\begin{align}\label{onequbitrecursion}
	\lambda_n^{(\alpha)}&=\left(p_0+p(\mu_1^{(\alpha)}+\mu_2^{(\alpha)})+r\mu_3^{(\alpha)}\right) \lambda^{(\alpha)}_{n-1}\\
	&+ \Delta \, p\left(\mu_1^{(\alpha)}-\mu_2^{(\alpha)}\right)^2 \lambda^{(\alpha)}_{n-2}\,,\quad \alpha=0,\dots,3\,, \nonumber
\end{align}
where $\phi_\beta[\sigma_\alpha]=\mu_\beta^{(\alpha)}\sigma_\alpha$, $\mu_0^{(\alpha)}=\mu_\beta^{(0)}=1$ and $\mu_\beta^{(\alpha)}=-1$ for $\beta\ne \alpha\in\{1,2,3\}$. Accordingly,~\eqref{onequbitrecursion} are rewritten as
\begin{align*}
	\lambda_n^{(1,2)}&=\left(1-2(p+r)\right)\ \lambda_{n-1}^{(1,2)}
	+ 4\,  p\, \Delta\ \lambda_{n-2}^{(1,2)}\,,  \\ 
	\lambda_{n}^{(3)}&=\,(1-4p)\ \lambda_{n-1}^{(3)}\,,
\end{align*}
which are easily solvable, yielding
\begin{align}
	\lambda_n:=\lambda_n^{(1,2)}&
	=\left(\frac{B_{p,r,\Delta}+A_{p,r}}{2\,B_{p,r,\Delta}}\right)\left(\frac{A_{p,r}+B_{p,r,\Delta}}{2}\right)^{n}\nonumber\\&\quad +\left(\frac{B_{p,r,\Delta}-A_{p,r}}{2 \,B_{p,r,\Delta}}\right)\left(\frac{A_{p,r}-B_{p,r,\Delta}}{2}\right)^{n}, \nonumber \\
	\lambda_n^{(3)}&=(1-4p)^n \,,\label{1qubit_recursion_explicit}   
\end{align}
where we set $A_{p,r}=1-2(p+r)$ and $B_{p,r,\Delta}=\sqrt{A_{p,r}^2+16\, p \,\Delta}$.
In~\cite{FBGN_PhysicaScr}, it was proved that, for $A_{p,r}>0$, the divisibility features of the one-qubit dynamical map $\Lambda_n^\ddagger$ are 
controlled by the parameter $\Delta$ as follows:
\begin{align} \nonumber
	\Lambda_n^\ddagger \;\textrm{is CP-divisible iff}& 
	& \frac{\Delta}{A_{p,r}} &\le \frac{r}{2p} \,,\\
	\Lambda_n^\ddagger \otimes \Lambda_n^\ddagger \;\textrm{is P-divisible iff}& & &\nonumber\\
	& & & \hskip-3cm  \frac{\Delta}{A_{p,r}} \le \frac{r}{2p}+\frac{1}{2}-\frac{1-\sqrt{1-4p(1-2p)}}{4p}\,,   	\nonumber \\
	\Lambda_n^\ddagger \;\textrm{is P-divisible iff}&
	& \frac{\Delta}{A_{p,r}} &\le \frac{r}{2p}+\frac{1}{2}\, .\label{conditions_Delta}\,,
\end{align}
\begin{remark}
	Note that we discuss divisibility of the dynamical map only in the Schrödinger picture. Recently, an interesting dependence of the divisibility properties of the dynamics on the chosen picture was discussed in~\cite{SettimoPRXQuantum}. However, for Pauli maps, which are functionally commutative in time, the divisibility regimes in the Schrödinger and Heisenberg pictures coincide.
\end{remark}
Hence, by tuning $\Delta$ and keeping fixed the other parameters,  conditions~\eqref{conditions_Delta} establish that one can change the divisibility degree of the dynamics by increasing the correlations of the environment. 
Indeed, by increasing $\Delta$ one first loses CP-divisibility, then the P-divisibility of the second tensor power $\Lambda_n^\ddagger	\otimes	\Lambda_n^\ddagger$
and finally the P-divisibility of the one qubit dynamics $\{\Lambda_n^\ddagger\}_{n\geq 0}$. Notice that only in the latter case one can detect back-flow of information as distinguishability revivals for one qubit. Indeed, lack of P-divisibility means lack of contractivity of some intertwiner $\Lambda_{n,n-1}^\ddagger$ and thus the possibility that the trace-distances of suitable states increase from time-step $n-1$ to time-step $n$.

\begin{remark}
	Recall that, for continuous-time dynamics, the behavior of integrated measures of non-Markovianity is closely related to their divisibility properties. One of the most widely used measures of non-Markovianity is the Rivas–Huelga–Plenio (RHP) measure~\cite{RHPmeasure},
	\begin{equation}
		\mathcal{N}_{RHP}=\int_0^{\mathrlap{\infty}}\dd{t} g_t,\quad g_t=\lim_{\epsilon\to 0}\frac{\normt{\Lambda^\ddagger_{t+\epsilon,t}\otimes\mathrm{id}_d\left[\pplus{d}\right]}-1}{\epsilon}
		,
	\end{equation}
with $\pplus{d}=\ketbra{\psiplus{d}}$. Then, $\mathcal{N}_{RHP}$ vanishes if and only if the dynamics is CP-divisible. On the other hand, the measure proposed by Breuer-Laine-Piilo (BLP)~\cite{BLP} rather focuses on state-distinguishability revivals
	\begin{equation}
		\mathcal{N}_{BLP}=\max_{\rho,\sigma}\int_{\mathrlap{\delta_t>0}}\, \dd{t} \delta_t(\rho,\sigma),\quad \delta_t(\rho,\sigma)=\dv{t}   \normt{\Lambda_t^\ddagger\left[\frac{\rho-\sigma}{2}\right]}.
	\end{equation}
	Note that  P-divisibility of $\Lambda_t^\ddagger$ yields $\mathcal{N}_{BLP}=0$. Moreover, $\mathcal{N}_{RHP}=0\implies \mathcal{N}_{BLP}=0$, while the converse is obviously not true.
	
\label{rmk:measurecontinuous}
\end{remark}
\subsubsection{Pauli dynamics: non-Markovianity measures.}
We further illustrate that the standard measures of Markovianity of Remark~\ref{rmk:measurecontinuous}, when non-zero, exhibit a monotonic dependence on the chain correlations $\Delta$, analogous to the behavior observed for the divisibility degree in~\eqref{conditions_Delta}. We first need to  suitably adapt the aforementioned measures to our discrete-time scenario. Let us fix $N\in\mathbb{N}$ and consider the following discrete-time generalization of the RHP measure
\begin{equation}\label{RHP_discreteMT}
	\mathcal{N}_{RHP}^{(N)}=\sum_{n=1}^{N} g_n\,, \quad g_n=\normt{\Lambda^{\ddag}_{n,n-1}\otimes\mathrm{id}_d\left[\pplus{d}\right]}-1.
\end{equation}
Letting  $E^{(i)}_{n,n-1}$  be the eigenvalues of $\Lambda^{\ddag}_{n,n-1}\otimes\mathrm{id}_d\left[\pplus{d}\right]$, one has  $\normt{\Lambda^{\ddag}_{n,n-1}\otimes\mathrm{id}_d\left[\pplus{d}\right]}-1=2\sum_{i: \ E_{n,n-1}^{(i)}<0}\abs {E^{(i)}_{n,n-1}}$.
  For the one-qubit Pauli dynamics at hand, with $A_{p,r}>0$, propagators $\Lambda^\ddagger_{n,n-1}$ might develop at most one negative eigenvalue  $E_{n,n-1}^{(3)}=(1-2p-\lambda_{n,n-1})/2$ (see also~\cite{FBGN_PhysicaScr}), where $\lambda_{n,n-1}:={\lambda_{n}}/{\lambda_{n-1}}$; in such case
\begin{equation}\nonumber
	g_n=\lambda_{n,n-1}-(1-2p)\,.
\end{equation}
We then claim that 
\begin{equation}\label{the_derintMT}
	\pdv{g_n}{\Delta} =	\pdv{\lambda_{n,n-1}}{\Delta}\ge0\,, \qquad \forall \, n\in \mathbb{N}\,,
\end{equation}
namely, $g_n$ increases monotonically with $\Delta$, thereby ensuring that $\mathcal{N}_{RHP}^{(N)}$ likewise increases monotonically with the correlations $\Delta$.
 The proof is reported in Appendix~\ref{app:discretized_measures}.
Similarly, for the discrete-time analogue of the BLP measure, we fix $N \in \mathbb{N}$ and consider
\begin{equation}
	\mathcal{N}_{BLP}^{(N)}=\max_{\rho,\,\sigma}\, \sum_{\substack{n=1 \\  d_n-d_{n-1}> 0}}^N \left(d_n(\rho,\sigma)-d_{n-1}(\rho,\sigma) \right)\,, \label{discrete_BLPMT}
\end{equation}
where $\displaystyle d_n(\rho,\sigma):=\normt{\Lambda_n^\ddagger\left[\frac{\rho-\sigma}{2}\right]}$ and the optimization is performed on all pair of density matrices $\rho$, $\sigma$.
Then, one proves that for states $\rho$, $\sigma$ and $n\in \mathbb{N}$ such that $d_n(\rho,\sigma)>d_{n-1}(\rho,\sigma)$,
\begin{equation}
	\pdv{\Delta}\left(d_n(\rho,\sigma)-d_{n-1}(\rho,\sigma) \right)\ge0\,.
\end{equation}
The detailed derivation is again reported in Appendix~\ref{app:discretized_measures}.
As a consequence, both  discretized measures~\eqref{RHP_discreteMT} and~\eqref{discrete_BLPMT} are monotonically increasingly with	the correlations $\Delta$.
\subsubsection{Pauli dynamics: correlation maps and ALF entropy.}
In the context of the qubit dynamics described in Section~\ref{ex:Markov1}, the correlation maps $\mathbb{T}_n: M_2^{\otimes n}(\mathbb{C})\to M_2^{\otimes n}(\mathbb{C}) $, introduced in~\eqref{mapdefinition_enunciato}, are $n$-qubit Pauli channels:
\color{black}
\begin{equation}
	\mathbb{T}_n[X]=\sum_{\bs{i}_{[1,n]}} \,p_{\bs{i}_{[1,n]}} \,\sigma^{[1,n]}_{\bs{i}_{[1,n]}} \, X \,\sigma^{[1,n]}_{\bs{i}_{[1,n]}}\,, 
	\quad \sigma^{[1,n]}_{\bs{i}_{[1,n]}}:=\bigotimes_{k=1}^n \sigma_{i_k}\,,
\end{equation}
so that $\mathbb{T}_n=\mathbb{T}_n^\ddag$. 
The purification of $ \left({\mathds{1}_2}/{2}\right)^{\otimes n} $ corresponds to the maximally entangled state 
 	$$
 	\ket{\psi_+^{{[1,n]}}}=2^{-n/2}\sum_{\bs{k}_{[1,n]}} \ket{\bs{k}_{[1,n]}\otimes\bs{k}_{[1,n]}}\,.
 $$ Set then $\ket{\psi_{\bs{i}_{[1,n]}}}:=\left(\bigotimes_{k=1}^n \sigma_{i_k}\right)\otimes\mathds{1}_2^{\otimes n}\ket{\psi_+^{[1,n]}}$. It is easy to check that such vectors are orthonormal, so that the state
	\begin{equation*}\label{pseudo_spectraltext}
		\mathbb{T}_n^\ddag\otimes\mathrm{id}_{d}^{\otimes n}\left[\ketbra{\psi_+^{[1,n]}}\right]=\sum_{\bs{i}_{[1,n]}}\,p_{\bs{i}_{[1,n]}} \ketbra{\psi_{\bs{i}_{[1,n]}}} 
	\end{equation*}	
	 is diagonal with eigenvalues $\,p_{\bs{i}_{[1,n]}}$ and associated eigenvectors $\ket{\psi_{\bs{i}_{[1,n]}}}$. Then, as in~\eqref{statement_whenequalityholds}, the dynamical entropy of the system equals the entropy rate of the Markov source.
With $T$ parametrized as in~\eqref{tmatrix}, one gets
\begin{multline}\label{specificexample}
	\mathfrak{h}_{S}(\Theta)=\eta(p_0)+\eta(r)\\+2 \left(2p\,  \frac{\eta(p+\Delta)+\eta(p-\Delta)}{2}+(1-2p)\eta(p)\right),
\end{multline}
where $\eta(x)=-x\log x$, $\eta(0):=0$. 
The associated two-site mutual information reads
\begin{align}
	I(\pi_{1};\pi_{2})&=2p\Big(2\eta(p)-\eta(p+\Delta)-\eta(p-\Delta)\Big)\nonumber\\&=4 p^2 \left(\log 2- h\left(\frac{1}{2}+\frac{\Delta}{2p}\right)\right),
	\label{concrete_I}
\end{align}
where $h(x):=\eta(x)+\eta(1-x)$, $0\le x\le 1$ is the Shannon binary entropy, which is a decreasing function of $x$ for $1/2 \le x \le 1$. Therefore, the two-site correlations increase with $\Delta$ and are maximal at $p=\Delta$. Correspondingly, the open-system ALF  entropy~\eqref{specificexample} is a monotonically decreasing function of the chain correlations $\Delta$. 
\par
Notice that the  conditions~\eqref{conditions_Delta} show that the reduced dynamics can be P or even CP-divisible while, for $\Delta>0$, QR does not hold (as a consequence of Proposition~\ref{prop:QRsg}, see also the following Remark~\ref{rmk:QRCPdiv_model}). In such case, standard measures of back-flow of information based on the dynamical map~\cite{BLP,RHPmeasure}, such as the BLP one, would vanish, while the quantity
	$
	\Delta \mathfrak{h}_S(\Theta,\mathcal{F})>0
	$, 
	introduced in Definition~\ref{def:entr-mark},
	namely, the ALF entropy results diminished with respect to its maximal value corresponding to the QR regime. 
    Such behaviour is depicted in Figure~\ref{fig:regionsentropy}, where we superimpose the regions of parameters corresponding to the QR, CP-divisible, P-divisible and non-P-divisible regimes.
	Notice also that the transition between 
	the various regimes is continuous in the parameter $\Delta$, as already observed in the literature on several occasions~\cite{FBGN2023,Megier_etal,Divisibility_hierarchy}.
Consider now $\Delta=p=1/2$, namely the parameters for which the two site correlations~\eqref{concrete_I} are maximal. Also, from~\eqref{tmatrix} it follows that $r=p_0=0$. Then,~\eqref{specificexample} yields
\begin{equation}
	p=\Delta=\frac{1}{2}\quad \Longrightarrow\quad \mathfrak{h}_S(\Theta)=0\,, 
\end{equation}
as one would have for a closed finite system. It is instructive to study the dynamical map $\Lambda_n^\ddagger$ in the same regime. One obtains the spectrum by setting $p=\Delta=1/2$ in~\eqref{1qubit_recursion_explicit}, so that $A_{\frac{1}{2},0}=0$, $B_{\frac{1}{2},0,\frac{1}{2}}=2$. The eigenvalues are
\begin{equation}
	\lambda_n=\frac{1+(-1)^n}{2}\,, \qquad  \lambda_n^{(3)}=(-1)^n\,.
\end{equation}
They correspond to the dynamical map 
\begin{equation}\label{extreme_dyn} 
	\Lambda_n^\ddagger=\begin{cases} \displaystyle\mathrm{id}_2\, & n \ \mathrm{even}\,,\\
		\displaystyle \Lambda_1^\ddagger\, & n \ \mathrm{odd}\,,  \qquad 
	\end{cases} \Lambda_1^\ddagger= \frac{\phi_1+\phi_2}{2}\,.
\end{equation}
Note that the dynamics is not algebraically invertible since $\lambda_{2m+1}=0$. Nevertheless, we can explicitly study the contractivity of the trace norm under $\Lambda_n$. In particular, the trace norm $\normt{X}=\Tr\sqrt{X^\dagger X}$ is always contractive between step $2m$ and $2m+1$,
\begin{equation}
	\normt{\Lambda^\ddagger_{2m+1}[X]}-\normt{\Lambda^\ddagger_{2m}[X]}=\normt{\Lambda^\ddagger_{1}[X]}-\normt{X}\le 0\, ,
\end{equation}
since $\Lambda_1$ is CPTP. Between steps $2m-1$ and $2m$, instead, we have the converse inequality:
\begin{equation}
	\normt{\Lambda^\ddagger_{2m}[X]}-\normt{\Lambda^\ddagger_{2m-1}[X]}=\normt{X}-\normt{\Lambda^\ddagger_{1}[X]}\ge 0\,,
\end{equation}
so that the trace norm always revives at odd times, signaling back-flow of information.

\begin{remark} 
	The operational interpretation of Section~\ref{sec:dynGNS_model}
	allows to better interpret the monotonically decreasing behaviour of the open-system ALF entropy~\eqref{specificexample} as function of the bath correlations $\Delta$. Note that 
	such behaviour is also reflected by the divisibility degree of the reduced dynamics. 
	Remarkably,  
	the dynamical map~\eqref{extreme_dyn} exhibits an extreme back-flow of information since the trace norm revives after each odd time, while the dynamical entropy is zero, as one would have for a reversible evolution (see Section~\ref{sec:ALFDISS}). This means that asymptotically, no new information can be gathered by further probing the open system. Such an interpretation can be generalized through inequality~\eqref{upperboundchain}: for a strongly correlated environment, we expect the  dynamics of the system to be affected by strong memory effects. On the other hand,  the occurrence  of a low, possibly zero, entropy production for a dissipative system is compatible with the interpretation of information having flown back to the open system since, asymptotically, repeated measurements on the system become little informative.
\end{remark}
\begin{figure}[t!]
	\centering
	\includegraphics[width=1\linewidth]{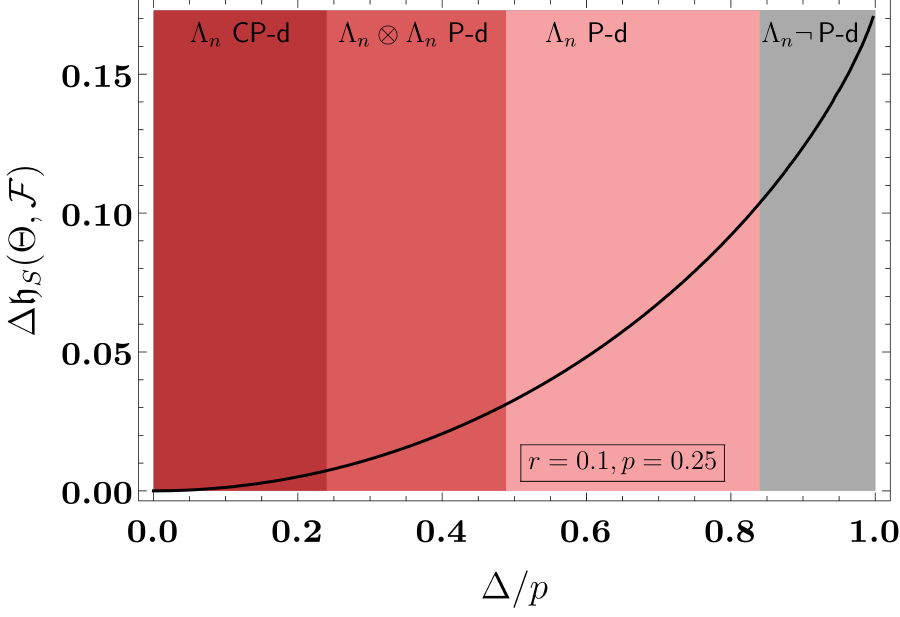}
	\caption{Entropic measure~\eqref{ALFmeasure} for the model of Section~\ref{ex:Markov1} as a function of $\Delta/p$ for fixed values $r=0.1$ and $p=0.25$. The colored regions correspond to different divisibility degrees of the reduced dynamics $\Lambda_n^\ddag$, as determined by conditions~\eqref{conditions_Delta}: in the grey region $\Lambda_n^\ddag$ is not P-divisible and shows distinguishability revivals; in the pink region $\Lambda_n^\ddag$ is P-divisible and does not show revivals, while they can super-activate for $\Lambda_n^\ddag\otimes\Lambda_n^\ddag$; the two-tensor power achieves P-divisibility in the red region, while in the dark-red region $\Lambda_n^\ddag$ is CP-divisible. 
	Notice that QR-Markovianity can be achieved only at $\Delta=0$, which yields a semigroup evolution.
	}
	\label{fig:regionsentropy}
\end{figure}

\begin{remark}
		 \label{rmk:QRCPdiv_model}
		When the classical environment is a Markov process we can write the coarse-grained density matrix as (compare with Remark~\ref{rmk:concrete_QR})
		\begin{multline*}
			\left(\rho_S\big[\mc{X}^{(2)}\big]\right)_{\bs{a},\bs{b}}=\sum_{i_2 i_1} \,T_{i_2 i_1} \,p_{i_1} \\ \omega_S\left( X_{b_0}^\dagger \phi_{i_1}\big[X_{b_1}^\dagger\big] \phi_{i_1} \circ\phi_{i_2} \big[X_{b_2}^\dagger X_{a_2}\big] \phi_{i_1}\big[X_{a_1}\big] X_{a_0} \right).
		\end{multline*}
		For the concrete case of $T_{i_2 i_1}$ as in~\eqref{tmatrix},
		we have 
		\begin{align*}
			\left(\rho_S\big[\mc{X}^{(2)}\big]\right)_{ \bs{a},\bs{b}}&=
			\omega_S\Big(X_{b_0}^\dagger \Lambda_1\left[ X_{b_1}^\dagger \Lambda_1\left[X_{b_2}^\dagger X_{a_2}\right]X_{a_1}\right]X_{a_0}\Big)\\
			&\hskip-2cm+p \,\Delta\, \bigg(\omega_S(X_{b_0}^\dagger \phi_1[X_{b_1}^\dagger] (\mathrm{id}_2-\phi_3)[X_{b_2}^\dagger X_{a_2}] \phi_1[X_{a_1}^{\phantom\dagger}]X_{a_0} )\\
			&\hskip-2cm+\omega_S(X_{b_0}^\dagger \phi_2[X_{b_1}^\dagger] (\mathrm{id}_2-\phi_3)[X_{b_2}^\dagger X_{a_2}] \phi_2[X_{a_1}^{\phantom\dagger}]X_{a_0} )\bigg)\,.
		\end{align*}
		Thus, $\Delta$, that is the strength of the environment correlations, signals the deviation of the evolution from the QR (and semigroup) regime. Nevertheless, for $\Delta>0$, the reduced dynamics could be (C)P-divisible, as displayed in Figure~\ref{fig:regionsentropy}.
\end{remark}
\subsection{Reduced Dynamics in the GNS representation}\label{sec:reducedGNS}
	In the GNS approach to the open-system ALF entropy, discussed in Section~\ref{sec:dynGNS_model}, the measurements affect only $\pi_S$ and not the commutant $\pi_S'$, 
	while the dynamics affects non-trivially also the commutant, as emerges from~\eqref{GNS_SE}.  We now remove the degrees of freedom of the environment by the expectation
	\vskip-1.5cm
	\begin{widetext}
		\begin{multline}
			\bra{\sqrt{\rho_S}\otimes \Omega_E} \pi_S\otimes\pi_E(X^\dag\otimes\mathds{1}_E) \left(\left(\wt{\mathbb{X}}\otimes \mathrm{id}_E\right)\circ\mathbb{U}_\Theta\right)^n \left[ \pi_{S}\otimes\pi_E(A\otimes\mathds{1}_E)\pi_{S}'\otimes\pi_E(B\otimes\mathds{1}_E) \right] \, \pi_S\otimes \pi_E(X\otimes\mathds{1}_E) \ket{\sqrt{\rho_S}\otimes\Omega_E}\\ 
			 =\Tr(\Gamma_n^{\mathbb{X}}
			\big[(X\otimes \mathds{1}_d) \ketbra{\sqrt{\rho_S}} (X^\dagger \otimes \mathds{1}_d)\big] A\otimes B)\label{generalGammamap}
		\end{multline}
	\end{widetext}
	where $A,B, X \in\Md{d}$, $\wt{\mathbb{X}}=\mathbb{X}\otimes \mathrm{id}_d$ and where we introduced a CPTP map $\Gamma_n^{\mathbb{X}}$ on $\Md{d}\otimes \Md{d}$, describing the joint dissipative evolution of $\pi_S$ and $\pi_S'$ together. From~\eqref{generalGammamap}, such map is well defined by the choice of a reference state of the form $\omega_S\otimes\omega_E$. Accordingly, as in the previous Section, the GNS triple can be identified with the triple $(\pi_S\otimes\pi_E, \mathcal{H}_S\otimes\mathcal{H}_E,\ket{\sqrt{\rho_S}\otimes\Omega_E})$. In particular, the representations $\pi_S$ and $\pi_S'$, amount to factor representations; namely, for $A\in \Md{d}$ (see Example~\ref{ex:GNS_finite})
	\begin{equation*}\label{modelstate_GNS}
		\pi_S(A)=A\otimes \mathds{1}_d\,, 	\qquad \pi_S'(A)=\mathds{1}_d\otimes\overline{A}\,.
	\end{equation*}
	The factorization of the $S+E$ cyclic vector, along with the complete positivity of $\left(\mathbb{U}_\Theta^\ddagger \circ\left(\wt{\mathbb{X}}^\ddagger\otimes\mathrm{id}_E\right)\right)^n$, ultimately ensure the complete positivity of~$\Gamma_n^{\mathbb{X}}$.
	We expect that such dissipative bipartite dynamics could help in assessing the lack of Markovianity of the stochastic process, for example by revealing memory effects that cannot be simply ascertained through the reduced dynamics $\Lambda_n^\ddag$.  
	In the case of the collisional model considered in Section~\ref{sec:collisionalmodels}, we fix the system $S$ reference state to be the maximally mixed one, 
	$\rho_S=\mathds{1}/{d}$.
	The system GNS vector then corresponds to the maximally entangled state  $\ket{\sqrt{\rho_S}}\equiv\ket{\psiplus{d}}=\sum_{i=1}^d \ket{i\otimes i}/\sqrt{d}$.
	The tensor shift on the chain is implemented by the unitary operator defined by
\begin{align}
	U_\sigma^\dagger\,\pi_E\left(\Pi_{\bs{i}^{[a,b]}}^{[a,b]}\right)\ket{\Omega_E}:=\pi_E\left(\Pi_{\bs{i}^{[a,b]}}^{[a+1,b+1]}\right)\ket{\Omega_E}\,,\label{Ushift}
\end{align}
and such that $U_\sigma^\dagger\ket{\Omega_E}=U_\sigma\ket{\Omega_E}=\ket{\Omega_E}$. Furthermore, we implement unitarily the automorphism $\Theta$ by the operator
\begin{align}\label{Udag_def}
	&U_\Theta^\dagger\,\pi_S(A) \pi_S'(B)\otimes\pi_E\left(\Pi_{\bs{i}^{[a,b]}}^{[a,b]}\right)\ket{\psiplus{d}\otimes\Omega_E}\\&\hskip0.2cm:=\sum_k U_k^\dagger A \otimes \overline{U_k}^\dagger  \overline{B} \otimes  U_\sigma^\dagger\pi_E\left(\Pi_k^{(0)}\Pi_{\bs{i}^{[a,b]}}^{[a,b]}\right) \ket{\psiplus{d}\otimes\Omega_E}.\nonumber
\end{align}
We then explicitly consider the evolution intertwined by measurements on the open system as appears in the coarse-grained density matrix, and evaluate explicitly~\eqref{generalGammamap}
\begin{widetext}
	\begin{align}	\label{exp_n}
	&\bra{\psiplus{d}\otimes\Omega_E} \pi_S\otimes\pi_E(X^\dag\otimes\mathds{1}_E) \big(\big(\wt{\mathbb{X}}\otimes\mathrm{id}_E\big)\circ\mathbb{U}_\Theta\big)^n [\pi_S\otimes\pi_E(A\otimes\mathds{1}_E)\pi_S'\otimes\pi_E(B\otimes\mathds{1}_E) ] \pi_S\otimes\pi_E(X\otimes\mathds{1}_E) \ket{\psiplus{d}\otimes\Omega_E}
	\\	&\hskip+2.5cm
	=\sum_{k_1\cdots k_n}  p_{k_1\cdots k_n} \Tr\bigg(A\otimes \overline{B}\,\phi_{k_n}^\ddagger\otimes\overline{\phi}_{k_n}^\ddagger\circ \wt{\mathbb{X}}^\ddagger\circ \cdots 
	 \circ\phi_{k_1}^\ddagger\otimes\overline{\phi}_{k_1}^{\ddagger}\circ\wt{\mathbb{X}}^\ddagger \, [X\otimes \mathds{1}_d\pplus{d} X^\dagger\otimes \mathds{1}_d]\bigg)\,.\nonumber
\end{align}
\end{widetext}
where $\pplus{d}=\ketbra{\psiplus{d}}$.
Accordingly, the dissipative GNS dynamics in the Schrödinger picture, resulting from the iterated measurement scheme and unitary evolution, is described by the following CPTP map
\begin{equation}\label{GammaXmodel}
	\Gamma_n^{\mathbb{X}}:=\sum_{k_1\cdots k_n}  p_{k_1\cdots k_n} \phi_{k_n}^\ddagger\otimes\overline{\phi}_{k_n}^\ddagger \circ\wt{\mathbb{X}}^\ddagger \circ 
	\cdots
	 \circ \phi_{k_1}^\ddagger\otimes\overline{\phi}_{k_1}^{\ddagger} \circ\wt{\mathbb{X}}^\ddagger\,.
\end{equation}
 where again $\wt{\mathbb{X}}^\ddagger=\mathbb{X}^\ddagger\otimes\mathrm{id}_d$. The detailed derivation of~\eqref{GammaXmodel} is reported in Appendix~\ref{app:GNSdissipative}.
The choice $\mathbb{X}=\mathrm{id}_d$, in particular,~yields
\begin{equation}\label{intrinsic_GNSdyn}
	\Gamma_n^{\mathrm{id}}\\
	=\sum_{\bs{k}_{[1,n]}} \, p_{\bs{k}_{[1,n]}}\, \phi_{\bs{k}_{[1,n]}}^\ddagger\otimes  \overline{\phi}_{\bs{k}_{[1,n]}}^\ddagger \,,  
\end{equation}
which is the intrinsic reduced evolution in the GNS space, having traced away the environment degrees of freedom. Formally,~\eqref{intrinsic_GNSdyn} corresponds to the joint evolution of two subsystems coupled in a common environment. 
Notice that the partial trace over the commutant gives,
\begin{align*}
	\Tr_{II}\left(\Gamma_n^{\mathrm{id}}[X\otimes Y]\right)&=\sum_{\bs{k}_{[1,n]}} \, p_{\bs{k}_{[1,n]}}\, \phi_{\bs{k}_{[1,n]}}^\ddagger[X] \Tr(Y)\\&=\Lambda_n^\ddag[\Tr(Y) \,X]=\Lambda_n^\ddag[\Tr_{II}(X\otimes Y)]\,.
\end{align*}
Hence, consistently, $\Gamma_n^{\mathrm{id}}$ provides a dissipative CP dilation of the dynamical map $\Lambda_n^\ddag$ of the system.

\subsection{Pauli GNS reduced dynamics}

We now consider the environment to be a classical Markov chain, as in Section~\ref{ex:Markov1}:
$$
p_{\bs{k}_{[1,n]}}=T_{k_n k_n-1}\dots T_{k_2 k_1} p_{k_1}\ ,
$$
 the open quantum system as a qubit $d=2$ and the environment as consisting of a one-dimensional lattice of classical four level systems $D=d^2=4$. Furthermore, we 
shall choose the  stochastic matrix $T$ as in~\eqref{tmatrix}, while setting the unitaries $U_k=\sigma_k$, $k=0,\ldots,3$. 
Accordingly, 
$\Gamma_n^{\mathrm{id}}:\Md{4}\to\Md{4}$ defined in~\eqref{intrinsic_GNSdyn} becomes a two-qubit 
Pauli map, 
namely for all $n\ge0$ it is  diagonal in the Pauli matrices
\begin{equation}
	\Gamma_n^{\mathrm{id}}[\sigma_\alpha\otimes\sigma_\beta]=\gamma_n^{(\alpha,\beta)}\sigma_\alpha\otimes\sigma_\beta\,, \qquad\alpha,   \beta=0,\dots,3\,.
\end{equation}
The spectrum of $\Gamma_n^{\mathrm{id}}$  
follows a similar rule than that of $\Lambda_n^\ddagger$
\begin{multline}\label{twoqubitrecursion}
	\gamma_n^{(\alpha,\beta)}=\left(p_0+p(\mu_1^{(\alpha,\beta)}+\mu_2^{(\alpha,\beta)})+r\mu_3^{(\alpha,\beta)}\right) \gamma_{n-1}^{(\alpha,\beta)}\\+ \Delta\, p\left(\mu_1^{(\alpha,\beta)}-\mu_2^{(\alpha,\beta)}\right)^2 \gamma_{n-2}^{(\alpha,\beta)}\ ,
\end{multline}
where  $\phi_k\otimes\phi_k[\sigma_\alpha\otimes\sigma_\beta]=\mu_k^{(\alpha,\beta)}\sigma_\alpha\otimes\sigma_\beta$, $ \mu_k^{(\alpha,\beta)}:=\mu_k^{(\alpha)}\mu_k^{(\beta)}$. The solutions of \eqref{twoqubitrecursion} are the same as that of~\eqref{onequbitrecursion}. The eigenvalues are associated to their respective eigenvectors $\sigma_\alpha\otimes\sigma_\beta$ according to the following $4\times 4$ matrix 
\begin{equation}
	G_n:=\left[\gamma_n^{(\alpha,\beta)}\right]=\begin{pmatrix}
		1 & \lambda_n &  \lambda_n &  \lambda_n^{(3)} \\
		\lambda_n & 1 & \lambda_n^{(3)} &  \lambda_n \\
		\lambda_n & \lambda_n^{(3)} & 1 & \lambda_n \\
		\lambda_n^{(3)} & \lambda_n &\lambda_n & 1 \\
	\end{pmatrix}\ .
\end{equation}
Note that all the Bell-diagonal states, which are linear combinations of $\sigma_\alpha\otimes\sigma_\alpha$, as well as the GNS reduced initial state $\pplus{d}$, are invariant under $\Gamma_n^{\mathrm{id}}$.
The spectrum of the intertwining propagators $\Gamma_{n,n-1}^{\textrm{id}}=\Gamma_{n}^{\textrm{id}}\circ\left(\Gamma_{n-1}^{\textrm{id}}\right)^{-1}$ is determined multiplicatively, 
\begin{equation}\nonumber
	\Gamma_{n,n-1}^{\textrm{id}}[\sigma_\alpha\otimes\sigma_\beta]=\gamma_{n,n-1}^{(\alpha,\beta)}\,\sigma_\alpha\otimes\sigma_\beta\,, \qquad \gamma_{n,n-1}^{(\alpha,\beta)}=\frac{\gamma_{n}^{(\alpha,\beta)}}{\gamma_{n-1}^{(\alpha,\beta)}}	\,,
\end{equation}
and given by
\begin{align*}
	G_{n,n-1}:&=\left[\gamma_{n,n-1}^{(\alpha,\beta)}	\right]\\&=\begin{pmatrix}
		1 & \lambda_{n,n-1} &  \lambda_{n,n-1} &  \lambda_{n,n-1}^{(3)} \\
		\lambda_{n,n-1} & 1 & \lambda_{n,n-1}^{(3)} &  \lambda_{n,n-1} \\
		\lambda_{n,n-1} & \lambda_{n,n-1}^{(3)} & 1 & \lambda_{n,n-1} \\
		\lambda_{n,n-1}^{(3)} & \lambda_{n,n-1} &\lambda_{n,n-1} & 1 \\
	\end{pmatrix},
\end{align*}
where $\lambda_{n,n-1}^{\phantom{(1)}}=\lambda_{n}/\lambda_{n-1}$,
$\lambda_{n,n-1}^{(3)}=\lambda_{n}^{(3)}/\lambda_{n-1}^{(3)}$ are the eigenvalues of $\Lambda_{n,n-1}$. Then, as derived in Appendix~\ref{rmk:Pauli1and2}, the action of $\Gamma_n^{\mathrm{id}}$  can be recast as
\begin{equation}
	\Gamma_n^{\mathrm{id}}[X]=\sum_{\alpha} \ q_n^{(\alpha)}\ \sigma_\alpha \otimes \sigma_\alpha\, X \,\sigma_\alpha \otimes \sigma_\alpha 
	\,.
\end{equation}
with
\begin{gather*}
	q_n^{(0)}=\frac{1}{4}\left(1+2\lambda_n+\lambda_n^{(3)}\right),\quad q_n^{(1)}=q_n^{(2)}=\frac{1}{4}\left(1-\lambda_n^{(3)}\right),\\  q_n^{(3)}=\frac{1}{4}\left(1+2\lambda_n-\lambda_n^{(3)}\right).
\end{gather*}
We now aim at characterizing the divisibility properties of the GNS discrete dynamics $\Gamma_n^{\textrm{id}}$. 
\begin{proposition} 
	Let 
	$\Gamma_n:\Md{4}\to\Md{4}$ be the Pauli map given by
	$$
	\Gamma_n[X]=\sum_{\alpha=0}^3 \, q_n^{(\alpha)} \,\sigma_\alpha\, \otimes \sigma_\alpha \,X\, \sigma_\alpha \otimes \sigma_\alpha \,, \qquad X\in\Md{4}\,.
	$$
	Then, consider the one-qubit CPU map obtained as $$\Lambda_n\circ\Tr_{II}[X]=\Tr_{II}(  \Gamma_n[X])\,.$$
	where $\Tr_{II}$ is the partial trace over the second qubit.
	Then, following conditions are equivalent
	\begin{enumerate}[label=\roman*.]
		\item $\Lambda_n$ is CP-divisible;
		\item  $\Gamma_n$ is P-divisible;
		\item  $\Gamma_n$ is CP-divisible.
	\end{enumerate}
\end{proposition}
\begin{proof}
	Take $X\in\Md{4}$ and let $Y=\Gamma_n[X]$. The partial trace of $X$ over the second qubit reads
	\(
		\Tr_{II}(X)=\Tr_{II}(\Gamma_n^{-1}[Y])
	\)
	but also
	\begin{multline*}
		\Tr_{II}(X)=\Lambda_n^{-1}\circ\Lambda_n[\Tr_{II}(X)]=\Lambda_n^{-1}[\Tr_{II}(\Gamma_n [X])]\\=\Lambda_n^{-1}[\Tr_{II}(Y)]\,,
	\end{multline*}
	Hence, for all $Y\in\Md{4}$,
	\(
		\Lambda_n^{-1}[\Tr_{II}(Y)]= \Tr_{II}(\Gamma_n^{-1}[Y]).
	\)
	Applying $\Lambda_m$, $m\ge n$, to both sides, yields
	\begin{equation}
		\Lambda_{m,n}[\Tr_{II}(Y)]=\Tr_{II}(\Gamma_{m,n}[Y])\,,\qquad m\ge n\,.
	\end{equation}
	Intertwiners of Pauli maps are Pauli maps, so that
	\begin{equation}
		\Gamma_{n,n-1}[X]=\sum_{\alpha=0}^3 q^{(\alpha)}_{n,n-1}\sigma_\alpha\otimes\sigma_\alpha X \sigma_\alpha\otimes\sigma_\alpha \,,
	\end{equation} 
	and 
	\(
		\Tr_{II}\left(\Gamma_{n,n-1}[X]\right)=\Lambda_{n,n-1}\left[\Tr_{II}X\right],
	\)
	$\forall\, n\ge1$, with 
	\begin{equation}
		\Md{2}\ni Y\longmapsto\Lambda_{n,n-1}[Y]=\sum_{\alpha=0}^3 q^{(\alpha)}_{n,n-1}\sigma_\alpha\,Y\,\sigma_\alpha\, .
	\end{equation}
	Positivity of $\Gamma_{n,n-1}$, $n\ge1$, is expressed by the condition
	\begin{equation}
		\Tr(Q\,\Gamma_{n,n-1}[P])\ge0 \,, 	\qquad	\forall \,P,Q\ge0\,,
	\end{equation}
	where $P,Q$ can be chosen as projectors in $\Md{4}$. Pick separable projectors
	\begin{equation}
		P=Q=\ketbra{+}{+}\otimes\ketbra{0}{0}\,,\qquad \ket{\pm}=\frac{\ket{0}\pm\ket{1}}{\sqrt{2}}\,,
	\end{equation}
	with $\sigma_3\ket{0}=\ket{0}$, $\sigma_3\ket{1}=-\ket{1}$. One readily finds
	\begin{equation}
		\Tr(Q \,\Gamma_{n,n-1}[P])=q_{n,n-1}^{(0)}\ge0 \ .
	\end{equation}
	The choice
	\begin{equation}
		P=\ketbra{+}{+}\otimes\ketbra{0}{0}\,,\qquad Q=\ketbra{-}{-}\otimes\ketbra{0}{0}\,, \quad PQ=0\,,
	\end{equation}
	yields instead the condition
	\begin{equation}
		\Tr(Q\,\Gamma_{n,n-1}[P])=q_{n,n-1}^{(3)}\ge0 \ .
	\end{equation}
	Similarly, one finds 
	$q_{n,n-1}^{(\alpha)}\ge0$, $\alpha=0,\dots,3$, so that $ii. \implies i\implies iii.$
	Clearly, $iii. \implies ii.$  
\end{proof}
Interestingly, for the model under consideration, the intrinsic evolution $\Gamma_n^{\mathrm{id}}$ in the GNS construction can display memory effects not present in the one-qubit reduced dynamics $\Lambda_n$. 
\begin{corollary}\label{SBFI_remark}
	Let $\Delta$ be such that $\Lambda_n^\ddagger$ is P-divisible but not CP-divisible. Then, the evolution of the system does not display back-flow of information in terms of revivals. 
	On the other hand, as the two-qubit dilation $\Gamma_n^{\mathrm{id}}$ is not P-divisible, it will show it.
\end{corollary}
\begin{proof}
	P-divisibility of the reduced qubit dynamics  ensures that
	\begin{equation*}
		\normt{\Lambda^\ddagger_n[X]}\le \normt{\Lambda^\ddagger_m[X]}\,,\qquad \forall\ n\ge m\,, \ \forall \ X\in\Md{2}\,. 
	\end{equation*}
	On the other hand,  as the two-qubit dilation $\Gamma_n^{\mathrm{id}}$ is not P-divisible and invertible,  it certainly breaks monotonicity, namely there exists $X\in\Md{4}$ and $n> m\in\mathbb{N}$ such that
	$	\big\|{\Gamma_n^{\rm{id}}[X]}\big\|_1 >
		\big\|\Gamma_m^{\mathrm{id}}{[X]}\big\|_1\,.$
\end{proof}

\section{Conclusions}
In this work, we have extended the ALF entropy which has mainly been applied in  reversible, quantum many-body settings, to dissipative finite dimensional 
quantum systems and named it \textit{open-system ALF entropy}. 
 The first part of the paper was devoted to presenting the general framework for the newly extended open-system ALF entropy and its operational interpretation. In Section~\ref{sec:ALFDISS}, we introduced the positive quantity~\eqref{ALFmeasure} signalling the deviation of the entropy rate from the QR regime. We proposed the latter as an entropic  measure of  information flowing from the environment back into the open quantum system which takes into account the full multi-time statistics. The latter operational interpretation is justified by the analysis in the GNS representation. 
In the second part of the paper, we have explicitly computed the open-system ALF entropy for a finite-level quantum system coupled to a classical stationary chain; in particular, we found conditions under which the ALF entropy equals the mean entropy of the environment. This result has a simple but clear information-theoretic interpretation: the ALF entropy measures the average amount of information that can be extracted by probing the open system alone. Such an information rate decreases as the environment becomes more correlated.
In the qubit case, this physical effect has been compared to memory effects arising in the reduced qubit dynamics and associated with the lack of (C)P-divisibility, which are likewise governed by the strength of environmental correlations. This also allows for an explicit comparison of our proposal with other standard measures of non-Markovianity.
\hfill\break
The open-system ALF entropy, obtained from  the coarse-grained matrix $\rho_S[\mc{X}^{(n)}]$ built from the system multi-time correlation functions, is neatly interpreted when going to the GNS representation which involves a purifying ancillary system. Accordingly, we studied the GNS reduced dynamics obtained by tracing away the environment for the model of Section~\ref{sec:collisionalmodels}. Interestingly, the GNS dilation is affected by memory effects that are seen by the reduced dynamics. In particular, it can show super-activation of back-flow of information in terms of revivals, as it typically occurs when a tensor product dilation of the reduced dynamics is considered~\cite{BenattiChrusFil,FBGN2023,FBGN_PhysicaScr}.\hfill\break
In a more general scenario of a system coupled to a quantum spin chain, the above interpretation is supported by an upper bound to the ALF entropy of the open system by the mean entropy of the environment (plus a quantum correction). For a stationary chain, the more correlated it is, the tighter is such bound. This is only a partial result, a thorough  study of the ALF entropy in the case of a genuinely quantum environment still awaiting to be developed.

\section*{Acknowledgments} G.N. and F.B.  acknowledge financial support from PNRR MUR project PE0000023-NQSTI.

\appendix
\begin{appendices}
	\counterwithin*{equation}{section}
	\renewcommand\theequation{\thesection\arabic{equation}}

	\section{Proof of Proposition~\ref{thm:partitionF_general}.}\label{app:proof_partitionF_general}

	We need to recast the system coarse-grained matrix $\rho_S\left[\mc{F}^{(n+1)}\right]$ for the POVM $\mc{F}$ defined in~\eqref{specialOPU} 
	in the form
	$$
	\rho_S\left[\mc{F}^{(n+1)}\right] =\rho_S\otimes \rho_S \otimes  \left(\mathbb{T}_n^\ddag \otimes \mathrm{id} \left[\ketbra{\sqrt{\rho_S^{\otimes n}}}\right]\right) .
	$$
	By definition, 
	\begin{align}
		\mathbb{T}_n\left[\;\bigotimes_{k=0}^{n-1} A_k\right]&=\omega_E\left(\wt{\Theta}_n^{\mathbb{Z}E}\Bigg[\;\bigotimes_{k=0}^{n-1} A_k^{(k-n)}\otimes \mathds{1}_E\Bigg]
		\right)
		\label{def_proofn_1}\\&\hskip-2cm = \omega_E\Bigg(
		\Theta_1 \bigg[
		A_0^{(0)}\otimes \Theta_{2,1}\Big[A_1^{(1)} \otimes  \ldots 
		\Theta_{n,n-1}\big[A_{n-1}^{(n-1)}\otimes \mathds{1}_E
		\big]
		\Big]
		\bigg]
		\Bigg)\,,\label{def_proofn_2}
	\end{align}
	where we set
	\begin{equation}
		\wt{\Theta}_n^{\mathbb{Z}E}:= \prod_{j=1}^{n} \left(\Theta_{j,j-1}\circ \sigma_S\otimes 	\mathrm{id}_E\right) \,,
	\end{equation}
	and $\sigma_S$ denotes the shift on $M_d^{\mathbb{Z}}(\mathbb{C})$. From~\eqref{def_proofn_1}, $\mathbb{T}_n$ is manifestly completely positive and unital. 
	Notice that $\mathbb{T}_1[A_0]=\omega_E\left(  \Theta_1\left(A_0^{(0)}\otimes \mathds{1}_E\right)    \right)= \Lambda_1[A_0]$. 
	Moreover, for sake of brevity, the positioning (upper indices) of system operators along the $S$-chain in~\eqref{def_proofn_2} has been chosen in accordance with how many times the shift has acted after the automorphism $\Theta$ made the $0$-th site system operators interact with the environment; indeed, 
	\begin{multline*}
		\Theta_{n,n-1}\circ (\sigma_S\otimes\mathrm{id}_E)\Big[A_{0}^{(-n)}\otimes\ldots\otimes A_{n-1}^{(-1)}\Big]\\:=A_0^{(-n+1)} \otimes \ldots \otimes \Theta_{n,n-1}\Big[A_{n-1}^{(0)}\otimes\mathds{1}_E\Big]\,.
	\end{multline*}
	Then, the action of the automorphism,
	$\Theta_{n,n-1}\Big[A_{n-1}^{(0)}\otimes\mathds{1}_E\Big]=\widetilde{A}^{(0)}\otimes B_E$, is such that
	\begin{align*}
		&\sigma_S\otimes\mathrm{id}_E\Big[A_{0}^{(-n+1)}\otimes\cdots\otimes\Theta_{n,n-1}\Big[A_{n-1}^{(0)}\otimes\mathds{1}_E\Big]\Big]\\&=
		A_{0}^{(-n+2)}\otimes\cdots\otimes A^{(0)}_{n-2}\otimes\widetilde{A}_{n-1}^{(1)}\otimes B_E\\
		& =
		A_{0}^{(-n+2)}\otimes\cdots\otimes A^{(0)}_{n-2}\otimes\Theta_{n,n-1}\Big[A_{n-1}^{(1)}\otimes\mathrm{id}_E\Big]\ .
	\end{align*}
	One thus gets
	\begin{multline*}
		\Big(\Theta_{n-1,n-2}\circ (\sigma_S\otimes\mathrm{id}_E)\Big) \\ \circ\Big(\Theta_{n,n-1}\circ (\sigma_S\otimes\mathrm{id}_E)\Big)
		\Big[A_{0}^{(-n)}\otimes\ldots\otimes A_{n-1}^{(-1)}\Big]\\\!=\!
		A_{0}^{(-n+2)}\otimes\cdots\otimes\Theta_{n-1,n-2}\Big[A^{(0)}_{n-2}\otimes\Theta_{n,n-1}\Big[A_{n-1}^{(1)}\otimes\mathrm{id}_E\Big]\Big] .
	\end{multline*}
		The elements of the coarse-grained density matrix with respect to the POVM~\eqref{specialOPU} read
	\begin{widetext}
	\begin{align}\nonumber
		\rho_S\left[\mathcal{F}^{(n+1)}\right]_{\bs{a},\bs{b}}
		&=\ \omega_S\otimes \omega_E\Big(F_{b_0 b_0'}^\dagger \otimes \mathds{1}_E\,\Theta_1\Big[F_{b_1 b_1'}^\dagger\otimes \mathds{1}_E\ldots \Theta_{n,n-1}\big[F_{b_nb_n'}^\dagger F_{a_{n}a_{n}'}^{\phantom{\dagger}}\otimes \mathds{1}_E\big] \ldots  F_{a_1 a_1'}\otimes \mathds{1}_E\Big]\, F_{a_0 a_0'}\otimes \mathds{1}_E\Big)\nonumber
		\\
		&\hskip-2cm= \ r_{a_n}\delta_{a_n b_n}r_{a_0'}\delta_{a_0' b_0'} 
		\left(\prod_{m=0}^{n-1}\sqrt{r_{a_m}r_{b_m}}\right)\Tr\Bigg(R_{a_0 b_0} \,\omega_E\bigg(\Theta_1\Big[R_{b_1 b_1'}^\dagger\otimes \mathds{1}_E \ldots \Theta_{n,n-1}\big[R_{a_n'b_n'}^\dagger \otimes \mathds{1}_E\big] \ldots  R_{a_1 a_1'}\otimes \mathds{1}_E\Big]\bigg)\Bigg)\,,\label{compare_n}
	\end{align}
\end{widetext}
	with multi-indices $\bs{a}=(a_0a'_0,\ldots a_na'_n)$, $\bs{b}=(b_0b'_0,\ldots,b_bb'_n)$ and where $R_{a b}=\ketbra{r_a}{r_b}$ are the matrix units corresponding to the eigenbasis of $\rho_S$; they satisfy
	$	R^\dag_{ab}=R_{ba}$, $ R_{ab}\,R_{cd}=\delta_{bc}\,R_{ad}$.
	We shall now obtain another expression for~\eqref{compare_n} in terms of the CPU map~\eqref{def_proofn_1}. 
	To this end, we split the proof in several steps.
	
	\hfill\break
	\noindent 
	\textbf{Step 1.}\quad Firstly, we introduce
	a state functional $\Omega_S^{\mathbb{Z}\mathbb{Z}}:M_d^{\mathbb{Z}}(\mathbb{C})\,\otimes \,M_d^{\mathbb{Z}}(\mathbb{C})\,\to\, \mathbb{C}$ on the doubled $S$-chain, whose local restrictions are the purified vector states corresponding to the tensor products of the system $S$ state; namely,
	\begin{gather}
		\Omega^{\mathbb{Z}\mathbb{Z}}_S\left(X^{[a,a+n]}\right)=\bra{\sqrt{\rho_S^{\otimes [1,n]}}} X^{[a,a+n]}\ket{\sqrt{\rho_S^{\otimes [1,n]}}}\ ,\label{purifiedtensor}\\ 
		\hskip+2.cm X^{[a,a+n]}\in M_d^{[a,a+n]}(\mathbb{C})\otimes M_d^{[a,a+n]}(\mathbb{C})\ ,\nonumber
		\nonumber\\ \label{aux2}
		\ket{\sqrt{\rho_S^{\otimes [1,n]}}}=\sum_{\bs{a}_{[1,n]} } \sqrt{r_{\bs{a}_{[1,n]}}}         \,\ket{r_{\bs{a}_{[1,n]}}\otimes r_{\bs{a}_{[1,n]}}},
	\end{gather}
	where $r_{\bs{a}_{[1,n]}}=\prod_{k=1}^n r_{a_k}$.   Consider then the compound dynamical system 
	\begin{gather}
		\label{supporting_dynsyst}
		\left(M_d^{\mathbb{Z}}(\mathbb{C})\otimes M_d^{\mathbb{Z}}(\mathbb{C})\otimes \mc{A}_E,\,\Omega^{\mathbb{Z}\mathbb{Z}}_S\otimes \omega_E, \,
		\wt{\Theta}_n^{\mathbb{Z}\mathbb{Z}E}
		\right), \\ \hbox{where}\qquad	\wt{\Theta}_n^{\mathbb{Z}\mathbb{Z}E}:=\prod_{k=1}^n \left( \Theta_{k,k-1}\otimes\mathrm{id}_{\mathbb{Z}} \circ \sigma_S\otimes \sigma_S \otimes \mathrm{id}_E\right).\nonumber
	\end{gather}
	Introducing the operators 
	\begin{align}
		\label{aux0}
		&Z^{[-n,-1]}_{LR}\otimes\mathds{1}_E\ ,\quad Z^{[-n,-1]}_{LR}:=Z^{[-n,-1]}_L\otimes Z^{[-n,-1]}_R\ ,\\
		\label{aux0b}
		&Z_L^{[-n,-1]}:=\bigotimes_{k=1}^nR_{a_k' b_k'}^{\dagger (k-n-1)}\ ,\quad
		Z_R^{[-n,-1]}:=\bigotimes_{l=0}^{n-1} R_{a_l b_l}^{\dagger (l-n)}\ ,
	\end{align}
	the definition in~\eqref{def_proofn_2} of the CPU map $\mathbb{T}_n$ yields
	\begin{multline}
		\label{eq:aux0}
		\Omega_S^{\mathbb{Z}\mathbb{Z}}\otimes \omega_E \Bigg( 	\wt{\Theta}_n^{\mathbb{Z}\mathbb{Z}E} \left[Z^{[-n,-1]}_{LR}\otimes\mathds{1}_E\right]\Bigg)
		\\=\Omega_S^{\mathbb{Z}\mathbb{Z}}\left(
		\mathbb{T}_n\left[Z_L^{[-n,-1]}\right]\otimes Z_R^{[0,n-1]}\right)\ .
	\end{multline}
	\noindent
	\textbf{Step 2.}\quad
	Since $Z^{[-n,-1]}_{LR}$  is supported by the double chain sites $-n\leq j\leq -1$, the action of the $n$-th power of the shift $\sigma_S\otimes\sigma_S$, makes the operator
	$\mathbb{T}_n\big[Z_L^{[-n,-1]}\big]\otimes Z_R^{[0,n-1]}$
	supported by the interval $[0,n-1]$.
	Then, in the Schrödinger picture, using the dual $\mathbb{T}^\ddag_n$, one can rewrite the expectation on the right hand side of~\eqref{eq:aux0} as 
	\begin{align}
		\nonumber
		&\Omega_S^{\mathbb{Z}\mathbb{Z}}\left(
		\mathbb{T}_n\left[Z_L^{[-n,-1]}\right]\otimes Z_R^{[0,n-1]}\right)=\\
		&=\Tr_{[-n,-1]}\Bigg( \mathbb{T}^\ddag_n \otimes \mathrm{id}_d^{\otimes n}\left[\ketbra{\sqrt{\rho_S^{\otimes[0,n-1]}}}\right]     \,  \Bigg.\nonumber\\
		& \hskip+4cm \Bigg. Z_R^{[-n,-1]}\otimes Z_R^{[0,n-1]}\Bigg)\,.
		\label{map_n}
	\end{align}
	On the other hand, the left hand side of~\eqref{eq:aux0} reads
		\begin{align*}
		&\Omega_S^{\mathbb{Z}\mathbb{Z}}\otimes \omega_E \Bigg( \widetilde{\Theta}^{\mathbb{Z}\mathbb{Z}E}_n \left[Z_{LR}^{[-n,-1]}\otimes \mathds{1}_E\right]\Bigg)\\&=\Omega_S^{\mathbb{Z}\mathbb{Z}}\otimes \omega_E \Bigg( \Theta_n^{\mathbb{Z}E} \left[
		Z_L^{[-n,-1]}\otimes \mathds{1}_E\right]\otimes Z_R^{[0,n-1]}\Bigg)
		\\
		&=\bra{\sqrt{\rho_S^{\otimes[0,n-1]}}}
		\omega_E \Bigg( \wt{\Theta}_n^{\mathbb{Z}E} \left[Z_L^{[ -n,-1]}\otimes\mathds{1}_E\right]\Bigg)
		\\&\hskip+4cm\otimes Z_R^{[0,n-1]}\, \ket
		{\sqrt{\rho_S^{\otimes[0,n-1]}}}
		.
	\end{align*}
	By explicitly using the spectral form of the vector state in~\eqref{aux2}, one gets
	\begin{align*}
		&\Omega_S^{\mathbb{Z}\mathbb{Z}}\otimes \omega_E \Bigg( \wt{\Theta}_n^{\mathbb{Z}\mathbb{Z}E}  \left[Z_{LR}^{[-n,-1]}\otimes \mathds{1}_E\right]\Bigg)=\sum_{\substack{\bs{i}_{[0,n-1]}\\ \bs{j}_{[0,n-1]}}} \prod_{m=0}^{n-1} \sqrt{r_{i_m} r_{j_m}}
		\\
		&\quad\times\Tr\Bigg( \bigotimes_{q=0}^{n-1} R_{i_q j_q}^{(q)}  \,
		\omega_E \bigg(
		\wt{\Theta}_n^{\mathbb{Z}E}  \Big[Z_L^{[-n,-1]}\otimes \mathds{1}_E
		\Big]
		\bigg)
		\Bigg)\\
		&\quad \times \Tr\left(
		\bigotimes_{p=0}^{n-1} R_{i_p j_p}^{(p)}\,Z_R^{[0,n-1]}\right)\, .
	\end{align*}
	Using the expression for $Z_R^{[0,n-1]}$ in~\eqref{aux0b} yields $\displaystyle \Tr\left(\bigotimes_{p=0}^{n-1} R_{i_p j_p}^{(p)}\,Z_R^{[0,n-1]}\right)=\prod_{p=0}^{n-1}\delta_{i_pa_p}\,\delta_{j_pb_p},$
	so that
	\begin{align}
		\nonumber
		&\Omega_S^{\mathbb{Z}\mathbb{Z}}\otimes \omega_E \Bigg( \wt{\Theta}_n^{\mathbb{Z}\mathbb{Z}E} \left(Z_{LR}^{[-n,-1]}\otimes \mathds{1}_E\right)\Bigg)
		= \left(\prod_{m=0}^{n-1} \sqrt{r_{a_m} r_{b_m}}
		\right) \nonumber\\
		&\ \times \Tr_{[0,n-1]}\! \left(
		\bigotimes_{l=0}^{n-1} R_{a_l b_l}^{(l)}
		\omega_E\bigg(
		\Theta_1\Big[ R_{a_1',b_1'}^{\dagger (0)} \otimes\Theta_{2,1}\Big[R_{a_2',b_2'}^{\dagger (1)} \ldots 
			\right.
		\nonumber\\ &\left. \ \otimes \,\Theta_{n,n-1} \big[R_{a_{n}' b_{n}'}^{\dagger (n-1)}\otimes \mathds{1}_E\big] \Big]\Big]
		\bigg)
		\right)\ .\label{pre_rewrite}
	\end{align}	
	\begin{widetext}
	\noindent
	\textbf{Step 3.}\quad 
	The trace over the local sub-algebras $M_d^{\otimes n}(\mathbb{C})$ in~\eqref{pre_rewrite} can be converted into a trace over 
	the single algebra $M_d(\mathbb{C})$ at time-step $0$; namely,
	\begin{align}	\label{rewrite}
		&\Omega_S^{\mathbb{Z}\mathbb{Z}}\otimes \omega_E \Bigg( \wt{\Theta}_n^{\mathbb{Z}\mathbb{Z}E} \left[Z_{LR}^{[-n,-1]}\otimes \mathds{1}_E\right]\Bigg)=\left(\prod_{m=0}^{n-1}\sqrt{r_{a_m}r_{b_m}}\right)	 \Tr\Bigg(R_{a_0 b_0}\omega_E\bigg(\Theta_1\Big[R_{b_1 b_1'}^\dagger\otimes \mathds{1}_E\,\Theta_{2,1}\, \Big[
		\cdots
		\nonumber\\
		&\cdots\Theta_{n-1,n-2}\Big[R^\dag_{b_{n-1}b'_{n-1}}\otimes\mathds{1}_E\,\Theta_{n,n-1}\big(R_{a_n' b_n'}^\dagger \otimes \mathds{1}_E\big) R_{a_{n-1} a_{n-1}'}\otimes \mathds{1}_E\Big]\cdots 
		\Big]
		R_{a_1 a_1'}\otimes \mathds{1}_E\Big]\bigg)\, \Bigg)\,.
		\nonumber
	\end{align}
	This a fact is a consequence of the following equality,
	\begin{align}
		& \Tr_{[0,n-1]}\!\left(
		\bigotimes_{l=0}^{n-1} R_{a_l b_l}^{(l)}
		\Theta_1\Big[ R_{a_1',b_1'}^{\dagger (0)} \otimes\Theta_{2,1}\Big[R_{a_2',b_2'}^{\dagger (1)} \ldots   \otimes \Theta_{n,n-1} \Big[R_{a_{n}' b_{n}'}^{\dagger (n-1)}\otimes E\Big] \Big]\Big]\right)
		\\  
		&=\Tr\Bigg(R_{a_0 b_0}\Theta_1\Big[R_{b_1 b_1'}^\dagger\otimes \mathds{1}_E\,\Theta_{2,1}\, \Big[ R^\dag_{b_2b'_2}\otimes\mathds{1}_E\,\Theta_{3,2}\Big[
		\cdots\nonumber
		\\
		&	\quad \cdots\,\Theta_{n-1,n-2}\Big[R^\dag_{b_{n-1}b'_{n-1}}\otimes\mathds{1}_E\, \Theta_{n,n-1}\Big[R_{a_n' b_n'}^\dagger \otimes E\Big]R_{a_{n-1} a_{n-1}'}\otimes \mathds{1}_E\Big]\cdots R_{a_2a_2'}\otimes\mathds{1}_E\Big] R_{a_1 a_1'}\otimes \mathds{1}_E\Big]\Big]\, \Bigg)\, ,   
		\label{pre-rewrite_1}
	\end{align}
	that holds for any environment operator $E$ and which can be proved by recursion. 
	Indeed, the equality is true for $n=1$; assume it holds for up to $n=k$ and consider its left hand side for $n=k+1$:
	\begin{equation}
		\label{aux_rewrite}
		\Tr_{[0,k]}\!\left(
		\bigotimes_{l=0}^{k} R_{a_l b_l}^{(l)}
		\Theta_1\Bigg[ R_{a_1',b_1'}^{\dagger (0)} \otimes\Theta_{2,1}\bigg[R_{a_2',b_2'}^{\dagger (1)} \ldots   \otimes \Theta_{k,k-1}\Big[R_{a_{k}' b_{k}'}^{\dagger (k)}\otimes\Theta_{k+1,k}\Big[R_{a_{k+1}' b_{k+1}'}^{\dagger (k)}\otimes E\Big] \Big]\bigg]\Bigg]\right) \, .   
	\end{equation}

	Note that, expanding with respect to the system matrix units, one can always cast the action of the automorphism $\Theta_{k+1,k}$ on $M_d^{(0)}(\mathbb{C})\otimes\mathcal{A}_E$ as
	\begin{equation}
		\label{aux_aut_1}
		\Theta_{k+1,k} \Big[R^\dagger_{a_{k+1}' b_{k+1}'}\otimes E\Big]=\sum_{ij}R_{ij}\otimes E_{ij}^{a_{k+1}'b_{k+1}'}\, ,
	\end{equation}
	with suitable environment operators $E_{ij}^{a_{k+1}'b_{k+1}'}$ and then shift $k$ times to the right the system operators, so that:
	\begin{equation}
		\label{aux_aut_2}
		\Theta_{k+1,k} \Big[R^{\dag (k)}_{a_{k+1}' b_{k+1}'}\otimes E\Big]=\sum_{ij}R^{(k)}_{ij}\otimes E_{ij}^{a_{k+1}'b_{k+1}'}\, .
	\end{equation}
	Tracing over the site $k$ then yields
	$
	\Tr_{\{k\}}\Bigg(R^{(k)}_{a_k b_k}\,R^{\dag (k)}_{a_{k+1}' b_{k+1}'}\otimes E\Bigg)=E_{b_k a_k}^{a_{k+1}'b_{k+1}'}\, .
	$
	Plugging this result in the expression~\eqref{aux_rewrite}, it becomes 
	\begin{equation}
		\label{aux_rewrite_1}
		\Tr_{[0,k-1]}\!\left(
		\bigotimes_{l=0}^{k} R_{a_l b_l}^{(l)}
		\Theta_1\Big[ R_{a_1',b_1'}^{\dagger (0)} \otimes\Theta_{2,1}\Big[R_{a_2',b_2'}^{\dagger (1)} \ldots   \otimes \Theta_{k,k-1}\Big[R_{a_{k}' b_{k}'}^{\dagger (k)}\otimes E_{b_k a_k}^{a_{k+1}'b_{k+1}'}\Big] \Big]\Big]\right) \, .   
	\end{equation}
	By the recursion assumption, this expression equals
	\begin{align}
		&\Tr\Bigg(R_{a_0 b_0}\Theta_1\bigg[R_{b_1 b_1'}^\dagger\otimes \mathds{1}_E\,\Theta_{2,1}\, \Big[
		\cdots
		\label{aux_rewrite_3}
		\Theta_{k-1,k-2}\Big[R^\dag_{b_{k-1}b'_{k-1}}\otimes\mathds{1}_E\Theta_{k,k-1}\Big[R_{a_k' b_k'}^\dagger \otimes E_{b_k a_k}^{a_{k+1}'b_{k+1}'}\Big] R_{a_{k-1} a_{k-1}'}\otimes \mathds{1}_E\Big]\cdots 
		\Big] R_{a_1 a_1'}\otimes \mathds{1}_E\bigg]\Bigg)\, . 
	\end{align}
	The proof of equality~\eqref{pre-rewrite_1} is concluded by using~\eqref{aux_aut_1} with $E=\mathds{1}_E$ and by observing that
	\begin{equation}
		\label{aux_aut_3}
		R_{a_k' b_k'}^\dagger \otimes E_{b_k a_k}^{a_{k+1}'b_{k+1}'}= R^\dag_{b_kb'_k}\otimes\mathds{1}_E\,\Theta_{k+1,k}\Big[R^\dag_{a'_{k+1}b'_{k+1}}\otimes\mathds{1}_E\Big]R_{a_ka'_k}\otimes\mathds{1}_E\ .
	\end{equation}
	
	\noindent
	\textbf{Step 4.}\quad Putting together the equalities~\eqref{eq:aux0},~\eqref{map_n} and~\eqref{rewrite} one gets
	\begin{align}
		\nonumber
		&\left(\prod_{m=0}^{n-1}\sqrt{r_{a_m}r_{b_m}}\right)\,\Tr\Bigg(R_{a_0 b_0}\omega_E\bigg(\Theta_1\Big[R_{b_1 b_1'}^\dagger\otimes \mathds{1}_E\,\Theta_{2,1}\, \Big[ R^\dag_{b_2b'_2}\otimes\mathds{1}_E\Theta_{3,2}\Big[\cdots\\
		\nonumber    
		&
		\cdots\Theta_{n-1,n-2}\Big[R^\dag_{b_{n-1}b'_{n-1}}\otimes\mathds{1}_E\,\Theta_{n,n-1}\Big[R_{a_n' b_n'}^\dagger \otimes \mathds{1}_E\Big] R_{a_{n-1} a_{n-1}'}\otimes \mathds{1}_E\Big]\cdots R_{a_2a_2'}\otimes\mathds{1}_E\Big] R_{a_1 a_1'}\otimes \mathds{1}_E\Big]\bigg)\, \Bigg)\\
		&
		=\Tr_{[-n,-1]}\Bigg( \mathbb{T}_n^\ddag  \otimes \mathrm{id}_d^{\otimes n}\left[\ketbra{\sqrt{\rho_S^{\otimes[0,n-1]}}}\right]     \, Z_R^{[-n,-1]}\otimes Z_R^{[0,n-1]}\Bigg)\ .    
		\label{eq:aux_1}
	\end{align}
	Plugging this relation in~\eqref{compare_n}, one gets
	\begin{multline}
		\rho_S\left[\mathcal{F}^{(n+1)}\right]_{\bs{a},\bs{b}}=r_{a_0'}\delta_{a_0' b_0'}\,r_{a_n}\delta_{a_n b_n} \, \Tr_{\,[-n,-1]}\Bigg( \bigotimes_{k=1}^n R_{a_k' b_k'}^{\dagger(k-n-1)}\otimes \bigotimes_{l=0}^{n-1} R_{a_l b_l}^{\dagger (l-n)} \,
		\mathbb{T}_n^\ddag \otimes \mathrm{id}_d^{\otimes n}\left[\ketbra{\sqrt{\rho_S^{\otimes[0,n-1]}}}\right]  \Bigg) .\label{pre_final_n}
	\end{multline}
	Finally, a more convenient spectrally equivalent expression for the coarse-grained matrix $\rho_S\left[\mc{F}^{(n+1)}\right]$ can be obtained by unitarily swapping factors in the following way:
	\begin{equation}
		\rho_S\left[\mc{F}^{(n+1)}\right]=\sum_{\bs{a} \,\bs{b}}        \ \rho_S\left[\mc{F}^{(n+1)}\right]_{\bs{a},\bs{b}}  \
		R_{a_0' b_0'} \otimes R_{a_n b_n} \otimes\,\bigotimes_{k=1}^n R_{a_k' b_k'}^{(k-n-1)}\otimes \bigotimes_{l=0}^{n-1} R_{a_l b_l}^{(l-n)} \,,
	\end{equation}
	so that, by substituting~\eqref{pre_final_n} into the latter, one obtains 
	$$
	\rho_S\left[\mc{F}^{(n+1)}\right]\,=\,\rho_S\otimes\rho_S \otimes \left(\mathbb{T}_n^\ddag \otimes \mathrm{id}_d^{\otimes n}\left[\ketbra{\sqrt{\rho_S^{\otimes n}}}\right]\right).
	$$

	\section{Proof of Proposition~\ref{prop:QRprop_NEW}.}\label{app:propQRnew}
	We need to prove that 
	$ \mathrm{QR} \iff \mathbb{T}_n=\bigotimes_{k=1}^n \Lambda_{k,k-1}\,.$
	Consider generic operators $\{A_k\}_{k=0}^n$ and $\{B_k\}_{k=0}^n$ in $\Md{d}$
	and evaluate
	\begin{align}
		&\sum_{\substack{\bs{i}_{[1,n-1]}\,,\ \bs{j}_{[1,n-1]}}} \Tr_{[1,n-1]} \!\Bigg(\ \bigotimes_{k=1}^{n-1}R_{i_kj_k}^{\dagger(k)} \ \mathbb{T}_n \Big[B_1^\dagger\, R_{i_1 j_1}\,A_1\otimes B_2^\dagger\, R_{i_2 j_2}\,A_2 \otimes 
		 \ldots \otimes B_{n-1}^\dagger\, R_{i_{n-1} j_{n-1}}\,A_{n-1}\otimes B_n^\dagger A_n\Big]\Bigg)\label{tn_togeneral}
		\\
		&=\sum_{\substack{\bs{i}_{[1,n-1]}\,,\ \bs{j}_{[1,n-1]}}} \Tr_{[1,n-1]}\!\Bigg(\ \bigotimes_{k=1}^{n-1} R_{i_k j_k}^{\dagger(k)}    \ \omega_E\bigg(\Theta_1\Big[B_1^\dagger R_{i_1 j_1}^{(0)} A_1\otimes 
		\Theta_{2,1}\Big[B_2^\dagger R_{i_2 j_2}^{(1)} A_2 \ldots
		\nonumber\\
		&\hskip+3.5cm \ldots \Theta_{n-1,n-2}\big[B_{n-1}^\dagger R_{i_{n-1}j_{n-1}}^{(n-2)} A_{n-1}\otimes  \Theta_{n,n-1}[B_n^\dagger A_n^{(n-1)}\otimes \mathds{1}_E]\big] \Big]\Big]  \bigg)\Bigg)\nonumber\\
		&=\omega_E\Bigg( \Theta_1\bigg[B_1^\dagger\otimes\mathds{1}_E\, \Theta_{2,1}\Big[B_2^\dagger \otimes \mathds{1}_E \ldots \Theta_{n,n-1}\left[B_{n}^{\dagger} A_{n}^{\phantom{\dagger}}\otimes \mathds{1}_E \right] \ldots  A_2 \otimes \mathds{1}_E\Big]A_{1}\otimes\mathds{1}_E\bigg]\Bigg)\,. \label{tn_togeneral2}
	\end{align}
	Similarly to~\eqref{pre-rewrite_1},
	we used the fact that, for a generic operator $Z_{SE}^{(k)}=\sum_{ij} R_{ij}^{(k)}\otimes Z_{ab}^{E}$, with the system part localized in $M_d^{(k)}(\mathbb{C})$,
	\begin{align*}
		\sum_{i_k j_k} B_{k}^\dagger R_{i_k j_k}^{(k-1)} A_k\otimes \Tr_{k}\left(R_{i_k j_k}^{\dagger (k)}  \, Z_{SE}^{(k)}\right)&= \sum_{i_k j_k} B_{k}^\dagger R_{i_k j_k}^{(k-1)} A_k \otimes Z^E_{i_k j_k}\\
		&\hskip-3.5cm =B_k^\dagger \otimes   \mathds{1}_E\,\left(\, \sum_{i_k j_k} R_{i_k j_k}^{(k-1)}\otimes Z^E_{i_k j_k}\right)\, A_k \otimes \mathds{1}_E=B_k^\dagger \otimes  \mathds{1}_E\; Z_{SE}^{(k-1)} \; A_k \otimes  \mathds{1}_E\,.
	\end{align*}
	In particular, if $\mathbb{T}_n=\bigotimes_k \Lambda_{k,k-1}$, \eqref{tn_togeneral} is also equal to
	\begin{align}
		\sum_{\substack{\bs{i}_{[1,n-1]} ,\ \bs{j}_{[1,n-1]}}} &\Tr_{[1,n-1]}\!\Bigg(\bigotimes_{k=1}^{n-1} R_{i_k j_k}^{\dagger (k)}\
		\Lambda_1 \left[B_1^\dagger R_{i_1 j_1}^{(0)} A_1\right] \otimes \Lambda_{2,1} \left[B_2^\dagger R_{i_2 j_2}^{(1)} A_2\right] \otimes \ldots
		\nonumber\\ & \hskip+4cm\ldots \otimes \Lambda_{n-1,n-2}\left[B_{n-1}^\dagger R_{i_{n-1} j_{n-1}}^{(n-2)} A_{n-1}\right]\otimes \Lambda_{n,n-1}\left[B_{n}^\dagger  A_n^{(n-1)}\right]
		\Bigg)\label{tn_QR0}\\
		&\hskip-1cm=
		\Lambda_1\Bigg[
		B_1^\dagger\, \Lambda_{2,1} \bigg[B_2^\dagger \ldots \Lambda_{n-1,n-2}\Big[B_{n-1}^\dagger \,\Lambda_{n,n-1}\big[B_n^\dagger A_n\big] \,A_{n-1}\Big]\ldots A_2\bigg] A_1\Bigg]\,,\label{tn_QR}
	\end{align}
	so that, by equating~\eqref{tn_togeneral2} and~\eqref{tn_QR}, one obtains the QR condition.
	Conversely, choose $B_k=R_{b_k b_k'}^\dagger$ and $A_k=R_{a_k a_k'}, k=1,\dots,n$, in~\eqref{tn_togeneral2} and~\eqref{tn_togeneral}. Then, set  $a_n=b_n$ and consider
	\begin{align}
		& \Tr\Bigg(R_{a_0 b_0} \omega_E\bigg(\Theta_1\Big[R_{b_1 b_1'}^\dagger\otimes \mathds{1}_E\,\Theta_{2,1}\, \Big[ R^\dag_{b_2b'_2}\otimes\mathds{1}_E
		\cdots\Theta_{n,n-1}\Big[R_{a_n b_n'}^\dagger R_{a_n a_n'} \otimes \mathds{1}_E\Big] \cdots 
		 R_{a_2a_2'}\otimes\mathds{1}_E\Big] R_{a_1 a_1'}\otimes \mathds{1}_E\Big]\bigg)\Bigg) \nonumber
		\\
		& \hskip+8cm = \Tr_{[0,n-1]} \!\left(\bigotimes_{k=0}^{n-1}R_{b_k  a_k}^{\dagger(k)} \ \mathbb{T}_n\left[ \bigotimes_{l=1}^n R_{b_l' a_l'}\right] \right) \,, \label{proofQR1_comparison}
	\end{align}
	and, choosing the same operators in~\eqref{tn_QR} and~\eqref{tn_QR0}, 
	\begin{multline}
		\Tr\Bigg(R_{a_0 b_0}\Lambda_1\Big[R_{b_1 b_1'}^\dagger \Lambda_{2,1}\, \Big[ R^\dag_{b_2b'_2}
		\cdots\Lambda_{n,n-1}\Big[R_{a_n b_n'}^\dagger R_{a_n a_n'} \Big] \cdots R_{a_2a_2'}\Big] R_{a_1 a_1'}\Big]\Bigg) 
		\\
		= \Tr_{[0,n-1]} \!\left(\bigotimes_{k=0}^{n-1}R_{b_k  a_k}^{\dagger(k)} \ \bigotimes_{j=1}^{n} \Lambda_{k,k-1}\left[ \bigotimes_{l=1}^n R_{b_l' a_l'}\right] \right) \,. \label{proofQR2_comparison}
	\end{multline}
	If QR holds,~\eqref{proofQR1_comparison} and~\eqref{proofQR2_comparison} are equal. This implies that $\mathbb{T}_n= \bigotimes_{j=1}^{n} \Lambda_{k,k-1}$.
	\qed
	
	\section{Proof of Proposition~\ref{prop:QRsg}.}
	\label{app:QRsg}
	We need to prove that
		$	\Lambda_{k,k-1}=\Lambda_{1}=:\Lambda$, $\forall\, k\ge1\,$.
		We proceed by induction and first 
		show it for $n=2$. Suppose that QR holds and consider again the special POVM $\mc{F}$ defined in~\eqref{specialOPU}, with associated CPU map $\mathbb{F}[X]=\sum_{ij} r_i R_{ji} X R_{ij}=\Tr(\rho_SX) \mathds{1}_d $. Accordingly, for $n=2$, one has
		\begin{align}\label{rhoF3_sg_noQR}
			\rho_S\left[\mathcal{F}^{(3)}\right]&=\omega_S\otimes\omega_E\left(F_{a_0a'_0}^{\dagger}\otimes \mathds{1}_E\,\Theta_{1}\left[F^{\dagger}_{b_1b_1'}\otimes \mathds{1}_E
			\,\Theta_{2,1}\left[ F_{b_{2}b_{2}'}^{\dagger} F_{a_{2}a'_{2} }\otimes\mathds{1}_E\right]\,F_{a_1a_1'}\otimes\mathds{1}_E\right]\,F_{a_0a'_0}\otimes \mathds{1}_E\right)
			\\ \label{rhoF3_sg_yesQR}
			&=\omega_S\left(
			F_{a_0a'_0}^{\dagger}\Lambda_1\left[F^{\dagger}_{b_1b_1'}\,\Lambda_{2,1} \left[F_{b_{2}b_{2}'}^{\dagger} F_{a_{2}a'_{2}}\right]\,F_{a_1a_1'}\right]\,F_{a_0a'_0}
			\right);
		\end{align}
		the latter is a density matrix acting on $\mathbb{C}^{d^2}\otimes\mathbb{C}^{d^2}\otimes\mathbb{C}^{d^2}$. 
		Consider then the marginal density matrix obtained from by tracing over the first Hilbert space from expression~\eqref{rhoF3_sg_noQR}. This amounts to setting $a_0=b_0$, $a'_0=b'_0$ and to summing over $a_0$ and $a'_0$.
		Due to invariance of the state under both $\Theta$ and the POVM 
		\begin{align}\nonumber
			\left(\Tr_{I} \rho\big[\mc{F}^{(3)}\big]\right)_{a_1 a_1', b_1 b_1'}&=
			\sum_{a_0 a'_0}\omega_S\otimes\omega_E\left(F_{a_0a'_0}^{\dagger}\otimes \mathds{1}_E\,\Theta_{1}\left[F^{\dagger}_{b_1b_1'}\otimes \mathds{1}_E
			\,\Theta_{2,1}\left[ F_{b_{2}b_{2}'}^{\dagger} F_{a_{2}a'_{2}}\otimes \mathds{1}_E\right]\,F_{a_1a_1'}\otimes \mathds{1}_E\right]\,F_{a_0a'_0}\otimes \mathds{1}_E\right)\nonumber\\
			&=\omega_S\otimes\omega_E\left(
			F^{\dagger}_{b_1b_1'}\otimes \mathds{1}_E
			\,\Theta_{2,1}\left[ F_{b_{2}b_{2}'}^{\dagger} F_{a_{2}a'_{2}}\otimes \mathds{1}_E\right]\,F_{a_1a_1'}\otimes \mathds{1}_E
			\right)\nonumber\\&=\omega_S\otimes\omega_E\Big(F^{\dagger}_{b_1b_1'}\otimes \mathds{1}_E
			\,\Theta_{1}\left[ F_{b_{2}b_{2}'}^{\dagger} F_{a_{2}a'_{2}}\otimes \mathds{1}_E\right]\,F_{a_1a_1'}\otimes \mathds{1}_E\Big)\,,
			\nonumber
		\end{align}
		where, in the last equality, we used the group assumption: $\Theta_{2,1}= \Theta_{1}=\Theta$. Thus, we have  $\Tr_{I}\left(\rho[\mc{F}^{(3)}]\right)=\rho[\mc{F}^{(2)}]$.
		Taking instead the partial trace over the first Hilbert space from expression~\eqref{rhoF3_sg_yesQR}, we have
		\begin{align*}
			\nonumber
			&\left(\Tr_{I} \rho\big[\mc{F}^{(3)}\big]\right)_{a_1 a_1', b_1 b_1'}=
			\sum_{a_0,a'_0}\Tr\left(\rho_S\, F_{a_0,a'_0}^{\dagger}\Lambda_1 \left[F^{\dagger}_{b_1,b_1'}\,\Lambda_{2,1} \left[F_{b_{2}b_{2}'}^{\dagger} F_{a_{2}a'_{2}}\right]\,F_{a_1a_1'}\right]\,F_{a_0a'_0}\right)\\
			\nonumber
			&=\Tr\left(\Lambda_1^\ddag \circ\mathbb{F}^\ddag[\rho_S] \,
			F^{\dagger}_{b_1b_1'}\,\Lambda_{2,1} \left[F^\dagger_{b_2b_2'} F_{a_2a'_2}\right]\, F_{a_1a_1'}	\right)		=\Tr\left(\rho_S\,
			F^{\dagger}_{b_1,b_1'}\,\Lambda_{2,1}\left[F^\dagger_{b_2b_2'} F_{a_2a'_2}\right]\, F_{a_1a_1'}\right),
		\end{align*}
		where, in the last equality, we used invariance of $\rho_S$ under $\mathbb{F}^\ddag$ and $\Lambda_1^\ddag$. \sloppy
		Then,  $\Tr_I\left(\rho\left[\mc{F}^{(3)}\right]\right)=\rho\left[\mc{F}^{(2)}\right]$ yields
		\begin{equation*}
			\Tr\left(\rho_S\, F^{\dagger}_{b_1b_1'}\,\Lambda_{2,1} \left[F_{b_{2}b_{2}'}^{\dagger} F_{a_{2}a'_{2}}\right]\,F_{a_1a_1'}\,\right)=\Tr\left(
			\rho_S\, F^{\dagger}_{b_1b_1'}\,\Lambda_1\left[F_{b_{2}b_{2}'}^{\dagger} F_{a_{2}a'_{2}}\right]\,F_{a_1a_1'}\,\right)\,.
		\end{equation*}
		Taking $a_2=b_2$, $a_1'=b_1'$, yields, in particular
		$$
			r_{a_2} r_{a_1'}	\sqrt{r_{a_1} r_{b_1}}  \bra*{r_{a_2'}}\Lambda_{2,1}^\ddag\big[\ketbra{r_{a_1}}{r_{b_1}}\big] \ket*{r_{b_2'}}=r_{a_2} r_{a_1'}	\sqrt{r_{a_1} r_{b_1}} \bra*{r_{a_2'}}\Lambda_{1}^\ddag\big[\ketbra{r_{a_1}}{r_{b_1}}\big] \ket*{r_{b_2'}},
		$$
		which implies $\Lambda_{2,1}=\Lambda_1$. Now, suppose $\Lambda_{k,k-1}=\Lambda_{1}=:\Lambda$ 
		holds for $k\le n-1$.
Choose again the POVM $\mc{F}$ as in~\eqref{specialOPU}.
Due to the assumed invariance of the $\omega_S\otimes\omega_E$ under both $\Theta$ and the chosen POVM, the state on the half-spin chain defined through the coarse-grained density matrices $\rho_S\big[\mc{F}^{(n+1)}\big], \, n\ge1 $ is invariant under the shift to the right.
On the other hand, explicit evaluation of $\rho\left[\mc{F}^{(n+1)}\right]$ yields
\begin{multline}    \label{second_rhoLambda}
	\rho\big[\mathcal{F}^{(n+1)}\big]=\sum_{\substack{a_0 b_0 \ldots  a_n  b_n \\ a_0' b_0' \ldots  a_n'  b_n'} }	\delta_{b_{n} a_n}    r_{a_{n}} \delta_{b_0' a_0'} r_{a_0'}  \prod_{k=1}^{n} \sqrt{r_{a_{k-1}}r_{b_{k-1}}}
	\bra*{r_{a_{k}'}}\Lambda^\ddag_{k,k-1}\left[\ketbra*{r_{a_{k-1}}}{r_{b_{k-1}}}\right]\ket*{r_{b_{k}'}}
	\\ \ketbra{r_{a_0} r_{a_0'} \ldots r_{a_n} r_{a_n'} }{r_{b_0} r_{b_0'} \ldots r_{b_n} r_{b_n'}}\,.
\end{multline}
Shift invariance means that
$\Tr_{I}{\rho\left[\mathcal{F}^{(n+1)}\right]}=\rho\left[\mathcal{F}^{(n)}\right]$. Here, $\Tr_{I}$ denotes the trace on the first two tensor factors in~\eqref{second_rhoLambda}, this amounts to setting $a_0=b_0$ and $a_0'=b_0'$ and then summing over $a_0, a_0'$. Consider thus the entries

\begin{align*}
	\Tr_{I}\big(\rho\big[\mathcal{F}^{(n+1)}\big]\big)_{a_1 a_1' a_n a_n' \ldots  b_1 a_1' a_n b_n'}&=
	r_{a_{n}}  \, \sum_{a_0'} r_{a_0'}  \bra{r_{a_{1}'}}\Lambda_{1}^\ddag\left[\rho_S\right]\ket{r_{b_{1}'}}  
	 \prod_{k=2}^n \sqrt{r_{a_{k-1}}r_{b_{k-1}}} \bra*{r_{a_{k}'}}\Lambda^		\ddag_{k,k-1}\left[\ketbra*{r_{a_{k-1}}}{r_{b_{k-1}}}\right]\ket*{r_{b_{k}'}} \\
	& \hskip-4cm =r_{a_{n}}  r_{a_1'}    \prod_{k=2}^{n-1} \sqrt{r_{a_{k-1}}r_{b_{k-1}}} \bra*{r_{a_{k}'}}\Lambda^\ddag\left[\ketbra*{r_{a_{k-1}}}{r_{b_{k-1}}}\right]\ket*{r_{b_{k}'}}   \bra*{r_{a_{n}'}}\Lambda^\ddag_{n,n-1}\left[\ketbra*{r_{a_{n-1}}}{r_{b_{n-1}}}\right]\ket*{r_{b_{n}'}} \,.
\end{align*}
The latter must be equal to 
\begin{align*}
	\rho\big[\mathcal{F}^{(n)}\big]_{a_1 a_1' a_n a_n' \ldots  b_1 a_1' a_n b_n'}=
	r_{a_{n}} 
	r_{a_1'}    \prod_{k=2}^{n} \sqrt{r_{a_{k-1}}r_{b_{k-1}}} \bra*{r_{a_{k}'}}\Lambda^\ddag\left[\ketbra*{r_{a_{k-1}}}{r_{b_{k-1}}}\right]\ket*{r_{b_{k}'}} \,,
\end{align*}
which  yields $
\bra*{r_{a_{n}'}}\Lambda^\ddag_{n,n-1}\left[\ketbra*{r_{a_{n-1}}}{r_{b_{n-1}}}\right]\ket*{r_{b_{n}'}}=\bra*{r_{a_{n}'}}\Lambda^\ddag\left[\ketbra*{r_{a_{n-1}}}{r_{b_{n-1}}}\right]\ket*{r_{b_{n}'}}
$
and, consequently, $\Lambda_{n,n-1}=\Lambda$.
\qed


	\section{Proof of Corollary~\ref{cor:partition}.}
	\label{proof:corollary_part}
	Consider again the dynamical system \eqref{supporting_dynsyst} as in the Proof of~Proposition~\ref{thm:partitionF_general}, now for a group of automorphisms $\Theta_n=\Theta^n$.
	Note that setting $A_{n-1}=\mathds{1}_d$ in~\eqref{def_proofn_2},
	\begin{equation}
		\mathbb{T}_n[A_0\otimes \ldots\otimes A_{n-2}\otimes \mathds{1}_d]= \mathbb{T}_{n-1}[A_0\otimes \ldots\otimes A_{n-2}]\otimes \mathds{1}_d\;,
	\end{equation}
	yields
	\begin{align*}
		&\Omega_S^{\mathbb{Z}\mathbb{Z}}\otimes \omega_E\left(\mathbb{T}_n\otimes \mathrm{id}_d^{\otimes n} [A_0\otimes \ldots \otimes A_{n-2}\otimes \mathds{1}_d\otimes A_0'\otimes \ldots \otimes A_{n-2}'\otimes \mathds{1}_d ]
		\right)\\
		&=\Tr_{[-n,-2]}\left(\Tr_{\{n-1\}} \left(
		\mathbb{T}_{n}^\ddag\otimes \mathrm{id}_d^{\otimes n} \left[\ketbra{\sqrt{\rho_S^{[0,n-1]}}}\right]\right) \bigotimes_{k=0}^{n-2} A_k \bigotimes_{l=0}^{n-2} A_l'
		\right) \\
&=\Tr_{[-n,-2]}\left(
		\mathbb{T}^\ddag_{n-1}\otimes \mathrm{id}_d^{\otimes n-1} \left[\ketbra{\sqrt{\rho_S^{[0,n-2]}}}\right] \bigotimes_{k=0}^{n-2} A_k \bigotimes_{l=0}^{n-2} A_l' \right)
	\end{align*}
	from which
	\begin{equation}\label{tnid_comp}
		\Tr_{\{-1\}}\left(\mathbb{T}^\ddag_{n}\otimes \mathrm{id}_d^{\otimes n} \left[\ketbra{\sqrt{\rho_S^{[0,n-1]}}}\right]\right)=\mathbb{T}^\ddag_{n-1}\otimes \mathrm{id}_d^{\otimes n-1} \left[\ketbra{\sqrt{\rho_S^{[0,n-2]}}}\right].
	\end{equation}
	On the other hand, one also notices that
	\begin{align*}
		\Tr_{\{0\}}\!\Bigg(\rho_S\,
		\mathbb{T}_{n}\left[\mathds{1}_d\otimes \bigotimes_{k=1}^{n-1} A_k\right] \Bigg)&=\omega_S
		\otimes\omega_E\Bigg(
		\Theta \bigg[
		\mathds{1}_d^{(0)} \otimes \Theta\Big[A_1^{(1)} \otimes  \ldots \Theta\big[A_{n-1}^{(n-1)}\otimes \mathds{1}_E
		\big]
		\Big]
		\bigg]
		\Bigg)\\
		&\hskip-1.5cm=\omega_E\bigg( \Theta\Big[A_1^{(1)} \otimes  \ldots \Theta\big[A_{n-1}^{(n-1)}\otimes \mathds{1}_E
		\big]
		\Big]
		\bigg)=\mathbb{T}_{n-1}[A_1\otimes \ldots\otimes A_{n-1}]\,.
	\end{align*}
	Hence,
	\begin{align*}
		&\Omega_S^{\mathbb{Z}\mathbb{Z}}\otimes \omega_E\left(\mathbb{T}_n\otimes \mathrm{id}_d^{\otimes n} [\mathds{1}_d\otimes A_1 \ldots \otimes A_{n-1}\otimes \mathds{1}_d\otimes A_1'\otimes \ldots \otimes A_{n-1}' ]
		\right)\\
		&=\Tr_{[1,n-1]}\!\left(\ketbra{\sqrt{\rho_S^{[1,n-1]}}}\Tr_{\{0\}}\!\Bigg(\rho_S\,
		\mathbb{T}_{n}\left[\mathds{1}_d\otimes \bigotimes_{k=1}^{n-1} A_k\right] \Bigg) \otimes  \bigotimes_{l=1}^{n-1} A_l'
		\right)\\
		&=\Tr_{[1,n-1]}\!\left(\ketbra{\sqrt{\rho_S^{[1,n-1]}}}
		\mathbb{T}_{n-1}\otimes \mathrm{id}_d^{\otimes n-1}[A_1\otimes \ldots \otimes A_{n-1} \otimes A_1'\otimes \ldots\otimes A_{n-1}']
		\right)\\
		&=\Tr_{[-n+1,-1]}\left(\mathbb{T}_{n-1}^\ddag \otimes \mathrm{id}_d^{\otimes n-1}\left[\ketbra{\sqrt{\rho_S^{[1,n-1]}}}\right]
		A_1\otimes \ldots \otimes A_{n-1} \otimes A_1'\otimes \ldots\otimes A_{n-1}'
		\right)\,,
	\end{align*}
	from which one has:
	\begin{equation}\label{tnid_staz}
		\Tr_{\{-n\}}\!\left(\mathbb{T}^\ddag_{n} \otimes \mathrm{id}_d^{\otimes n}\left[\ketbra{\sqrt{\rho_S^{\otimes [0,n-1]}}}\right]\right)=\mathbb{T}^\ddag_{n-1} \otimes\mathrm{id}_d^{\otimes n-1}\left[\ketbra{\sqrt{\rho_S^{\otimes [1,n-1]}}}\right]\,.
	\end{equation}
	We can now prove~\eqref{Lindblad_criterion}. By exploiting subadditivity of the von Neumann entropy along with~\eqref{tnid_comp} and~\eqref{tnid_staz}, one has:
	\begin{align}
		 S\left(\rho\big[\mc{F}^{(n+1)}\big]\right)&=  2\,S(\rho_S) 
		+
		 S\left(\mathbb{T}_n^\ddag \otimes\mathrm{id}_d^{\otimes n}\left[\ketbra{\sqrt{\rho_S^{\otimes n}}}\right] \right)\nonumber\\
		& \le 2\,S(\rho_S) + n\, S\left(\mathbb{T}_1^\ddag \otimes\mathrm{id}_d\left[\ketbra{\sqrt{\rho_S}}\right] \right) \label{ineq_compare_n}
		\\
		&=2\,S(\rho_S)+ n \,S\left(\Lambda^\ddag  \otimes \mathrm{id}_d\left[\ketbra{\sqrt{\rho_S}}\right]\right).\nonumber
	\end{align}
	Hence, by dividing both sides of the inequality by $n$ and taking the limit,
	\begin{equation}\label{bound_F_n}
		\lim_n 
		\frac{1}{n+1} S\big(\rho_S\left[\mc{F}^{(n+1)}\right]\big)  \le S\left(\Lambda^\ddag  \otimes \mathrm{id}_d\left[\ketbra{\sqrt{\rho_S}}\right]\right).
	\end{equation}
	Equality is achieved in~\eqref{ineq_compare_n}, and consequently in~\eqref{bound_F_n}, if and only if
		$\mathbb{T}_n =\bigotimes_{k=1}^n \mathbb{T}_1 =\bigotimes_{k=1}^n \Lambda \,,$
	corresponding to the QR regime by Propositions~\ref{prop:QRprop_NEW} and~\ref{prop:QRsg}.
	\qed
	
\end{widetext}
	
	\section{ALF entropy in the GNS representation.}\label{app:GNSbig}
	\subsection{Measurement model and entropy exchange.}\label{app:exchange}
	Consider a POVM on a $d$-level quantum system,
	\begin{equation}
		\mathcal{X}\, =\, \left\{X_a\right\}_{a=1}^{\abs{\mc{X}}}\; \subseteq\; \Md{d}\,,
	\end{equation}
	and the associated CPTP map 
	$
		\mathbb{X}^\ddagger[\rho]=\sum_a X_a \rho X_a^\dagger
	$.
	To a POVM $\mc{P}=\{P_a\}_a$ made of orthogonal projections, $P_a P_b=\delta_{a,b}P_a$, there corresponds a so-called projective (non-selective) measurement,
	\begin{equation}
		\mathbb{P}[Y]=\sum_a P_a Y P_a, \qquad \mathbb{P}\circ \mathbb{P}=\mathbb{P}\,,
	\end{equation}
	which is itself a projector on $\Md{d}$. Measuring orthogonal projectors $P_a$ can be reduced to distinguishing  the positions of a pointer on a grade. Less straightforward is sorting out the different labels $a$ in the case of measurements described by a generic POVM $\mathcal{X}=\{X_a\}_a$, whereby they must be read by measuring generically non-orthogonal components $X_a$. For this one needs first dilating  to a larger Hilbert space over which performing again an orthogonal measurement. This can be seen by means of the following simple model~\cite{MikeandIke}.
	Let $\ket{0}$ a fixed vector in~$\mathbb{C}^{|\mc{X}|}$ and
	define the isometry $V:\mathbb{C}^{d}\otimes \ket{0}\to \mathbb{C}^{d}\otimes \mathbb{C}^{\abs{\mc{X}}}$, 
	\begin{equation}\nonumber
		\label{meas-app0}
		V \ket{\psi}\otimes \ket{0} = \sum_{a=1}^{\abs{\mc{X}}} X_a \ket{\psi} \otimes \ket{a}\,. 
	\end{equation}
	\begin{equation}\label{measurement_dilation}
		\rho':=	V\rho\otimes \ketbra{0} V^\dagger=\sum_{a,b} X_a \rho X_b^\dagger \otimes \ketbra{a}{b}\,,\ \Tr_2\rho'=\mathbb{X}^\ddagger[\rho]\,.
	\end{equation}
	The isometry $V$ provides a Stinespring-like dilation of the map $\mathbb{X}^\ddagger$ and can be extended to a unitary on the whole $\mathbb{C}^d\otimes \mathbb{C}^{|\mc{X}|}$, so that~\eqref{measurement_dilation} can be seen as a joint evolution of system and measurement apparatus~\cite{MikeandIke}.
	Performing then a projective measurement $\mathbb{P}$, with $P_a=\ket{a}\bra{a}$, on the apparatus (which does not influence the state of the system) yields,
	\begin{equation}
		\mathrm{id}_d \otimes \mathbb{P}[\rho']=\sum_a X_a \rho X_a^\dagger \otimes \ketbra{a}{a}.
	\end{equation}
	Thus, measuring the projectors $\mathds{1}\otimes\ketbra{a}$, the post-measurement state of the system reads
	$$
	\rho_a'=\frac{X_a\rho X_a^\dagger }{\Tr(\rho X_a^\dagger X_a)}\,,
	$$
	with probability $p_a:=\Tr(\rho X_a^\dagger X_a)$. The non-selective version of the measurement corresponds indeed to the action of the CPTP map $\rho\mapsto\mathbb{X}^\ddagger[\rho]=\Tr_2(\rho')$. 
	Instead, tracing over the degrees of freedom of the 
	apparatus yields:
	\begin{equation}
		\label{meas-app}
		\rho_2'=
		\Tr_{1}(\rho')= \sum_{a,b=1}^{\abs{\mc{X}}} \Tr(\rho X_b^\dagger X_a )\ketbra{a}{b} =:\rho[\mc{X}]\,.
	\end{equation}
	Its diagonal entries correspond to the outcome probabilities relative to the projective measurement performed on the apparatus
	$p_a$,
	while the off-diagonal ones
	are correlation terms contain further information about the apparatus. In fact, the von Neumann entropy of a density matrix is always less or equal to the Shannon entropy of its diagonal, so that
	\begin{equation}\label{decoh_exch}
		S\left(\rho[\mc{X}]\right) \ \le\ H\left(\big\{p_a\big\}_{a=1}^{\abs{\mc{X}}}\right)\ ,
	\end{equation}
	the inequality being saturated when the POVM $\mc{X}$ is made of orthogonal projections. The quantity  $S(\rho[\mc{X}])$ is sometimes called~\emph{entropy exchange}~\cite{Schumacher96,MikeandIke}. It represents the entropy increase in the apparatus, initially in the 
	pure state $\ketbra{0}$ and thus at zero entropy, due to the measurement process that sends it into the mixed state $\rho[\mathcal{X}]$.
	When the state $\rho$ of the system is pure, so is  the compound state $\rho'$ in~\eqref{measurement_dilation} emerging from the measuring unitary interaction; then, its two marginals ${\rm Tr}_2\rho'=\mathbb{X}^\ddag[\rho]$ (see~\eqref{measurement_dilation}) and ${\rm Tr}_1\rho'=\rho[\mathcal{X}]$
	(see~\eqref{meas-app}) have the same von Neumann entropy, whence the energy growth in the apparatus, $S(\rho[\mathcal{X}])$, equals the entropy growth in the system, $S(\mathbb{X}^\ddagger[\rho])$, due to the non-selective measurement process $\rho\mapsto\mathbb{X}^\ddag[\rho]$.
	We now show, following~\cite{Lindblad1973}, that the entropy gained by the measurement apparatus is at least as large as the information gained during the measurement. Recalling the definition of quantum relative entropy $S(\rho\|\sigma)=\Tr(\rho\,(\log\rho-\log\sigma))$, consider
	\begin{equation}
		S(\rho'\|\rho_1'\otimes\rho_2')=S(\mathbb{X}^\ddagger[\rho])+S(\rho[\mc{X}])-S(\rho)\,,
	\end{equation}
	where we used the fact that, due to~\eqref{measurement_dilation}, $S(\rho')=S(\rho)$.
	Further, consider the action of the full decoherence map $\mathbb{P}$ on the second party,
	\begin{equation}
		\mathrm{id}\otimes \mathbb{P}[\rho']
		=\sum_{a=1}^{|\mc{X}|}\, p_a\, \rho_a'\otimes \ketbra{a}, \qquad  \rho_a'=\frac{X_a\rho X_a^\dagger}{p_a}\,,
	\end{equation}
	so that
	\begin{equation}
		S(\mathrm{id}\otimes\mathbb{P}[\rho']\|\mathrm{id}\otimes\mathbb{P}[\rho'_1\otimes \rho'_2]) =S(\mathbb{X}^\ddagger[\rho])- \sum_a p_a S(\rho_a')\,.
	\end{equation}
	Hence, due to the monotonicity of the relative entropy under CPTP maps, the difference 
	\begin{multline*}
		S\left(\mathrm{id}\otimes\mathbb{P}[\rho']\middle\|\mathrm{id}\otimes\mathbb{P}[\rho_1'\otimes\rho_2']\right)-S\left(\rho'\middle\|\rho_1'\otimes\rho_2'\right)\\
		=S(\rho)-\sum_a p_a S(\rho_a')-S(\rho[\mc{X}])\le0\,,
	\end{multline*}       
	or, equivalently,
	\begin{equation}\label{information_gain_upperbound}
		S(\rho)-\sum_a p_a S(\rho_a') \, \le\, S(\rho[\mc{X}]) \,.
	\end{equation}                             
	On the left-hand side, one has the average information gain about the state $\rho$ when the measurement associated with $\mathcal{X}$ is performed. Hence, the entropy exchange $S(\rho[\mc{X}])$ provides a universal upper bound for the average information gain.

		\section{Proof of Proposition~\ref{prop:model_Tn}.}
	\label{app:model_Tn}
	From its definition~\eqref{thm1_1} in Proposition~\ref{thm:partitionF_general}, the map $\mathbb{T}^\ddag_n$ reads
	\begin{align*}		&\mathbb{T}_n\left[\bigotimes_{k=0}^{n-1} A_k^{(k-n)} \right]\\&=\omega_E\Bigg(\big(\Theta\circ\sigma_S\otimes\mathrm{id}_E\big)^n\left[\bigotimes_{k=0}^{n-1} A_{k}^{(k-n)}\otimes \mathds{1}_E \right]\Bigg)\, ,
	\end{align*}
	with $ \Theta=\mathrm{id}_S \otimes \sigma_E\circ \Phi$.	Note that in the one-step automorphism $\Theta\circ\sigma_S\otimes\mathrm{id}_E$ the shift appears twice: 
	the right-most one, $\sigma_S$, acts on the chain $M_d^{\mathbb{Z}}(\mathbb{C})$ of infinite copies of the open system (as in Proposition~\ref{thm:partitionF_general}), while
	$\sigma_E$ is the collisional shift on the spin chain environment.
	Then, let us first evaluate
	\begin{multline*}\nonumber
		\big(\Theta\circ\sigma_S\otimes\mathrm{id}_E\big)\left[\bigotimes_{k=0}^{n-1} A_{k}^{(k-n)}\otimes \mathds{1}_E \right]\\=\bigotimes_{k=0}^{n-2} A_k^{(k-n+1)}\otimes \sum_{i_n} \phi_{i_n}\big[A_{n-1}^{(0)}\big] \otimes \Pi_{i_n}^{(1)}\,,
	\end{multline*}
	where $\Pi_{i_n}$ is localized in the first site of the environment chain due to the action of $\sigma_E$. The second iteration leads to
	\begin{align*}\nonumber
		\big(\Theta\circ\sigma_S\otimes\mathrm{id}_E\big)^2&\left[\bigotimes_{k=0}^{n-1} A_{k}^{(k-n)}\otimes \mathds{1}_E \right]=\\ 
		&\hskip-2cm=\bigotimes_{k=0}^{n-3} A_{k}^{(k-n+2)}\otimes \sum_{i_{n-1} i_n}  \phi_{i_{n-1}}\otimes \phi_{i_{n}}\left[A_{n-2}^{(0)}\otimes A_{n-1}^{(1)}\right]\\
		&\hskip+3cm\otimes\Pi_{i_{n-1}}^{(1)}\otimes \Pi_{i_n}^{(2)}\,.
	\end{align*}
	Hence, after $n$ iterations:
	\begin{align}\label{n_ite_concreteTn}
		&\big(\Theta\circ\sigma_S\otimes\mathrm{id}_E\big)^n\left[\bigotimes_{k=0}^{n-1} A_{k}^{(k-n)}\otimes \mathds{1}_E\right]
		\\
		&\hskip+2cm=\sum_{\bs{i}_{[1,n]}} \phi_{\bs{i}_{[1,n]}}^{\otimes[1,n]}\left[\bigotimes_{k=0}^{n-1} A_k^{(k)}\right]  \otimes
		\Pi_{\bs{i}_{[1,n]}}^{[1,n]} \,. \nonumber
	\end{align}
	Acting on the latter with the conditional expectation $\omega_E$, we get
	\begin{align*}
		\mathbb{T}_n\left[\bigotimes_{k=0}^{n-1} A_k^{(k)} \right]&= \sum_{\bs{i}_{[1,n]}} \phi_{\bs{i}_{[1,n]}}^{\otimes[1,n]}\left[\bigotimes_{k=0}^{n-1} A_k^{(k)}\right]     \omega_E\left(	\Pi_{\bs{i}_{[1,n]}}^{[1,n]}\right)\\&=\sum_{\bs{i}_{[1,n]}} p_{\bs{i}_{[1,n]}} \,\phi_{\bs{i}_{[1,n]}}^{\otimes[1,n]}\left[\bigotimes_{k=0}^{n-1} A_k^{(k)}\right] \,.
	\end{align*}
	With respect to the tracial state $\rho_S=\mathds{1}_d/d$,
	the coarse-grained density matrix then reads:
	\begin{align*}
		&\rho_S\left[\mathcal{F}^{(n+1)}\right]=\frac{\mathds{1}_d}{d}\otimes\frac{\mathds{1}_d}{d}\otimes\\
		&\hskip+3cm\otimes\bigg( \mathbb{T}_{n}^\ddag\otimes\mathrm{id}_{d}^{\otimes n}\left[\ketbra{\psi^{[1,n]}_+} \right]\bigg)\,,
		\\ 
		&\hbox{where} \  \ket{\psi^{[1,n]}_+}= \ket{\sqrt{\frac{\mathds{1}}{d}^{\otimes n}}}=d^{-\frac{n}{2}}\!\sum_{\bs{k}_{[1,n]}}\ket{\bs{k}_{[1,n]}\otimes \bs{k}_{[1,n]}} \nonumber
	\end{align*}
	and where $\mathbb{T}_n^\ddag$ is the dual of $\mathbb{T}_n$ in the Schr\"odinger picture, namely the map
	\begin{equation}\label{inter_TnShro}
		\mathbb{T}_n^\ddag=\sum_{\bs{i}_{[1,n]}} p_{\bs{i}_{[1,n]}}\, \phi_{\bs{i}_{[1,n]}}^{\otimes[1,n]\,\ddag}\,, \qquad 
		\phi_{k}^\ddag[\,\cdot\,]=U_k\, \cdot\,U_k^\dag\,.
	\end{equation}
	We now prove \eqref{statement_whenequalityholds}.
	First, by setting $\ket{\psi_{\bs{i}_{[1,n]}}}=\left(\bigotimes_{k=1}^n U_{i_k}\right)\otimes\mathds{1}_d^{\otimes n}\ket{\psi_+^{[1,n]}}$ one has
	\begin{equation}\label{pseudo_spectral}
		\mathbb{T}_n^\ddag\otimes\mathrm{id}_{d}^{\otimes n}\left[\ketbra{\psi_+^{[1,n]}}\right]=\sum_{\bs{i}_{[1,n]}}\,p_{\bs{i}_{[1,n]}} \ketbra{\psi_{\bs{i}_{[1,n]}}} ,
	\end{equation}	
	In particular, since the vector states $\ket{\psi_{\bs{i}_{[1,n]}}}$ are normalized,
	\begin{align}\nonumber
		S\left(\rho\left[\mc{F}^{(n+1)}\right]\right)&=2\,S(\rho_S)\\
		& \hskip0.5cm +S\left(\mathbb{T}_n^\ddag\otimes\mathrm{id}_{d}^{\otimes n}\left[\ketbra{\psi_+^{[1,n]}}\right]\right)\ \nonumber\\
		&\le\ 2\,S(\rho_S)+H\left(\pi_{[1,n]}\right)\, , \label{estimate_pure}
	\end{align}
	with equality holding if and only if $\ket{\psi_{\bs{i}_{[1,n]}}}$ form an orthonormal basis in $\mathbb{C}_{d}^{\otimes n}$, that is, when~\eqref{pseudo_spectral} is the spectral decomposition of $\mathbb{T}_n^\ddag\otimes\mathrm{id}_{d}^{\otimes n}\left[\ketbra{\psi_+^{[1,n]}}\right]$. Equivalently, the equality holds if and only if the unitary operators $U_{k}$ satisfy the orthogonality relations
	$	\Tr(U_{j}^\dagger U_{k}^{\phantom{\dagger}})=d \,\delta_{jk}
	$.
	In such case, dividing both sides of~\eqref{estimate_pure} by $n+1$ and taking the limit, one gets
	\begin{equation}
		\mathfrak{h}_{S}(\Theta,\mc{F})=\lim_{n}\frac{1}{n}H\left(\pi_{[1,n]}\right)=\mathfrak{S}_{\omega_E}\,,
	\end{equation}
	that, along with~\eqref{deupperbound}, implies~\eqref{statement_whenequalityholds}.
	\qed

	\subsection{Derivation of equation~\eqref{equivalenceGNS}.}
	
	\label{app:GNS_equivalence}
	We now derive  equation~\eqref{equivalenceGNS}.
	For a generic POVM $\mathcal{X}$, let $\{\ket{a}\}_{a=1}^{\abs{\mc{X}}}$ be a basis of $\mathbb{C}^{|\mc{X}|}$ and 
	define, similarly to Appendix~\ref{app:exchange}, an isometry $V:\mathcal{H}_{\omega}\otimes \ket{0}\longrightarrow\mathcal{H}_\omega\otimes \mathbb{C}^{ |\mathcal{X} | }$
	\begin{equation}\nonumber
		V\ket{\Omega_\omega}\otimes \ket{0}=\sum_{a=1}^{\abs{\mc{X}}} \pi_\omega(X_a) \ket{\Omega_\omega}\otimes \ket{a}=:\ket{\psi}\ \in \ \mc{H}_\omega\otimes\mathbb{C}^{|\mc{X}|}\,,
	\end{equation}
	corresponding to the projector,
	\begin{multline}    \label{dilation_gns}
		\ketbra{\psi}= V \ketbra{\Omega_\omega}\otimes\ketbra{0} V^\dagger\\ =\sum_{ab=1}^{|\mc{X}|} \pi_\omega(X_a) \ketbra{\Omega_\omega}\pi_\omega(X_b)^\dagger\otimes  \ketbra{a}{b}\,.
	\end{multline}
	The marginals of $\ketbra{\psi}$ yield
	\begin{equation}\label{prop:tr1}
		\Tr_{I}(\ketbra{\psi})=\sum_{a,b} \omega\left(X_b^\dagger X_a\right) \ketbra{a}{b}=\rho[\mc{X}]\,,
	\end{equation}
	and 
	\begin{multline}\label{prop:tr2}
		\Tr_{II}(\ketbra{\psi})=\sum_{a} \pi_\omega(X_a) \ketbra{\Omega_\omega}\pi(X_a)^\dagger\\=:\widetilde{\mathbb{X}}^\ddagger[\ketbra{\Omega_\omega}]\, .
	\end{multline}
	\color{black}
	The marginals~\eqref{prop:tr1} and~\eqref{prop:tr2} of the pure state $\ketbra{\psi}$ have the same spectrum, multiplicities included, a part from the zero eigenvalue; thus, they have the same von Neumann entropy,
	$
	S(\rho[\mc{X}]) = S\left(\widetilde{\mathbb{X}}^\ddagger[\ketbra{\Omega_\omega}]\right)
	$.
	From the latter expression, it is also evident how the entropy exchange discussed in Appendix~\ref{app:exchange} does depend on the chosen POVM, but not on its specific  Kraus representation.
	Repeating the same steps for the time-refined POVM $\mc{X}^{(n)}$, one extends~\eqref{dilation_gns} to
	\begin{equation}\label{multi-time_dilation}
		\sum_{\bs{a},\bs{b}} \pi_\omega\left(X_{\bs{a}}^{(n)}\right) \ketbra{\Omega_\omega}
		\pi_\omega\big(X_{\bs{b}}^{(n)\dagger}\big) \otimes \ketbra{\bs{a}_{[0,n-1]}}{\bs{b}_{[0,n-1]}}.
	\end{equation}
	so that
	\begin{equation}\label{prop:n_equal_entropy}
		S\big(\rho[\mathcal{X}^{(n)}]\big)=S\left(\sum_{\bs{a}} \pi \big(X_{\bs{a}}^{(n)}\big)\ketbra{\Omega_\omega}\big(X_{\bs{a}}^{(n)}\big)^\dagger\right).
	\end{equation}
	From~\eqref{time_partitionelement},
	\begin{align*}
		&\pi\left(X_{\bs{a}}^{(n)}\right)\ket{\Omega_\omega}\\&=  \pi_\omega(\Theta^{n-1}(X_{a_{n-1}}))\,\dots\,\pi_\omega(\Theta(X_{a_1}))\pi(X_{a_0})\ket{\Omega_\omega}\\
		&=(U^{\dagger })^{n} \left(U \pi_\omega(X_{a_{n-1}}) \ldots  U\pi_\omega(X_{a_{0}}) \right)\ket{\Omega_\omega}.
	\end{align*}
	Then,
	\begin{multline*}
		\sum_{\bs{a}} \pi \big(X_{\bs{a}}^{(n)}\big)\ketbra{\Omega_\omega}\big(X_{\bs{a}}^{(n)}\big)^\dagger\\=\mathbb{U}^{n} \circ\left(\mathbb{U}^\ddagger\circ\widetilde{\mathbb{X}}^\ddagger\right)^{n}[\ketbra{\Omega_\omega}]\,,
	\end{multline*}
	so that by the invariance of the von Neumann entropy under unitary maps and~\eqref{prop:n_equal_entropy},
	\begin{equation}\label{app:commonentropy_equivalence}S\left(\rho\big[\mc{X}^{(n)}\big]\right)=S\big(\big(\mathbb{U}^\ddagger\circ\widetilde{\mathbb{X}}^\ddagger\big)^{n}\,\big[\ketbra{\Omega_\omega}\big]\big).
	\end{equation}
	Following the model of measurement process described in Appendix~\ref{app:exchange},~\eqref{multi-time_dilation} can be then interpreted as a dilation of the POVM process $\big(\mathbb{U}^\ddagger\circ\widetilde{\mathbb{X}}^\ddagger\big)^{n}$ where the apparatus is explicitly taken into account. The only difference is that, now, multiple copies of the apparatus are considered due to the repeated measurements.
	Accordingly,~\eqref{app:commonentropy_equivalence} can be also seen as the entropy exchange with the apparatus due to iterated measurements.

	\section{$\Delta$-monotonicity of RHP and BLP measure.}\label{app:discretized_measures}
	\paragraph*{Discretized RHP measure.}
	To obtain a meaningful analogue of the RHP measure in the our discrete setting, we set a reference time scale  $N\in\mathbb{N}$ and consider the following  discretized version of the RHP measure:
	\begin{align*}
		\mathcal{N}_{RHP}^{(N)}=\sum_{n=1}^{N} g_n\,, \quad g_n&=\normt{\Lambda^{\ddag}_{n,n-1}\otimes\mathrm{id}_d\left[\pplus{d}\right]}-1\\
		&=\sum_{i}\abs {E^{(i)}_{n,n-1}}-1\,,
	\end{align*}
	where $E^{(i)}_{n,n-1}$ are the eigenvalues of $\Lambda^{\ddag}_{n,n-1}\otimes\mathrm{id}_d\left[\pplus{d}\right]$. For the one-qubit model at hand, one has
	\begin{equation}\nonumber
		g_n=\lambda_{n,n-1}-(1-2p)\,,
	\end{equation}
	where $\lambda_{n,n-1}=\lambda_{n}/\lambda_{n-1}$. We shall consider, as in the main text and~\cite{FBGN_PhysicaScr},  the case $A_{p,r}>0$. Then,  $\lambda_n$ is positive and is given by 
	\begin{align*}
		\lambda_{n,n-1}
		&=\frac{A_{p,r}}{2} \frac{(1+x_{p,r,\Delta})^{n+1}-(1-x_{p,r,\Delta})^{n+1}}{(1+x_{p,r,\Delta})^{n}-(1-x_{p,r,\Delta})^{n}}\,, 
	\end{align*}
	where we let $x_{p,r,\Delta}:={B_{p,r,\Delta}}/{A_{p,r}}\ge 1$.
	 to prove that
	\begin{equation}\label{the_derint}
		\pdv{g_n}{\Delta} =	\pdv{\lambda_{n,n-1}}{\Delta}= \pdv{\lambda_{n,n-1}}{x}\pdv{x_{p,r,\Delta}}{\Delta}\ge0\,, \qquad \forall \, n\in \mathbb{N}\,.
	\end{equation}
	Note that
	$\displaystyle	\pdv{x_{p,r,\Delta}}{\Delta}=\frac{8 p}{A_{p,r}\, B_{p,r,\Delta}} \ge0$. Then, letting $y:=x_{p,r,\Delta}^{-1}$, $0\le y\le 1$, one has
	\begin{equation}
		\pdv{\lambda_{n,n-1}}{x}
	\label{thelambdaderivative}	
		=\frac{(1+y)^{2n}-(1-y)^{2n}+(-)^n 4ny\big(1-y^2\big)^{n-1}}{\left((1+y)^{n}-(1-y)^{n}\right)^2}\,, 
	\end{equation}
	 To verify the positivity of~\eqref{thelambdaderivative},
	it is sufficient to check the positivity of
	\begin{align}
		f_n(y):&=(1+y)^{2n}-(1-y)^{2n}- 4ny	\big(1-y^2\big)^{n-1}. 
	\end{align}
	By the binomial expansion, one rewrites
	\begin{align}\nonumber
		f_n(y)
		=4\, n\,y \left(1-\big(1-
		y^2\big)^{n-1}\right)+2\sum_{k=1}^{n-1}\binom{2n}{2k+1} y^{2k+1}.\nonumber
	\end{align}
	Since $0\le y\le1$, it follows that $f_n(y)\ge0$ for all $n\in\mathbb{N}$. This suffices to establish the monotonicity of $\lambda_{n,n-1}$, $g_n$, and $\mathcal{N}_{RHP}^{(N)}$ as functions of~$\Delta$.

	\paragraph*{Discretized BLP measure.} Consider the following discretized version of the BLP measure:
	\begin{align}\nonumber
		&\mathcal{N}_{BLP}^{(N)}=\max_{\rho,\,\sigma}\, \sum_{\substack{n=1 \\  d_n-d_{n-1}> 0}}^N \left(d_n(\rho,\sigma)-d_{n-1}(\rho,\sigma) \right)\,, \\ &\hbox{where}\quad d_n(\rho,\sigma):=\normt{\Lambda_n^\ddagger\left[\frac{\rho-\sigma}{2}\right]}\,.\label{discrete_BLP}
	\end{align}
	With respect  to our qubit model, we shall now show that
	$d_n(\rho,\sigma)-d_{n-1}(\rho,\sigma)$ is  a non-decreasing  function of $\Delta$ for all $n\in\mathbb{N}$ and all $\rho$, $\sigma$. Let $X=(\rho-\sigma)/2$. Since $\Tr(X)=0$, we can expand it as $X=\sum_{i=1}^3 x_i \, \sigma_i$.
	Then,
	\begin{equation}
		d_n(\rho,\sigma)=\sqrt{\lambda_n^2\,\big(x_1^2+x_2^2\big)+\mu_n^2 \, x_3^2}\,, 
	\end{equation}
	where, for the sake of brevity, we set $\lambda_n:=\lambda_n^{(i)}$, $i=1,2$, $\mu_n:=\lambda_n^{(3)}$.
	Since $d_n(\rho,\sigma)$ depends on $\Delta$ only through $\lambda_n$, we have
	\begin{equation*}
		\pdv{\Delta}d_n(\rho,\sigma)=\big(x_1^2+x_2^2\big) \frac{\lambda_n'\lambda_n}{d_n}\,, \qquad \lambda_n'=\pdv{\lambda_n}{\Delta}\,,
	\end{equation*}
	and, accordingly,
	\begin{equation*}
		\pdv{\Delta}\big(d_n(\rho,\sigma)-d_{n-1}(\rho,\sigma)\big)=\big(x_1^2+x_2^2\big) \left(\frac{\lambda_n'\lambda_n}{d_n}-\frac{\lambda_{n-1}'\lambda_{n-1}}{d_{n-1}}\right).
	\end{equation*}
	The latter is positive if and only if
	\begin{equation}\label{ineqBLPdiscr}
		\lambda_{n,n-1}\frac{\lambda_n'}{\lambda_{n-1}'} \ge \frac{d_n}{d_{n-1}}\,.
	\end{equation}
	For the model at hand, lack of contractivity of $\Lambda_{n,n-1}^\ddagger$ occurs only if $\lambda_{n,n-1}>1$. We previously proved that  $\lambda_{n,n-1}'>0$. These  facts are sufficient to show that
	\begin{equation}\label{ineq_prime}
		\lambda_{n}'-\lambda_{n-1}'= (\lambda_{n,n-1}-1) \lambda_{n-1}'+\lambda_{n,n-1}'\lambda_{n-1} \ge0 \,.
	\end{equation}
	From $\mu_n\le \mu_{n-1}$ one further  deduces that
	\begin{align}
		d_n &\le \sqrt{\lambda_n^2 (x_1^2+x_2^2) +\mu_{n-1}^2 x_3^2} \label{ineq_dn} \\
		&<\sqrt{\lambda_{n,n-1}^2\lambda_{n-1}^2 (x_1^2+x_2^2) +\lambda_{n,n-1}^2\mu_{n-1}^2 x_3^2}=\lambda_{n,n-1}\,d_{n-1}\,.\nonumber
	\end{align}
	Inequalities~\eqref{ineq_prime} and~\eqref{ineq_dn} yield, finally,
	\begin{equation*}
		\lambda_{n,n-1}\frac{\lambda_n'}{\lambda_{n-1}'} \ge\lambda_{n,n-1} > \frac{d_n}{d_{n-1}}\,,
	\end{equation*}
	thus verifying condition~\eqref{ineqBLPdiscr} and ensuing monotonicity of $d_n(\rho,\sigma)-d_{n-1}(\rho,\sigma)$ as a function of $\Delta$.

	\color{black}
	
	\section{Derivation of GNS reduced dynamics.}\label{app:GNSdissipative}
 We check that~\eqref{Udag_def} is a bona fide GNS implementation of the collisional dynamics.
First notice that, in~\eqref{Udag_def},  the action on $\pi_S(\mc{A}_S)'$ is such that the GNS cyclic vector is invariant; namely,
\begin{align}
	&U_\Theta^\dagger\ket{\psiplus{d}\otimes\Omega_E}\nonumber
	\\&=\sum_k U_k^\dagger  \otimes \overline{U_k}^\dagger \ket{\psiplus{d}} \otimes  U_\sigma^\dagger\pi_E\left(\Pi_k^{(0)}\right) \ket{\Omega_E}\nonumber\\
	&=\ket{\psiplus{d}}\otimes U_\sigma^\dagger\,\pi_E\left(\sum_k\Pi_k^{(0)}\right)\ket{\Omega_E}=\ket{\psiplus{d} \otimes\Omega_E}\,, \label{Uinvariance}
\end{align}
where we used the fact that, along with~\eqref{Ushift},  $V\otimes \overline{V}\ket{\psiplus{d}}=\ket{\psiplus{d}}$ for any unitary $V\in\Md{d}$. 
Also, it suffices to define $U_\Theta$ only on $\pi_E(\mathcal{A}_E)$. Indeed, because the environment is a classical chain, this dense subalgebra 
$\pi_E(\mathcal{A}_E)$
is contained in its commutant $\left(\pi_E(\mathcal{A}_E)\right)'$.
By taking the adjoint with respect to the scalar product in the GNS Hilbert space, the action of $U_\Theta$ can be also inferred,
\begin{align}\label{U_def}
	&U_\Theta\,\pi_S(A) \pi_S'(B)\otimes\pi_E\left(\Pi_{\bs{i}^{[a,b]}}^{[a,b]}\right)\ket{\psiplus{d}\otimes\Omega_E}\\&=\sum_k U_k A \otimes \overline{U_k} \, \overline{B} \otimes  \pi_E\big(\Pi_k^{(0)}\big) U_\sigma\,\pi_E\big(\Pi_{\bs{i}^{[a,b]}}^{[a,b]}\big) \ket{\psiplus{d}\otimes\Omega_E}\, \nonumber
\end{align}
from which it also follows that $U_\Theta\ket{\psiplus{d}\otimes\Omega_E}=\ket{\psiplus{d}\otimes\Omega_E}$. By acting on local operators as in \eqref{Udag_def} and \eqref{U_def}, one checks explicitly that $U_\Theta^{\phantom{\dagger}} U_\Theta^\dagger=U_\Theta^\dagger U_\Theta^{\phantom{\dagger}}=\mathds{1}$. Moreover, the automorphism $\Theta$ is correctly implemented. Indeed, from~\eqref{Udag_def},
\begin{align*}\label{check_implementation}
	&U_\Theta^\dagger\,\pi_S(A) \otimes\pi_E\left(\Pi_{\bs{i}^{[a,b]}}^{[a,b]}\right)U_\Theta\ket{\psiplus{d}\otimes\Omega_E}\\
	&=\sum_{k,l} U_k^\dagger A U_l \otimes \overline{U_k}^\dagger \overline{U_l} \otimes  U_\sigma^\dagger\,\pi_E\left(\Pi_k^{(0)}\Pi_{\bs{i}^{[a,b]}}^{[a,b]} \Pi_l^{(0)}\right)\nonumber\\
	&\hskip+3cm \mathds{1}_d\otimes\mathds{1}_d\otimes U_\sigma \ket{\psiplus{d}\otimes\Omega_E}\nonumber\\
	&=\mathds{1}_d\otimes\mathds{1}_d\otimes U_\sigma^\dagger\nonumber\\&\qquad \bigg(\sum_k U_k^\dagger A U_k\otimes\mathds{1}_d\otimes \pi_E\left(\Pi_k^{(0)} \Pi_{\bs{i}^{[a,b]}}^{[a,b]} \Pi_k^{(0)}\right) \bigg)  \nonumber\\&\hskip+4.cm  \mathds{1}_d\otimes\mathds{1}_d\otimes U_\sigma\ket{\psiplus{d}\otimes\Omega_E}\nonumber\\
	&=\mathds{1}_{d^2}\otimes U_\sigma^\dagger\pi_S\otimes\pi_E\left(\Phi[A\otimes\Pi_{\bs{i}^{[a,b]}}^{[a,b]}]\right)\nonumber\\&\hskip+4.cm \mathds{1}_{d^2}\otimes U_\sigma \ket{\psiplus{d}\otimes\Omega_E} 
	\nonumber\\
	&=\pi_S\otimes\pi_E\left(\sigma_E\circ\Phi[A\otimes\Pi_{\bs{i}^{[a,b]}}^{[a,b]}]\right)\ket{\psiplus{d}\otimes\Omega_E}\nonumber\\
	&=\pi_S\otimes\pi_E\left(\Theta[A\otimes\Pi_{\bs{i}^{[a,b]}}^{[a,b]}]\right) \ket{\psiplus{d}\otimes\Omega_E}\nonumber.
\end{align*}
\begin{widetext}
Recalling that $\widetilde{\mathbb{X}}=\mathbb{X}\otimes \mathrm{id}_d$,  consider the operator:
$$
\mathbb{Y}_m:=\left(\left(\wt{\mathbb{X}}\otimes\mathrm{id}_E\right)\circ \mathbb{U}_\Theta \right)^m[\pi_S\otimes\pi_E(A\otimes\mathds{1}_E) \pi_S'\otimes\pi_E(B\otimes\mathds{1}_E)  ]=\left(\left(\wt{\mathbb{X}}\otimes\mathrm{id}_E\right)\circ \mathbb{U}_\Theta \right)^m[A\otimes \overline{B} \otimes \pi_E(\mathds{1}_E) ]\,, \quad m\in\mathbb{N}.
$$
We now show by induction that
\begin{equation}\label{supportope}
	\mathbb{Y}_m=\sum_{k_1\cdots k_m} \wt{\mathbb{X}}\circ \phi_{k_1}\otimes \overline{\phi}_{k_1}\circ\wt{\mathbb{X}}\circ \phi_{k_2} \otimes  \overline{\phi}_{k_2} \circ\cdots\circ\wt{\mathbb{X}}\circ \phi_{k_m}\otimes \overline{\phi}_{k_m}[A\otimes \overline{B}]\otimes  \pi_E\left(\Pi_{i_1}^{(1)}\otimes \cdots \otimes \Pi_{i_m}^{(m)}\right).
\end{equation}
For $m=1$ we have indeed
\begin{align*}
	\left(\wt{\mathbb{X}}\otimes\mathrm{id}_E\right)\circ \mathbb{U}_\Theta [A\otimes \overline{B}\otimes \pi_E(\mathds{1}_E)]&=\sum_{jk} \wt{\mathbb{X}}\left[U_k^\dagger A U_j\otimes   \overline{U}_k^\dagger \overline{B} \overline{U}_j\right] \otimes \pi_E\left(\Pi_{k}^{(1)} \Pi_{j}^{(1)} \right)
	\\&=\sum_{k} \left(\wt{\mathbb{X}}\circ\phi_k  \otimes   \overline{\phi}_k[A\otimes \overline{B}] \right)\otimes \pi_E\left(\Pi_{k}^{(1)} \right).
\end{align*}
Suppose  that~\eqref{supportope} holds for $m\le n$ and evaluate $\mathbb{Y}_{n+1}$:
\begin{align*}
	\mathbb{Y}_{n+1}&=\left(\left(\wt{\mathbb{X}}\otimes\mathrm{id}_E\right)\circ \mathbb{U}_\Theta \right)^{n+1}[A\otimes \overline{B} \otimes \pi_E(\mathds{1}_E) ]\\
	&=\left(\left(\wt{\mathbb{X}}\otimes\mathrm{id}_E\right)\circ \mathbb{U}_\Theta \right)\left[ \sum_{k_2\cdots k_{n+1}} \wt{\mathbb{X}}\circ \phi_{k_2}\otimes \overline{\phi}_{k_2}\circ\cdots\circ\wt{\mathbb{X}}\circ \phi_{k_{n+1}}\otimes \overline{\phi}_{k_{n+1}} [A\otimes\overline{B}]\otimes  \pi_E\left(\mathds{1}_D^{(0)}\otimes \Pi_{k_{2}}^{(1)}\otimes \cdots \otimes \Pi_{k_{n+1}}^{(n)}\right) \right]\\
	& = \left(\wt{\mathbb{X}}\otimes\mathrm{id}_E\right)\bigg[\sum_{j_1 k_1 k_2\cdots\, k_{n+1}} U_{k_1}^\dag\otimes \overline{U}_{k_1}^\dag \left(\wt{\mathbb{X}}\circ \phi_{k_2}\otimes\overline{\phi}_{k_2}\circ\cdots\circ\wt{\mathbb{X}}\circ \phi_{k_{n+1}}\otimes\overline{\phi}_{k_{n+1}}  [A\otimes \overline{B}]\right)
	\\&\hskip+8cm
	 U_{j_1} \otimes  \overline{U}_{j_1}\otimes  \pi_E\left(\Pi_{j_1}^{(1)}\Pi_{k_1}^{(1)}\otimes  \Pi_{k_2}^{(2)}\otimes \cdots \otimes \Pi_{k_{n+1}}^{(n+1)}\right)\bigg]\\
	&=\sum_{k_1\cdots \,k_{n+1}}  \wt{\mathbb{X}}\circ \phi_{k_1}\otimes\overline{\phi}_{k_1}\circ\cdots\circ\wt{\mathbb{X}}\circ \phi_{k_{n+1}}\otimes\overline{\phi}_{k_{n+1}}[A\otimes \overline{B}]\otimes  \pi_E\left(\Pi_{k_1}^{(1)}\otimes  \Pi_{k_2}^{(2)}\otimes \cdots \otimes \Pi_{k_{n+1}}^{(n+1)}\right).
\end{align*}
Then, we now consider  expectation~\eqref{exp_n}:
\begin{align*}
	&\bra{\psiplus{d}\otimes\Omega_E} \pi_S\otimes\pi_E(X^\dag\otimes\mathds{1}_E) \left(\left(\wt{\mathbb{X}}\otimes\mathrm{id}_E\right)\circ\mathbb{U}_\Theta\right)^n [\pi_S\otimes\pi_E(A\otimes\mathds{1}_E)\pi_S'\otimes\pi_E(B\otimes\mathds{1}_E) ] \pi_S\otimes\pi_E\left(X\otimes\mathds{1}_E\right) \ket{\psiplus{d}\otimes\Omega_E}\\
	&=\sum_{k_1\cdots k_n}  p_{k_1\cdots k_n} \Tr\bigg(A\otimes \overline{B}\,\phi_{k_n}^\ddagger\otimes\overline{\phi}_{k_n}^\ddagger\circ \wt{\mathbb{X}}^\ddagger\circ \cdots 
	\circ\phi_{k_1}^\ddagger\otimes\overline{\phi}_{k_1}^{\ddagger}\circ\wt{\mathbb{X}}^\ddagger \left[ X\otimes \mathds{1}_d\pplus{d} X^\dagger\otimes \mathds{1}_d\right]\bigg)\\
	&=\Tr(A\otimes \overline{B}\,\Gamma_n^{\mathbb{X}}\left[X\otimes \mathds{1}_d\pplus{d} X^\dagger\otimes \mathds{1}_d\right]),
\end{align*}
where, in the second line, we moved by duality to Schrödinger picture, so to  finally deduce~\eqref{GammaXmodel}.
\color{black}
\end{widetext}

	\section{Two-qubit Pauli maps.}\label{rmk:Pauli1and2}
	A one qubit Pauli map has a spectral decomposition
	\begin{equation}
		\Lambda[X]=\frac{1}{2}\sum_{\alpha=0}^3 \lambda^{(\alpha)} \Tr(\sigma_\alpha X) \sigma_\alpha\,,
	\end{equation}
	and can be equivalently rewritten as
	\begin{equation}
		\Lambda[X]=\sum_{\alpha=0}^3 q^{(\alpha)} \sigma_\alpha X \sigma_\alpha\,, \qquad \sum_{\alpha=0}^3 {q}^{(\alpha)}=1\,.
	\end{equation}
	The spectrum $\lambda^{(\alpha)}$ and the coefficients $q^{(\alpha)}$ are related by a linear transformation
	\begin{equation}
		{q}^{(\alpha)}=\frac{1}{4} \sum_{\alpha \beta=0}^3 H_{\alpha\beta} \lambda^{(\beta)}\,, 
	\end{equation}
	given by an Hadamard matrix $H$, $H_{\alpha\beta}=\Tr(\sigma_\alpha\sigma_\beta\sigma_\alpha\sigma_\beta)$. Complete positivity of $\Lambda$ is equivalent to $q^{(\alpha)}\ge0$. Similarly, a two-qubit Pauli map is defined by its spectral decomposition
	\begin{equation}
		\Gamma[X]=\frac{1}{4}\sum_{\alpha\beta=0}^3 \lambda^{(\alpha,\beta)} \Tr(\sigma_\alpha\otimes\sigma_\beta X) \sigma_\alpha\otimes\sigma_\beta\,,
	\end{equation} 
	and can be recast in the form
	\begin{equation}\label{kraus_2pauli}
		\Gamma[X]=\sum_{\alpha\beta=0}^3 q^{(\alpha,\beta)}\sigma_\alpha\otimes\sigma_\beta X \sigma_\alpha\otimes\sigma_\beta \,, \quad \sum_{\alpha\beta=0}^3 q^{(\alpha,\beta)}=1\,,
	\end{equation}  
	where the matrices describing the spectrum $G:=[\gamma^{(\alpha,\beta)}]$, respectively, the coefficients $Q:=[q^{(\alpha,\beta)}]$ are similar and related through the Hadamard transformation,
	\begin{equation}\label{doubleHadamard}
		 Q=\frac{H\,G\,H}{16}\,.
	\end{equation} 
	Suppose now that the spectrum has the special form
	\begin{equation}
		G
		=\begin{pmatrix*}[c]
			1 & \lambda &  \lambda &  \lambda^{(3)} \\
			\lambda & 1 & \lambda^{(3)} &  \lambda \\
			\lambda & \lambda^{(3)} & 1 & \lambda \\
			\lambda^{(3)} & \lambda &\lambda & 1 \\
		\end{pmatrix*},
	\end{equation}
	Then, using~\eqref{doubleHadamard}, the matrix coefficients is diagonal,
	\begin{equation*}
		Q=\frac{1}{4}\begin{pmatrix*}[c]
			1+\lambda^{(3)}+2\lambda & 0 &  0 &  0\\
			0 & 1-\lambda^{(3)} & 0 & 0\\
			0 & 0	&1-\lambda^{(3)}  & 0 \\
			0 & 0 &0 & 	1+\lambda^{(3)}-2\lambda \\
		\end{pmatrix*},
	\end{equation*}
	whence~\eqref{kraus_2pauli} has the form
	\begin{equation*}
		\Gamma[X]=\sum_{\alpha=0}^3 q^{(\alpha,\alpha)}\sigma_\alpha\otimes\sigma_\alpha X \sigma_\alpha\otimes\sigma_\alpha\,.
	\end{equation*}  

\end{appendices}

\hfill

\bibliographystyle{apsrev4-2}


%

\end{document}